%% file: main.tex
\documentclass[11pt]{article}
\usepackage{graphicx} 
\usepackage{verbatim}
\usepackage[a4paper, left=1in, right=1in, top=1in, bottom=1in]{geometry}
\usepackage{appendix}
\usepackage{xcolor}
\usepackage{braket}

\definecolor{daxcolor}{rgb}{0.7, 0.1, 0.5}

\usepackage{amsmath, amssymb, amsthm}
\usepackage[
  colorlinks=true,
  breaklinks=true,
  pdfusetitle=true,
  hypertexnames=false,
  allcolors=blue]{hyperref}
\usepackage{comment}
\usepackage{mathtools}
\usepackage[giveninits=true,maxbibnames=99,style=alphabetic,maxalphanames=4,minalphanames=3,backend=biber]{biblatex}
\usepackage{algorithm}
\usepackage[noend]{algpseudocode}
\usepackage{float}
\usepackage[section]{placeins}
\usepackage{thm-restate}

\usepackage{accents}

\AtEveryBibitem{\clearlist{language}}
\renewbibmacro*{doi+eprint+url}{%
  \iftoggle{bbx:doi}
    {\printfield{doi}}
    {}%
  \newunit\newblock
  \ifboolexpr{togl {bbx:eprint} and test {\iffieldundef{doi}}}
    {\usebibmacro{eprint}}
    {}%
  \newunit\newblock
  \ifboolexpr{togl {bbx:url} and test {\iffieldundef{doi}}  and test {\iffieldundef{eprint}}}
    {\usebibmacro{url+urldate}}
    {}}

\usepackage[T1]{fontenc}
\usepackage[smallerops,vvarbb]{newtx}
\usepackage[narrow,varqu,varl,scaled=0.95]{zi4}
\usepackage{microtype}

\input{commands}

\usepackage[capitalise]{cleveref}
\newtheorem{theorem}{Theorem}[section]
\newtheorem{definition}[theorem]{Definition}
\newtheorem{proposition}[theorem]{Proposition}

\newtheorem{corollary}[theorem]{Corollary}

\newtheorem{lemma}[theorem]{Lemma}
\newtheorem{remark}[theorem]{Remark}
\newtheorem{setup}{Setup}
\usepackage{enumitem}
\setlist[itemize]{itemsep=1pt, topsep=2pt}
\setlist[enumerate]{itemsep=1pt, topsep=2pt}

\newcommand{\FermObsCon}{\cfont{FermObsCon}}
\newcommand{\cfont}{\mathsf}

\newcommand{\CSV}{\cfont{CSV}}

\newcommand{\p}{\cfont{P}}
\newcommand{\PP}{\cfont{PP}}
\newcommand{\coNP}{\cfont{coNP}}
\newcommand{\NEXP}{\cfont{NEXP}}
\newcommand{\BQP}{\cfont{BQP}}
\newcommand{\QMA}{\cfont{QMA}}
\newcommand{\QCPH}{\cfont{QCPH}}
\newcommand{\PH}{\cfont{PH}}
\newcommand{\QPH}{\cfont{QPH}}
\newcommand{\pureQPH}{\cfont{PureQPH}}
\newcommand{\QCMA}{\cfont{QCMA}}

\newcommand{\BPP}{\cfont{BPP}}
\newcommand{\LH}{\cfont{LH}}
\newcommand{\XZ}{\cfont{XZ}}
\newcommand{\CLDM}{\cfont{CLDM}}

\newcommand{\iDXZLH}{1D\mhyphen\XZ\mhyphen\LH}

\newcommand{\qc}{\cfont{qc\text{-}\Sigma_2}}
\newcommand{\lmin}{\lambda_{\textup{min}}}
\newcommand{\ObsCon}{\cfont{ObsCon}}

\newcommand{\pauliObsCon}{\cfont{AllPauli}\ObsCon_{\exp}}

\newcommand{\E}{\mathbb{E}}

\newcommand{\C}{\mathbb{C}}

\DeclarePairedDelimiter\abs{\lvert}{\rvert}%

\newcommand{\blnk}{ % an empty spot on the right side of the chain
\kern-1bp
\setlength{\unitlength}{10bp}
\begin{picture}(1,1)
\put(0.15,0.05){$\bigcirc$}
\end{picture}
\kern+3bp
}
\usepackage{tikz}

\newcommand{\turn}{ % a turn symbol, used in AGIK'09
\setlength{\unitlength}{10bp}
\kern-1bp
\begin{picture}(1,1)
\put(0.27,0.05){$\circlearrowleft$}
\put(0.15,0.05){$\bigcirc$}
\end{picture}
\kern+3bp
}
\newcommand{\insi}{ % a spot inside the qubit sequence, used in d=8
\setlength{\unitlength}{10bp}
\kern-1bp
\begin{picture}(1,1)
\put(0.4,0.05){$\circ$}
\put(0.15,0.05){$\bigcirc$}
\end{picture}
\kern+3bp
}
\newcommand{\dead}{ % a site that is done, left end of the chain
\kern-1bp
\setlength{\unitlength}{10bp}
\begin{picture}(1,1)
\put(0.27,0.05){$\times$}
\put(0.15,0.05){$\bigcirc$}
\end{picture}
\kern+3bp
}
\newcommand{\qubit}{ % a qubit holding particle, used in d=8
\setlength{\unitlength}{10bp}
\kern+1bp
\begin{picture}(1,1)
\linethickness{1bp}
\put(0,-0.2){\line(0,1){1}}
\put(0,-0.2){\line(1,0){1}}
\put(1,0.8){\line(0,-1){1}}
\put(1,0.8){\line(-1,0){1}}
\end{picture}
\kern+1bp
}
\newcommand{\gate}{ % a qubit holding particle, ACTIVE (do gate)
\setlength{\unitlength}{10bp}
\kern+1bp
\begin{picture}(1,1)
\put(0.1,0){$\blacktriangleright$}
\linethickness{1bp}
\put(0,-0.2){\line(0,1){1}}
\put(0,-0.2){\line(1,0){1}}
\put(1,0.8){\line(0,-1){1}}
\put(1,0.8){\line(-1,0){1}}
\end{picture}
\kern+1bp
}
\newcommand{\rmove}{ % a qubit holding particle, ACTIVE (move right)
\setlength{\unitlength}{10bp}
\kern+1bp
\begin{picture}(1,1)
\put(0.1,0){$\vartriangleright$}
\linethickness{1bp}
\put(0,-0.2){\line(0,1){1}}
\put(0,-0.2){\line(1,0){1}}
\put(1,0.8){\line(0,-1){1}}
\put(1,0.8){\line(-1,0){1}}
\end{picture}
\kern+1bp
}
\newcommand{\abet}{ % a beta qubit
\kern+1bp
\setlength{\unitlength}{10bp}
\begin{picture}(1,1)
\put(0.15,-0.05){$\beta$}
\linethickness{1bp}
\put(0,-0.2){\line(0,1){1}}
\put(0,-0.2){\line(1,0){1}}
\put(1,0.8){\line(0,-1){1}}
\put(1,0.8){\line(-1,0){1}}
\end{picture}
\kern+1bp
}
\newcommand{\alpa}{ % an alpha qubit
\kern+1bp
\setlength{\unitlength}{10bp}
\begin{picture}(1,1)
\put(0.15,0.05){$\alpha$}
\linethickness{1bp}
\put(0,-0.2){\line(0,1){1}}
\put(0,-0.2){\line(1,0){1}}
\put(1,0.8){\line(0,-1){1}}
\put(1,0.8){\line(-1,0){1}}
\end{picture}
\kern+1bp
}
\newcommand{\lqubit}{ % an L qubit
\kern+1bp
\setlength{\unitlength}{10bp}
\begin{picture}(1,1)
\put(0.15,0.05){$\scriptstyle{L}$}
\linethickness{1bp}
\put(0,-0.2){\line(0,1){1}}
\put(0,-0.2){\line(1,0){1}}
\put(1,0.8){\line(0,-1){1}}
\put(1,0.8){\line(-1,0){1}}
\end{picture}
\kern+1bp
}
\newcommand{\rqubit}{ % an R qubit
\kern+1bp
\setlength{\unitlength}{10bp}
\begin{picture}(1,1)
\put(0.15,0.05){$\scriptstyle{R}$}
\linethickness{1bp}
\put(0,-0.2){\line(0,1){1}}
\put(0,-0.2){\line(1,0){1}}
\put(1,0.8){\line(0,-1){1}}
\put(1,0.8){\line(-1,0){1}}
\end{picture}
\kern+1bp
}

\title{Verification Complexity and Extension of Classical Shadows}

\author{Georgios Karaiskos\footnotemark[1]\and Asad Raza\footnotemark[2]\and Dorian Rudolph\footnotemark[3]\and Dax Enshan Koh\footnotemark[4]\and Sevag Gharibian\footnotemark[3]}

\date{}

\begin{document}
\begingroup
\renewcommand\thefootnote{\fnsymbol{footnote}}
\maketitle
\footnotetext[1]{Department of Computer Science, Paderborn University, Germany. Email: georgios.karaiskos@upb.de.}
\footnotetext[2]{Dahlem Center for Complex Quantum Systems,
Freie Universit\"at Berlin, 14195 Berlin, Germany. Email: asad.raza@fu-berlin.de.}
\footnotetext[3]{Department of Computer Science and Institute for Photonic Quantum Systems (PhoQS), Paderborn University, Germany. Email: \{dorian.rudolph, sevag.gharibian\}@upb.de.}
\footnotetext[4]{Engineering Cluster, Singapore Institute of Technology, 1 Punggol Coast Road, Singapore 828608, Republic of Singapore; Quantum Innovation Centre (Q.InC), Agency for Science, Technology and Research (A*STAR), 4 Fusionopolis Way, Kinesis \#05-01, Singapore 138635, Republic of Singapore; and Science, Mathematics and Technology Cluster, Singapore University of Technology and Design, 8 Somapah Road, Singapore 487372, Republic of Singapore. Email: dax.koh@singaporetech.edu.sg.}
\endgroup

\begin{abstract}

Classical shadows are an influential framework for compressing copies of a given quantum state $\rho$ into classical data $S$, enabling many properties of $\rho$ to be predicted from relatively few copies. In this work, we study two natural questions involving shadows: (1) Given $S$, when can one efficiently verify that $S$ came from a genuine $n$-qubit state? This is called the Classical Shadow Validity ($\CSV$) problem, introduced by Karaiskos, Rudolph, Meyer, Eisert, and Gharibian [ICALP 2026]. (2) Given $S$ that allows one to capture $2$-local properties of $\rho$, can one fake or spoof a shadow $S'$ which predicts $3$-local properties of some state? For (1), we show $\CSV$ is efficiently solvable for permutation-invariant shadows, $\QMA$-hard for real, fermionic, and bosonic shadows, and both $\coNP$-hard and $\QMA$-hard when the observable family consists of all $n$-qubit Pauli strings. A result of independent interest along the way is a new upper bound $\qc\subseteq\p^{\PP}$, where $\qc$ is a quantum analogue of the second level of the polynomial hierarchy in which the first proof is quantum. For (2), we show intractability: Given the $2$-local marginals $S$ of a quantum state $\rho$, estimating the $3$-local marginals of $\rho$ is intractable unless $\QCMA\subseteq\BPP$, even if the state $\rho$ is the unique state consistent with $S$.
\end{abstract}

\tableofcontents

\section{Introduction}\label{scn:introduction}
Let $O=\set{O_i}_{i=1}^m$ denote a set of $n$-qubit observables, i.e., operators on a Hilbert space of dimension $D:=2^n$. The influential framework of classical shadows~\cite{aaronsonShadowTomographyQuantum2018,brandaoQuantumSDPSolvers2019,huangPredictingManyProperties2020} shows that, remarkably, for some sets $O$, for any $n$-qubit state $\rho$, only $\polylog(m, D)$ physical copies of $\rho$ are needed to classically estimate $\trace(\rho O_i)$ for all $i\in[m]$ with high probability. At a high level, the main steps of this framework can be formalized as (\Cref{def:classical-shadow}): (1) Measurements ``compress'' physical copies of $\rho$ into a multiset $S=\set{s_i}_{i=1}^N$ of $\poly(n)$-bit strings, which we call the \emph{shadow}, and (2) a classical \emph{recovery} algorithm $A$ takes as input the shadow $S$ and an index $i\in[m]$, and outputs an estimate of $\trace(\rho O_i)$. A canonical set $O$ for which this is possible~\cite{huangPredictingManyProperties2020} is the set of all $k$-local observables, for which the naive measurement protocol (i.e., measure a copy of $\rho$ for each $O_i$) would require $O(n^k)$ copies, exponentially more than what classical shadows achieve.

Of course, once we have $S$ in hand, a natural question arises: \emph{Can we extract other information about $\rho$ from $S$, aside from expectations $\trace(\rho O_i)$?} For example, can we verify whether $S$ is actually consistent with some quantum state? This is relevant, e.g., in the settings of noisy measurements and quantum cloud computing. In the former, if measurement noise is not properly accounted for by the recovery algorithm $A$, the resulting estimates can be biased~\cite{chen20221robust,koh2022classical} and, if sufficiently corrupted, may fail to be jointly consistent with \emph{any} state, let alone with $\rho$. For the latter, an untrusted quantum cloud may claim to have measured $\rho$ to produce $S$ for a client, but how can the client be sure\footnote{One can instead use, e.g., the $\BQP$ verification framework of Mahadev~\cite{mahadevClassicalVerificationQuantum2018}. This, however, requires interaction with the cloud, whereas here we are interested in the natural setting where the only input for the client is the shadow $S$ itself.} given just $S$? Another natural question is whether more mileage can be gained by classically postprocessing $S$; for example, if $S$ allows us to predict all $2$-body marginals of $\rho$, can we then extend this to allow predicting  $3$-body marginals of some $\rho'$ (where we require the $2$-body marginals of $\rho$ and $\rho'$ to coincide)? This is relevant from the cryptographic perspective, where an adversary is attempting to build on a simpler shadow to ``fake'' or ``spoof'' a more complicated one (since $\rho'$ need not equal $\rho$).

\vspace{-2mm}
\paragraph{Task 1: Checking shadow validity.} The first task is formalized by \emph{classical shadow validity}
($\CSV$, \Cref{def:CSV})~\cite{KRMEG25}. Given a shadow $S$, observables
$O=\{O_i\}$, and an efficient recovery algorithm $A$, $\CSV$ asks whether there exists a state
$\rho$ satisfying
\begin{equation}
    \bigl|\Tr(O_i\rho)-A(S,i)\bigr|\le\alpha
    \qquad\text{for all }i\in[m],
\end{equation}
or whether every $\rho$ violates some prediction $i$ by at least $\beta$, for $\abs{\beta-\alpha}\geq 1/\poly(n)$. That this problem can be $\QMA$-hard is not at all surprising --- this is because it includes as a special case the well-known $\QMA$-complete Consistency of Local Density Matrices problem (CONSISTENCY) of Liu~\cite{liuConsistencyLocalDensity2006}. (In CONSISTENCY, the input is a set of classical descriptions of $k$-local marginals $S:=\set{\rho_i}$, and the output is whether $S$ is consistent with some $n$-qubit global state $\rho$. Such a set $S$ can also be viewed as an example of a shadow.) The \emph{challenge} is to understand how this complexity varies with the specific structure of a given shadow protocol, since the set $S$ may look nothing like a set of local marginals. 

Indeed, the complexity of this problem can be subtle, ranging from classically solvable to harder than $\QMA$! Specifically, \cite{KRMEG25} showed: 
\begin{enumerate} 
    \item The global Clifford shadow protocol of \cite{huangPredictingManyProperties2020}, which allows one to efficiently predict expectations against observables $O$ of polynomial Frobenius norm, is unlikely to even be $\BQP$-hard, in that it is in $\BPP$ under now-standard sampling assumptions (as in dequantization works~\cite{tangQuantuminspiredClassicalAlgorithm2019}).
    \item The local Clifford shadow protocol of \cite{huangPredictingManyProperties2020} (as well as its high-dimension generalization~\cite{maoQuditShadowEstimation2025}), which can efficiently predict $k$-local observables, is $\QMA$-complete, even for $6$-local observables on a spatially sparse hypergraph. Note that these shadows are highly non-local in structure, making their $\QMA$-hardness proof challenging.

   \item The ability to predict expectation values for \emph{exponentially} many $O_i$ is likely harder than $\QMA$, being $\qc$-complete. Here, $\qc$ (\Cref{def:qc-cs}) is a quantum generalization of the second level of the Polynomial-Time Hierarchy, where the first proof is quantum and the second classical. This result was inspired by the ``triply efficient'' shadow protocol of King, Gosset, Kothari, Babbush~\cite{kingTriplyEfficientShadow2025}, which predicts expectations for the observable set $O$ of all $n$-qubit Pauli strings\footnote{A Pauli string is a tensor product of $n$ single-qubit Pauli operators, e.g. $X\otimes Y\otimes X\otimes \cdots\otimes Z$.} with polynomial sample complexity. An important shortcoming of~\cite{KRMEG25} is that its $\qc$-hardness proof is not for $O$ the set of $n$-qubit Pauli strings, but for an \emph{arbitrary} exponential-size set of observables $O$. 
\end{enumerate}
Given the rich breadth of complexities obtainable~\cite{KRMEG25} already with a small set of shadow protocols, the natural question arises: \emph{What about $\CSV$ for other shadow protocols?} In this work, we will consider a broad array of such protocols, including for permutation-invariant observables~\cite{SL24}, real classical shadows~\cite{WMLC25}, matchgate fermionic shadows~\cite{WHLB23}, and continuous variable/bosonic shadows ~\cite{BDLR24}. We also will make progress on the complexity of the King-Gosset-Kothari-Babbush scheme \cite{kingTriplyEfficientShadow2025}, which as a bonus will include showing a new upper bound $\qc\subseteq\p^{\PP}$.

\vspace{-2mm}

\paragraph{Task 2: ``Faking'' or ``spoofing'' shadows.} Finally, we turn to the question of extracting more mileage from a shadow $S$ via classical postprocessing. Recall the arguably simplest version of this problem sketched earlier --- as input, one is given shadow $S$ which is just the set of $2$-local marginals of some state $\rho$ (this need not be verified, it is promised to be a valid shadow). The question is whether there exists an efficient classical algorithm $C$ taking $S$ as input and extending it by \emph{just a single qubit}, i.e., meaning $C$ can predict $3$-body marginals within constant trace distance. Let us make the problem yet easier --- we are additionally promised there is a \emph{unique} extension of $S$ to $3$-local marginals. Can this task now be efficiently solved?

\subsection{Our results}\label{sscn:our-results}
We begin with our $\CSV$ results, and then discuss our ``faking shadows'' result. For clarity, our definition of classical shadow (\Cref{def:classical-shadow}, as in \cite{KRMEG25}) makes standard assumptions from the shadows literature: succinct classical access to each observable $O_i$, as well as the ability to efficiently quantumly implement a measurement corresponding to $O_i$ for any input quantum state $\rho$. Background on each of the classical shadow protocols we consider is given in \Cref{sscn:protocols}. We remark that for some of our $\CSV$ proofs, it will be simpler to work with a reformulation of $\CSV$, denoted Observable Consistency ($\ObsCon$, \Cref{def:obscon})~\cite{KRMEG25}. In $\ObsCon$, the input is a set of observables $\set{O_i}$ and target expectations $y_i$, and the question is if there exists a global state $\rho$ such that for all $i$, $\trace(\rho O_i)\approx y_i$. 

\vspace{-2mm}
\paragraph{1. Efficiently decidable: Permutation-invariant $\CSV$.} We begin with the permutation-invariant classical shadow scheme of Sauvage and Larocca~\cite{SL24}, in which each observable $O_i$ satisfies $U_\pi O_i U_\pi^\dagger = O_i$ for all permutations $\pi\in S_n$. 

\begin{theorem}[$\CSV$ for permutation-invariant shadows (informal; see \Cref{thm:ObsConPIinP} and \Cref{rem:pi-general-representation})]\label{thm:PIinformal}
    $\CSV$ for permutation-invariant shadows is in $\p$ if the observables are specified in terms of their Schur-Weyl blocks, and otherwise is in $\BQP$.
\end{theorem}
\noindent Thus, we find a second shadow protocol for which validity can be efficiently checked, with the global Clifford protocol of~\cite{huangPredictingManyProperties2020} being the first such known protocol~\cite{KRMEG25}. We remark, however, that the latter required sampling access to the observables $O_i$, which is not required for \Cref{thm:PIinformal}. Of course, obtaining the Schur-Weyl representation of the $O_i$ is itself also a non-trivial assumption; in any case, efficient \emph{quantum} validity checking is always possible, even without the Schur-Weyl representation in hand.

\vspace{-2mm}
\paragraph{2. $\QMA$-hard: $\CSV$ for real, fermionic, and bosonic shadows.}

We next show that $\CSV$ for the following specific shadow protocols is $\QMA$-hard. 
\begin{itemize}
    \item \emph{Real shadows.} We first consider the real shadow protocol of West, Mele, Larocca, and Cerezo~\cite{WMLC25} (\Cref{par:real}). We show that $\CSV_{\mathrm{real}}$ (\Cref{def:csv-real}) is $\QMA$-hard
    (\Cref{lem:csv-real-hard}), even though each qubit is measured only
    in the $X$ or $Z$ basis. The observables $O$ have constant
    locality, norm at most one, and Pauli expansions containing only
    $I$, $X$, and $Z$. As an intermediate result, we prove
    $\QMA$-hardness of a one-dimensional $XZ$-only Local Hamiltonian
    problem on constant-dimensional sites (\Cref{lem:1d-xz-lh}).

    \item \emph{Fermionic shadows.} We next consider the fermionic shadow protocol of Wan, Huggins, Lee, and Babbush~\cite{WHLB23} (\Cref{par:fermionic-conventions}). We show that Matchgate fermionic shadow validity (\Cref{def:CSV_MFS})
    in a specified particle-number sector,
    $\CSV_{\mathsf{MFS}}^{(N)}$, is $\QMA$-hard
    (\Cref{lem:fermobscon-hard,thm:fermobscon4-to-csvmfs}). The particle
    number is part of the input, and the reduction uses observables
    supported on at most $5$ modes. Although the measurement
    transformations are fermionic Gaussian, the states whose
    consistency is being tested are arbitrary states in the specified
    sector.

    \item \emph{Bosonic shadows.} Next, we study the bosonic shadow protocols (both heterodyne and homodyne) of Becker, Datta, Lami, and Rouze~\cite{BDLR24} (\Cref{par:bosonic-conventions}). We first show that regularized heterodyne shadow
    validity $\CSV_{\mathrm{HetCV}}^{\ell,B,M}$ is $\QMA$-hard already
    for $M=1$, $\ell=O(\log n)$, and $B=2^\ell$
    (\Cref{thm:Bosonic_shadows_hard}). Here, $M$ is the Fock space cutoff used by
    the recovery algorithm (\Cref{par:bosonic-conventions}), while candidate states $\rho$ range over the
    full bosonic Hilbert space subject to moment bounds (\Cref{def:CSV_het})
    \[
        \Tr\bigl((I+\widehat N_j)^\ell\rho\bigr)\le B
        \qquad\text{for every mode }j,
    \]
    i.e., no Fock cutoff is assumed for $\rho$. The observables are constant-local, constant-degree polynomials in position and momentum operators, which additionally
    satisfy the weighted norm bounds of ~\Cref{eqn:norm}. We then show an analogous $\QMA$-hardness result for $\CSV^{\ell,B,M}_{\mathrm{HomCV}}$, i.e., for homodyne CV classical shadows.
\end{itemize}

\vspace{-2mm}

\paragraph{3. Towards $\qc$-completeness: $n$-qubit Pauli shadows.} Next, we consider the setting of King, Gosset, Kothari, and Babbush~\cite{kingTriplyEfficientShadow2025} when the observable set $O$ is all $n$-qubit Pauli strings. Let $\pauliObsCon$ (\Cref{def:allpauliobscon}) denote the corresponding $\ObsCon$ problem (recall $\ObsCon$ is an equivalent formulation of $\CSV$). By definition, $\pauliObsCon\in \qc$, since the $\forall$-quantified second/classical proof of $\qc$ can check all \emph{exponentially} many $n$-qubit Pauli observables. This exponentiality also makes containment in $\QMA$ unlikely; indeed, one expects $\qc$-completeness of $\pauliObsCon$, a question left open by~\cite{KRMEG25}. We show:

\begin{theorem}[Complexity of $\pauliObsCon$ (informal)]\label{thm:pauliObsConinformal}
    $\pauliObsCon$ is $\QMA$-hard (\Cref{thm:QMAhard}), $\coNP$-hard (\Cref{thm:coNPhard}), and contained in $\p^{\PP}$ (\Cref{thm:qcinP^PP}).
\end{theorem}
\noindent Thus, while we do not achieve the desired $\qc$-completeness of $\pauliObsCon$, we obtain strong evidence in this direction --- $\QMA$-hardness suggests an existentially quantified quantum proof is necessary, and $\coNP$-hardness suggests a universally quantified classical proof is necessary. We remark that of these two results, $\QMA$-hardness is challenging to obtain.

Finally, the $\p^{\PP}$ containment of \Cref{thm:pauliObsConinformal} deserves a special mention --- it follows directly from:
\begin{theorem}[Upper bound on $\qc$ (informal; \Cref{thm:qcinP^PP})]
    $\qc\subseteq\p^{\PP}$.
\end{theorem}
\noindent This is the first ``Toda-like'' (i.e., $\p^{\PP}$) bound on any level of a quantum polynomial hierarchy ($\PH$) involving a \emph{quantum} proof. For comparison, quantum PH with all \emph{classical} proofs (denoted $\QCPH$~\cite{gharibianQuantumGeneralizationsPolynomial2018}) is contained in $\BPP^{\PP}$~\cite{miloschewsky2026promises}; improving this to $\p^{\PP}$ remains open. With quantum proofs involved, however, known upper bounds are substantially worse: quantum PH with mixed quantum proofs ($\QPH$~\cite{gharibianQuantumGeneralizationsPolynomial2018}) and pure quantum proofs ($\pureQPH$~\cite{grewalPureQuantumPolynomial2026}) are not even known to be contained in $\NEXP$ (!). In fact, the smallest upper bound known in the setting of quantum proofs is for the second level of $\QPH$, which is in PSPACE~\cite{gharibianQuantumGeneralizationsPolynomial2018,jainParallelApproximationNoninteractive2009}.

\vspace{-2mm}
\paragraph{4. Faking shadows and unique extensions.} Finally, we show that faking or spoofing a $3$-local shadow from a $2$-local one is impossible, assuming $\QCMA\not\subseteq\BPP$. Recall that here we work with a particularly simple type of shadow, namely the set of $2$-local marginals of an $n$-qubit state $\rho$.

\begin{restatable}{theorem}{thmRDM}(No-go for shadow extension (informal; \Cref{thm:recover-verify}))\label{thm:RDM}
    Let $D=(\rho_{\{i,j\}})_{1\le i<j\le n}$ denote a set of two-qubit marginals, with the promise of exactly one consistent $n$-qubit state $\rho$. If there exists a poly-time randomized classical algorithm which, for any such $D$, outputs any requested three-qubit marginal of $\rho$ (within constant trace distance), then $\QCMA \subseteq \BPP$.
\end{restatable}

\subsection{Techniques}\label{sscn:techniques}

We briefly review the techniques behind each of our results. 
\vspace{-2mm}

\paragraph{Permutation-invariant shadows: reduce the search by symmetry.}
A permutation-invariant (PI) observable has the same expectation
before and after the qubits are permuted. Thus, if a state matches
the given targets, averaging it over all permutations gives a PI
state that matches the same targets~\cite{SL24}. We can therefore
restrict our search to PI states. Schur--Weyl duality makes
this search tractable by representing these states through $O(n)$
positive semidefinite blocks, each of dimension at most $n+1$.
Positivity and normalization ensure that the blocks describe a
state, while each observable adds a linear condition specifying its
expectation. For polynomially many observables, these conditions
form a polynomial-size semidefinite program. This yields a classical
polynomial-time algorithm, assuming the observable blocks are supplied
classically explicitly. When the latter assumption does not hold, we 
quantumly estimate the blocks using probe states prepared by the quantum inverse Schur
transform~\cite{BCH05} and then solve the same program, obtaining
the $\BQP$ upper bound.

\paragraph{Real shadows: keep inconsistency visible to $X/Z$ measurements.}
Real shadows capture only the Pauli components built from $I$, $X$, and
$Z$, so simply discarding the remaining components of a $\QMA$-hard
consistency instance could hide its inconsistency. We address this by
constructing an $XZ$-only Hamiltonian whose energy certifies inconsistency. Specifically, we first fix a universal real gate set, and apply~\cite{HNN13} to obtain a one-dimensional, real history state construction. This, however, does not suffice --- the propagation Hamiltonian produced can still contain terms such as $Y\otimes Y$ (since this is a real matrix), whereas real shadows cannot ``see'' observable $Y$. The next step is thus to introduce a nearest-neighbor overlapping label encoding that allows the label changes required for propagation to be implemented using products of Pauli $X$ operators. Intuitively, each neighboring pair of labels is encoded in unary, with a single non-zero qubit identifying the pair. Changing a register's label requires ``switching off'' one qubit via $X$ and ``switching on'' another via $X$. It is precisely here that shared labels prevent partial updates: changing only part of such an
encoded transition would leave adjacent registers disagreeing.
Diagonal penalties using only $Z$ operators enforce single excitations and agreement
between shared labels. Taking these penalties sufficiently large
controls transitions taking us outside the encoding, giving an $XZ$-only
Hamiltonian with constant locality and an inverse-polynomial promise gap.

The local consistency targets are constructed by adapting the approach used to prove $\QMA$-completeness of $\CSV$~\cite{KRMEG25} for the local Clifford shadows scheme of ~\cite{huangPredictingManyProperties2020}. Namely, we use simulatable verification~\cite{BG22,KRMEG25}
to compute local targets for the real history-state construction
without an accepting witness, and transfer these targets through
our overlap encoding. These targets assign the Hamiltonian total energy zero on every input: they are approximately
jointly realizable on YES inputs, but contradict its positive ground
energy on NO inputs. Because the Hamiltonian is $XZ$-only, this
contradiction remains visible to real shadows. To express the
targets as shadow predictions, we then use dynamic programming to solve a 1D integer programming problem, which allows us to choose
local snapshot counts that agree on shared sites~\cite{KRMEG25}. Finally, these local snapshots are carefully stitched together to obtain \emph{global} samples with the
required approximate expectations for the real shadow scheme.

\paragraph{Fermionic shadows: isolate one constraint at a time.}
Our starting point is exact $N$-representability \cite{KR25}, which
asks whether a given two-particle density matrix can arise as the
reduced state of $N$ fermions. To build an instance of $\FermObsCon^{(N)}_{\mathcal O_{2\mathrm{RDM}}}$ (\Cref{def:FermObsCon}), we start with a set of Hermitian observables on at most four modes. To then ``isolate these constraints'', we tensor each observable
with an auxiliary tag observable and construct matchgate samples whose
averaged estimators reproduce the tagged targets. The tags let us
choose a separate sample family for each constraint, with zero
contribution to the other tagged constraints, hence the isolation. Indeed, a Majorana
monomial's estimator vanishes unless its support is a union of
matching edges, so pairing an unwanted tag Majorana outside the
support of its tagged observable forces that observable's estimator
to vanish term by term.

When these families are combined, each contributes only a fraction
of the samples and must therefore realize proportionally larger
targets. Adding further ``padding modes'' makes this possible by enlarging the
ball around the origin contained in the convex hull of local estimator
vectors. Constant-size linear programs then determine suitable
mixtures, whose weights we round to integer sample counts.
This process adds auxiliary modes and observables, so to recover the original $\FermObsCon^{(N)}_{\mathcal O_{2\mathrm{RDM}}}$ instance, we also constrain the
auxiliary modes to be almost fully occupied. Projecting onto their
occupied states changes expectations only slightly and makes every
tag observable equal to one; removing the auxiliaries therefore
recovers the source constraints in the $N$-particle sector, since
total particle number is fixed.

\paragraph{Bosonic shadows: encode qubits and control truncation errors.}
The reduction starts from the $\QMA$-hard restriction of qubit
$\ObsCon$ (\Cref{def:obscon}) to polynomially many Pauli observables,
each acting on a constant number of qubits. For heterodyne and homodyne shadows, we encode each qubit in a
mode's first two Fock levels and restrict the recovery procedures
of~\cite{BDLR24} to those levels. On this subspace, $a+a^\dagger$,
$-i(a-a^\dagger)$, and $I-2\widehat N$ compress to $X$, $Y$, and $Z$.
Shifting the target expectations in our restricted $\ObsCon$ instance to zero now lets us exploit the decay
of the recovery functions (for regularized heterodyne recovery, these functions are the Fourier transforms of smooth, compactly
supported functions, and for homodyne recovery they are pattern
functions). Constructing the measurement record in the reduction by including sufficiently large
outcomes for each mode thus gives recovered values close to zero.

This restriction applies to recovery, while candidate consistent
states may still occupy arbitrary Fock levels. Following \cite{BDLR24}'s use of
photon-number moment bounds, our reduction chooses
$\Tr((I+\widehat N_j)^\ell\rho)\le2^\ell$, with $\ell=O(\log n)$.
Thus the probability of finding two or more photons in a given mode
is at most $(2/3)^\ell$. Since the observables are unbounded,
weighted norm estimates are also needed to control the resulting
expectation errors and cross terms. We can then decode the first
two levels and replace the rest by $\lvert0\rangle$, preserving the
source constraints within the required tolerance.

\paragraph{All-Pauli consistency: suppress unknown correlations.}
For $\coNP$-hardness, we show a reduction from the $\coNP$-complete UNSAT problem. Specifically, we construct a succinct target function in
which each satisfying assignment of a Boolean formula specifies two
commuting Paulis with target $+1$ and their product with target $-1$.
The first two requirements (the $+1$ constraints) force the product's expectation to be
$+1$, giving a contradiction with the third $-1$ constraint; this persists even at constant error.
If the formula is unsatisfiable, no such contradiction is activated
and all targets are jointly realizable. 

The $\QMA$ reduction is more challenging, and faces a different obstacle: The reduction begins from the $\QMA$-complete $\text{FullPauli}\ObsCon_{2}$ (\Cref{def:FullPauliObsCon}) problem in which all 2-local expectations are specified, whereas the output must assign
expectation targets to every $n$-qubit Pauli, including non-local ones. We address this by attaching a classical key
register to each data qubit and applying a key-dependent Pauli mask.
The masks spread each correlation across many encoded expectations,
suppressing each expectation's contribution exponentially in the number of data
qubits involved. We can therefore compute targets for the known
local correlations and approximate contributions from the unknown
correlations (i.e., locality $\geq 3$) by zero. Retaining the keys also makes the source
information recoverable: reading them and undoing the masks recovers
the original constraints from the input $\text{FullPauli}\ObsCon_{2}$ instance with only polynomial error amplification.

\paragraph{The $\mathsf P^{\mathsf{PP}}$ simulation: find rejecting challenges.}
Our upper bound adapts Aaronson's postselected-learning
method~\cite{Aar05}. After amplifying the verifier, we take the
maximally mixed $q$-qubit proof as a candidate accepting existentially quantified proof. 
For the universally quantified classical proof, we repeatedly find a challenge
accepted with probability below $1/2$, then condition on acceptance.
Each update at least halves the probability of passing the entire
sequence; if this probability becomes zero, we reject. On YES
inputs, however, the honest proof's initial weight $2^{-q}$ and the
quantum union bound~\cite{Gao15} limit the number of updates to
$O(q)$, whereas on NO inputs every conditional proof state still
has a strongly rejecting challenge.

To find each challenge, we repeat the verifier on independent copies
of the current conditional state and ask whether at least three
quarters of the runs reject. This is likely for a strongly rejecting challenge, but
exponentially unlikely for any challenge accepted at least half the
time. With enough repetitions, one strongly rejecting challenge raises
the average probability of this event above the level possible when
every challenge is accepted at least half the time. We then choose
the challenge bit by bit, using $\mathsf{PP}$ queries to retain the
group with the larger average probability. The common conditioning
factor cancels in these comparisons, and a final query verifies
acceptance below $1/2$, giving $\qc\subseteq\mathsf P^{\mathsf{PP}}$.

\vspace{-2mm}
\paragraph{Faking shadows: recover a witness and simulate its verification.}
The main idea is to turn a classical algorithm for extending two-qubit
marginals to three-qubit marginals into a $\mathsf{BPP}$ simulation of a $\mathsf{QCMA}$ computation.
We first use Valiant--Vazirani type \cite{VV86} isolation to the set of perfectly accepting $\mathsf{QCMA}$ (since $\mathsf{QCMA}$ = $\mathsf{QCMA}_1$ \cite{JKNN12}) witnesses to obtain a unique and perfectly accepting
classical witness. We then invoke the extension algorithm twice, for two
different purposes. The first invocation recovers the witness itself
from carefully constructed $2$-qubit data, while the second uses the
recovered witness string to classically simulate whether the quantum verifier
would accept it. Thus the extension algorithm substitutes both for
access to the QCMA witness and for the subsequent quantum verification
of that witness.

First we recover the witness from $x$ alone. We construct a family of two-qubit marginals whose
global completion is unique when the verifier has a unique and perfectly accepting witness $w^\star$. Although these marginals can be
computed without knowing $w^\star$, selected three-qubit marginals
encode the individual logical bits of $w^\star$. We do this via `simulatable' Steane codes \cite{BG22, KR25} whose two-qubit reduced density matrices are maximally mixed, i.e., reveal no logical information about the encoded state, but 3 qubit marginals do. We thus use the sign of a $Z^{\otimes3}$
expectation on a 3-qubit marginal produced by the hypothetical extension algorithm to produce each $w_i^\star$. A classical extension
algorithm that reconstructs these three-qubit marginals from the
two-qubit data therefore allows a $\mathsf{BPP}$ machine to recover the
entire witness with high probability.

Recovering a classical string is not enough, since checking it still
requires evaluating the quantum verifier. Given a candidate $w$, we
therefore construct a second family of two-qubit marginals whose global
completion is unique for every $x,w$, whether or not $w$ is accepted.
A designated three-qubit correlation in this completion approximates
$1-2p_x(w)$, where $p_x(w)$ is the verifier's acceptance probability.
The same extension algorithm can therefore distinguish a perfectly
accepting witness from any witness accepted with probability at most
$1/3$. This second stage is also what makes the simulation sound on NO
instances, even if the recovery construction is then outside its
uniqueness promise and outputs an arbitrary candidate, the verification
construction remains uniquely completable and rejects every candidate.
After amplification, the two stages yield a $\mathsf{BPP}$ simulation
of the original $\mathsf{QCMA}$ computation (\Cref{thm:recover-verify}).

\subsection{Open questions}\label{sscn:open-questions}

Our work leaves a number of questions open. First, is $\pauliObsCon$ complete for $\qc$, or does the
fixed Pauli observable family admit a stronger upper bound than general succinct consistency? If $\qc$-completeness holds, one would obtain the first natural $\qc$-complete problem for a natural observable family. Second, we showed $\qc\subseteq\p^{\PP}$. Can containment in $\p^{\PP}$ still be shown if we increase the number of alternating quantum and classical proofs? Third, does heterodyne shadow validity remain $\QMA$-hard with a fixed
sufficiently high moment order $\ell$, a fixed bound $B$, and a fixed
recovery cutoff $M$, all independent of the number of modes? Our reduction uses
$\ell=O(\log n)$ to make leakage outside the two-level encoding
smaller than the inverse-polynomial consistency gap. A fixed-moment
result would require another way to constrain this leakage while
retaining the full Fock space as the witness space and keeping the
polynomial observables well defined. Fourth, are there other classical shadow protocols captured by complexity classes other than those used here ($\BQP$, $\QMA$, and $\qc$)? Fifth, can we come up with application(s) of the hardness of faking shadows problem? For example, can we use the classical hardness of constructing a desired 3-qubit marginal from the set of 2-qubit marginal data in a concrete cryptographic setting? Finally, and more broadly, are there other shadow protocols which admit efficient validity checking, or a reasonable notion of shadow extension?

\subsection*{Acknowledgements}\label{sscn:acknowledgements}
SG and DR acknowledge support from the Deutsche Forschungsgemeinschaft
(DFG), project 563388236 (Bridge-QS, SPP 2514). SG and GK acknowledge
support from DFG project 572703436 (QPUP). SG additionally acknowledges support from EU QuantERA/DFG project 583918116
(SDPCODE), and BMFTR (PhoQuant). DEK is supported by the National Research Foundation, Singapore through the National Quantum Office, hosted in A*STAR, under the Advanced Quantum Algorithms and Solutions (AQAS) Funding Initiative (S25Q9DA001). The Berlin wing acknowledges funding by the BMFTR (Hybrid++,
MuniQC-Atoms), the Munich Quantum Valley, Berlin Quantum, the 
Quantum Flagship (Millenion, PasQuans2), the European Research Council (DebuQC),
the Clusters of Excellence (MATH+, ML4Q), and the 
DFG (CRC 183, SPP 2514, and BoLaCo).

\subsection*{Disclosure of AI assistance}
The authors used generative AI tools to explore proof ideas,
assist in identifying and closing gaps in proofs, and improve
the exposition. The authors take full responsibility for the
correctness of all arguments and results and for the final
content of the manuscript.

\section{Preliminaries}
\paragraph{Notation.}
For a vector $\ket{v}$, we write
$\|v\|_2:=\sqrt{\braket{v|v}}$ for its Euclidean norm.
For operators, we use the trace, Hilbert--Schmidt,
and operator norms, respectively:
\[
\begin{aligned}
\|X\|_1&:=\operatorname{Tr}\sqrt{X^\dagger X},\\
\|X\|_2&:=\sqrt{\operatorname{Tr}(X^\dagger X)},\\
\|X\|_\infty&:=\sup_{\|v\|_2=1}\|X\ket{v}\|_2.
\end{aligned}
\]
The trace-norm distance between states $\rho$ and $\sigma$
is $\|\rho-\sigma\|_1$.
For a state $\rho$ on labeled subsystems, we write
$\rho_\Lambda:=\operatorname{Tr}_{\bar\Lambda}(\rho)$
for its reduced state on $\Lambda$, where $\bar\Lambda$
denotes the complementary set of subsystems.

\subsection{Consistency problems}
\begin{definition}[Classical shadow~\cite{KRMEG25}]\label{def:classical-shadow}
    A shadow on $n$ qubits is a $4$-tuple $(S,O,A, \chi)$, where 
    \begin{itemize}
        \item (Shadow) $S=\set{s_i}_{i=1}^N$ is a multiset of $\poly(n)$-bit strings, with $N=\poly(n)$,
        \item (Observables) $O=\set{O_i}_{i=1}^m$ is a set of $n$-qubit observables satisfying $\|O_i\|_{\infty}\le1$, where $1\leq m\leq 2^{p(n)}$ for polynomial $p$. Given index $i$, a $\poly(n)$-bit description of $O_i$ can be produced in $\poly(n)$-time. Moreover, there exists a $\poly(n)$-time quantum algorithm which, for any $O_i$ and any $n$-qubit state $\rho$, applies\footnote{Formally, we can efficiently measure in the eigenbasis of $O_i$, and return the eigenvalue corresponding to the measurement result.} measurement $O_i$ to $\rho$.
        \item (Recovery algorithm) $A$ is a $\poly(n)$-time classical algorithm which, given $S$ and $i\in[m]$, produces a real number $A(S,i)\in[-1,1]$ within $\chi$ bits of precision.
    \end{itemize}    
\end{definition}
We use the following conventions throughout.
\paragraph{Succinct access assumption~\cite{KRMEG25}.} When we say that we assume succinct access to a set $\{A_i^{(n)}\}_{i=1}^m$, with $n$ the natural size parameter of the instance (e.g., 
number of qubits for observables or precision parameter for a real value), we mean that given an index $i$, a $\poly(n)$-bit description of $A_i$ can be produced in $\poly(n)$-time.
\paragraph{Recovery notation.}
When convenient, we write $A(S,O)$ instead of $A(S,i)$ to make
the observable being estimated explicit, in particular for
rescaled or modified observables. This notation is used whenever
$O$ belongs to the observable class supported by the corresponding recovery procedure.
\paragraph{Range convention for recovery values~\cite{KRMEG25}.}
Although the physical expectation value of any observable $O$ satisfying $\|O\|_\infty\le1$ lies
in $[-1,1]$, the raw estimator produced by a shadow recovery procedure
need not lie in this interval. In our protocol-specific CSV problems, we
therefore impose as a promise that the final reported
recovery values lie in $[-1,1]$; an out-of-range reported value is treated as an invalid prediction.

 \begin{definition}[Classical Shadow Validity ($\CSV$)~\cite{KRMEG25}]\label{def:CSV}
Given classical shadow $(S,O,A,\chi)$, parameters $\alpha$ and $\beta$ satisfying $\beta-\alpha\geq 1/\poly(n)$, decide between the following two cases:
\begin{itemize}
    \item \textbf{Yes}: $\exists$  $n$-qubit state $\rho$ s.t.\ $\forall$ $i\in [m]$, 
    $\left| \Tr\left(O_i\rho\right)-A(S,i)\right|\leq \alpha$.
    \item \textbf{No}: $\forall$ n-qubit states $\rho$ $\exists$ some $i\in [m]$ s.t.\ $\left| \Tr \left(O_i\rho\right)-A(S,i)\right|\geq \beta$.
\end{itemize}
Since $A(S,i)\in[-1,1]$ and $\|O_i\|_{\infty}\le1$, it is natural to assume $0\le\alpha<\beta\le2$.
\end{definition}

\begin{definition}[Observable consistency ($\ObsCon$)~\cite{KRMEG25}]\label{def:obscon}
The input is a set of observables, as in \cref{def:classical-shadow}, along with their target expectation values $(O_i, y_i)_{i=1}^m$, for which we assume succinct access, and parameters $\alpha$ and $\beta$ satisfying $\beta-\alpha\geq 1/\poly(n)$. We further assume w.l.o.g. that $y_i\in [-1,1]$ and $0\le\alpha<\beta\le2$.  The output is to decide between the following cases:
\begin{itemize}
    \item \textbf{Yes}: $\exists$  $n$-qubit state $\rho$ such that $\forall$ $i\in [m]$, 
    $\left| \Tr\left(O_i\rho\right)-y_i\right|\leq \alpha$.
    \item \textbf{No}: $\forall$ $n$-qubit states $\rho$, $\exists$ $i\in [m]$ such that $\left| \Tr \left(O_i\rho\right)-y_i\right|\geq \beta$.
\end{itemize}
\end{definition}
We write $\ObsCon_{\mathrm{poly}}$ for the restriction to
polynomially many observables and $\ObsCon_{\exp}$ for
the version allowing exponentially many observables
under succinct access.
\subsection{Complexity classes}
\begin{definition}[$\QMA$]
A promise problem $A=(A_{\mathrm{yes}},A_{\mathrm{no}})$ is in $\QMA$ if
and only if there exist polynomials $p,q$ and a polynomial-time uniform
family of quantum circuits $\{Q_n\}$, where $Q_n$ takes as input a string
$x\in\Sigma^*$ with $|x|=n$, a quantum proof
$\ket{y}\in(\mathbb{C}^2)^{\otimes p(n)}$, and $q(n)$ ancilla qubits in
state $\ket{0}^{\otimes q(n)}$, such that:
\begin{itemize}
    \item (Completeness) If $x\in A_{\mathrm{yes}}$, then there exists a
    proof $\ket{y}\in(\mathbb{C}^2)^{\otimes p(n)}$ such that $Q_n$
    accepts $(x,\ket{y})$ with probability at least $2/3$.

    \item (Soundness) If $x\in A_{\mathrm{no}}$, then for all proofs
    $\ket{y}\in(\mathbb{C}^2)^{\otimes p(n)}$, $Q_n$ accepts
    $(x,\ket{y})$ with probability at most $1/3$.
\end{itemize}
\end{definition}

\begin{definition}[$\qc$]\label{def:qc-cs}
A promise problem $A=(A_{\mathrm{yes}},A_{\mathrm{no}})$ is in $\qc$ if there is a
polynomial-time generated quantum verifier $V_x$ which, on input $x\in\{0,1\}^n$,
receives a polynomial-size quantum proof $\rho$ and a polynomial-size classical proof
$y$, and satisfies:
\begin{itemize}
    \item \textbf{Completeness}: If $x \in A_{\text{yes}}$, then $\exists\; \rho$ such that $\forall \ket{y}$, 
    $\Pr[V_x(\rho,\ket{y})=1]\geq 2/3$. 
    \item \textbf{Soundness}: If $x \in A_{\text{no}}$, then $\forall \rho$, $\exists \ket{y}$ such that $\Pr[V_x(\rho,\ket{y})=1]\leq 1/3$.   
\end{itemize}  
\end{definition}
\begin{definition}[$\coNP$]
For every $L\subseteq\{0,1\}^{*}$, we say that
$L\in\coNP$ if there exist a polynomial
$p:\mathbb{N}\to\mathbb{N}$ and a polynomial-time Turing machine $M$
such that, for every $x\in\{0,1\}^{*}$,
\[
    x\in L
    \quad\Longleftrightarrow\quad
    \forall u\in\{0,1\}^{p(|x|)},\;
    M(x,u)=1.
\]
\end{definition}

\begin{definition}
    [$\mathsf{PP}$] A promise problem $A=(A_{\mathrm{yes}},A_{\mathrm{no}})$ is in
$\mathsf{PP}$ if and only if there exists a polynomial-time
probabilistic Turing machine $M$ that accepts every string
$x\in A_{\mathrm{yes}}$ with probability strictly greater than $1/2$,
and accepts every string $x\in A_{\mathrm{no}}$ with probability
at most $1/2$.
\end{definition}

\begin{definition}[$\mathsf{P}^{\mathsf{PP}}$]
$\mathsf{P}^{\mathsf{PP}}$ is the class of decision problems which can be decided by a polynomial-time deterministic Turing machine making $\poly(n)$ queries to an oracle for $\mathsf{PP}$. 
\end{definition}

\subsection{Shadow protocols}\label{sscn:protocols}
\subsubsection{Permutation-invariant shadows}\label{par:PI}
\paragraph{Schur--Weyl duality.}\label{par:schur-weyl}
Let $U_\pi$ denote the unitary that permutes the $n$ qubits
according to $\pi\in S_n$. An operator $O$ is permutation
invariant (PI) if $U_\pi O U_\pi^\dagger=O$ for every
$\pi\in S_n$. Schur--Weyl duality gives the orthogonal
decomposition
\[
(\mathbb C^2)^{\otimes n}
\cong
\bigoplus_{\lambda}
\left(\mathbb C^{d_\lambda}\otimes\mathbb C^{m_\lambda}\right),
\]
where $\lambda=(n-k,k)$ ranges over partitions of $n$
with at most two rows, with
$k=0,\ldots,\lfloor n/2\rfloor$, and
\[
d_\lambda=\binom{n}{k}-\binom{n}{k-1},
\qquad
m_\lambda=n-2k+1.
\]
Here $\binom{n}{-1}:=0$, and $d_\lambda$ is the dimension
of the corresponding irreducible representation of $S_n$.
We call each summand a sector; its dimension is
$d_\lambda m_\lambda$.\\
\\
In the corresponding Schur--Weyl basis, every PI
Hermitian operator has the form
\[
O=\bigoplus_\lambda I_{d_\lambda}\otimes O_\lambda,
\quad
O_\lambda\in\operatorname{Herm}(m_\lambda),
\]
where $I_{d_\lambda}\otimes O_\lambda$ consists
of $d_\lambda$ identical copies of the reduced block
$O_\lambda$. Thus, although $d_\lambda$ can be exponentially
large, there are only $\lfloor n/2\rfloor+1$ reduced blocks,
each of dimension $m_\lambda\le n+1$. The sector labels,
block sizes, and repetition counts depend only on $n$;
only the entries of the reduced blocks depend on $O$. Since
\begin{equation*}
\sum_\lambda m_\lambda^2=O(n^3),
\end{equation*}
listing each reduced block once, with $\poly(n)$ bits per entry, gives a polynomial-size
classical description of $O$.
\paragraph{Protocol.}
The shallow permutation-invariant shadow protocol
of \cite{SL24} samples a Haar-random single-qubit
unitary $W_\ell\in\mathrm{SU}(2)$, applies
$W_\ell^{\otimes n}$ to an unknown $n$-qubit state
$\rho$, and measures in the computational basis.
Only the Hamming weight $h_\ell$ of the measurement
outcome is retained, giving the measurement record
$s_\ell=(W_\ell,h_\ell)$. Define
\[
    \Pi_h
    :=
    \sum_{\substack{b\in\{0,1\}^n\\
                    \operatorname{wt}(b)=h}}
    \ket{b}\bra{b},
    \qquad
    R_{W,h}
    :=
    (W^\dagger)^{\otimes n}\Pi_hW^{\otimes n}.
\]
The associated measurement map is
\[
    \mathcal M_{\mathsf{PI}}(X)
    :=
    \E_W\sum_{h=0}^{n}
    \Tr(R_{W,h}X)\,R_{W,h}.
\]
Its image is the space of permutation-invariant
operators. The snapshot obtained from its inverse
satisfies
\[
    \widehat\rho_\ell
    :=
    \mathcal M_{\mathsf{PI}}^{+}
    (R_{W_\ell,h_\ell}),
    \qquad
    \E[\widehat\rho_\ell]
    =
    \frac1{n!}\sum_{\pi\in S_n}
    U_\pi\rho U_\pi^\dagger.
\]
Consequently, for every observable satisfying
$U_\pi O U_\pi^\dagger=O$ for all $\pi\in S_n$,
the estimator $\widehat o_\ell(O):=\Tr(O\widehat\rho_\ell)$
is unbiased for $\Tr(O\rho)$, even when $\rho$ itself
is not permutation invariant.
Given an observable family $\{O_i\}_{i=1}^m$ and
records $S=\{s_\ell\}_{\ell=1}^L$, median of means
applied to $\{\widehat o_\ell(O_i)\}_{\ell=1}^L$
yields the recovery value $A_{\mathsf{PI}}(S,i)$.

\subsubsection{Real shadows}\label{par:real}
\paragraph{Protocol.}
The local real-Clifford instantiation of the orthogonal
shadow protocol of \cite{WMLC25} is equivalent to
independently measuring each qubit of the unknown state in a uniformly
random $X$ or $Z$ basis. For an unknown
$n$-qubit state $\rho$, each round independently samples
$P_{\ell,j}\in\{X,Z\}$ uniformly for every qubit $j$, measures
in the selected bases, and records the outcomes
$b_{\ell,j}\in\{-1,+1\}$. Thus each record is
$s_\ell=((P_{\ell,j},b_{\ell,j}))_{j=1}^n$.
The measurement channel acts on Pauli strings as
\[
    \mathcal M_{\mathrm{real}}(P)
    =
    \begin{cases}
        2^{-\operatorname{wt}(P)}P,
        & P\in\{I,X,Z\}^{\otimes n},\\
        0,
        & P \text{ contains a }Y\text{ factor}.
    \end{cases}
\]
It is therefore invertible on the operator space
$\mathcal V_{XZ}:=
\operatorname{span}_{\mathbb R}\{I,X,Z\}^{\otimes n}$.
Writing $\ket{\psi_{\ell,j}}$ for the eigenstate of
$P_{\ell,j}$ with eigenvalue $b_{\ell,j}$, the corresponding
snapshot is
\[
    \widehat\rho_\ell
    =
    \bigotimes_{j=1}^{n}
    \left(
        2\ket{\psi_{\ell,j}}\bra{\psi_{\ell,j}}
        -\frac{I}{2}
    \right),
    \quad
    \E[\widehat\rho_\ell]=\Pi_{XZ}(\rho),
\]
where $\Pi_{XZ}$ is the Hilbert--Schmidt orthogonal
projection onto $\mathcal V_{XZ}$. Consequently, for
every $O\in\mathcal V_{XZ}$, the single-shot estimator
$\widehat o_\ell(O):=\Tr(O\widehat\rho_\ell)$ is unbiased
for $\Tr(O\rho)$. Given an observable family $\{O_i\}_{i=1}^m$ and $S=\{s_\ell\}_{\ell=1}^{L}$, the recovery value $A_{\mathrm{real}}(S,i)$ is obtained by applying
median of means to $\{\widehat o_\ell(O_i)\}_{\ell=1}^{L}$. For $m$ observables in $\mathcal V_{XZ}$, each supported on at most $k$ qubits and satisfying $\|O_i\|_\infty\le1$,
\[
    L=O\left(
        \frac{3^k\log(m/\delta)}{\epsilon^2}
    \right)
\]
samples suffice to estimate all their expectation values
to additive error $\epsilon$ with probability at least
$1-\delta$. For individual Pauli strings of weight at
most $k$, the factor $3^k$ improves to $2^k$.
\subsubsection{Matchgate fermionic shadows}\label{par:MFS}
\paragraph{Fermionic conventions.}\label{par:fermionic-conventions}
For $n$ fermionic modes, let $a_j$ and $a_j^\dagger$ denote
the annihilation and creation operators, respectively,
satisfying the canonical anticommutation relations.
We define the Majorana operators by
\[
\gamma_{2j-1}:=a_j+a_j^\dagger,
\quad
\gamma_{2j}:=-i(a_j-a_j^\dagger),
\quad j\in[n].
\]
These satisfy $\{\gamma_\mu,\gamma_\nu\}=2\delta_{\mu\nu}I$.
For $F=\{f_1<\cdots<f_{2r}\}\subseteq[2n]$, define
the Hermitian Majorana monomial
\[
G_F:=i^r\gamma_{f_1}\cdots\gamma_{f_{2r}},
\qquad G_\varnothing:=I.
\]
Principal submatrices indexed by $F$ use this increasing order.
In particular,
\[
i\gamma_{2j-1}\gamma_{2j}=2a_j^\dagger a_j-I,
\]
so the occupied state has eigenvalue $+1$.
The same modewise convention applies to all auxiliary modes.

\paragraph{Protocol.}
The matchgate shadow protocol of \cite{WHLB23} applies
random fermionic Gaussian unitaries before
occupation-number measurements. In its discrete version,
each round samples $Q_\ell$ uniformly from $B(2n)$,
the group of $2n\times2n$ signed permutation matrices,
applies the corresponding Clifford-matchgate unitary
$U_{Q_\ell}$ to an unknown $n$-mode state $\rho$,
and measures in the occupation-number basis.
The classical record is $s_\ell=(Q_\ell,b_\ell)$,
where $b_\ell\in\{0,1\}^n$. We fix the convention
\[
U_Q^\dagger\gamma_\mu U_Q
= \sum_{\nu=1}^{2n}Q_{\mu\nu}\gamma_\nu,
\quad \mu\in[2n].
\]
Let $\mathcal P_{2r}$ denote the Hilbert-Schmidt
orthogonal projection onto the span of products of
$2r$ distinct Majorana operators. The measurement
channel is
\[
    \mathcal M_{\mathsf{MFS}}
    =
    \sum_{r=0}^{n}
    \frac{\binom{n}{r}}{\binom{2n}{2r}}\,
    \mathcal P_{2r}.
\]
Its image is the even operator space
$\Gamma_{\mathrm{even}}$, and its pseudoinverse
defines the snapshot
\[
    \widehat\rho_\ell
    =
    \mathcal M_{\mathsf{MFS}}^{+}
    \left(
        U_{Q_\ell}^{\dagger}
        \ket{b_\ell}\bra{b_\ell}
        U_{Q_\ell}
    \right),
    \qquad
    \E[\widehat\rho_\ell]=\Pi_{\mathrm{even}}(\rho),
\]
where $\Pi_{\mathrm{even}}=\sum_{r=0}^{n}\mathcal P_{2r}$.
Thus, for every even observable $O$, the single-shot
estimator $\widehat o_\ell(O):=\Tr(O\widehat\rho_\ell)$
is unbiased for $\Tr(O\rho)$. These estimators admit
efficient classical evaluation for local even fermionic
observables, fermionic Gaussian density operators,
and projectors onto Slater determinants.
Given $S=\{s_\ell\}_{\ell=1}^{L}$, median of means yields
the recovery value $A_{\mathsf{MFS}}(S,O)$.
In particular, for $m$ Hermitian Majorana monomials
of degree $2k$,
\[
    L=O\!\left(
        \frac{\binom{2n}{2k}}{\binom{n}{k}}\,
        \frac{\log(m/\delta)}{\epsilon^2}
    \right)
\]
samples suffice to estimate all their expectation
values to additive error $\epsilon$ with probability
at least $1-\delta$.

\begin{remark}[Absorbing measurement outcomes]
\label{rem:mfs-zero-outcome}
For a state $\rho$, define its Majorana covariance
matrix by
\[
(C_\rho)_{\mu\nu}
:=-\frac{i}{2}\Tr\!\left(\rho[\gamma_\mu,\gamma_\nu]\right),
\]
and write $C_{\ket b}:=C_{\ket b\bra b}$.
For $b\in\{0,1\}^n$, let
\[
D_b:=\bigoplus_{j=1}^n
\begin{pmatrix}
1&0\\
0&(-1)^{b_j}
\end{pmatrix}
\in B(2n).
\]
Then $C_{\ket b}=D_b^{\mathsf T}C_{\ket{0^n}}D_b$,
and hence
\[
Q^{\mathsf T}C_{\ket b}Q
=
(D_bQ)^{\mathsf T}C_{\ket{0^n}}(D_bQ).
\]
Since fermionic Gaussian states are determined by
their covariance matrices, the samples $(Q,b)$ and
$(D_bQ,0^n)$ produce the same snapshot.
Thus, when constructing shadow samples, we may
assume that every outcome is $0^n$, absorbing
the original outcome into the signed permutation.
\end{remark}
\subsubsection{Continuous-variable shadows}\label{par:HetCV}

\paragraph{Bosonic conventions.}\label{par:bosonic-conventions}
For $n$ bosonic modes, let $a_j$ and $a_j^\dagger$ denote
the annihilation and creation operators, respectively, satisfying
\[
[a_j,a_k^\dagger]=\delta_{jk}I,
\qquad [a_j,a_k]=0.
\]
We define the quadratures and number operator by
\[
q_j:=\frac{a_j+a_j^\dagger}{\sqrt2},
\qquad
p_j:=\frac{a_j-a_j^\dagger}{i\sqrt2},
\qquad
\widehat N_j:=a_j^\dagger a_j.
\]
Thus $[q_j,p_k]=i\delta_{jk}I$.
Write
\[
R:=(q_1,p_1,\ldots,q_n,p_n)^{\mathsf T},
\quad
\Omega:=\bigoplus_{j=1}^{n}
\begin{pmatrix}0&1\\-1&0\end{pmatrix}.
\]
For $\xi\in\mathbb R^{2n}$, define the Weyl operator
$D(\xi):=\exp(-i\xi^{\mathsf T}\Omega R)$.
The characteristic function of a trace-class operator $X$ is
\[
\chi_X(\xi):=\operatorname{Tr}\left(XD(\xi)\right).
\]
\paragraph{Protocol.}
We consider the local heterodyne and homodyne shadow protocols
of \cite{BDLR24} for an unknown $n$-mode bosonic state $\rho$.
We denote the chosen protocol by
$\nu\in\{\mathrm{het},\mathrm{hom}\}$.
Each round measures a fresh copy of $\rho$ and produces a
classical record $s_t$, giving $S=\{s_t\}_{t=1}^{L}$.\\
\\
Consider observable families $O=\{O_i\}_{i=1}^m$, with
$m=\poly(n)$, whose members have polynomial-bit descriptions.
There is a fixed $k=O(1)$ such that each $O_i$ is Hermitian,
is supported on $\Lambda_i\subseteq[n]$, with
$|\Lambda_i|\le k$, and is a constant-degree polynomial in the
canonical observables.\\
\\
For $\Lambda\subseteq[n]$, define
\[
    W_\Lambda
    :=
    \bigotimes_{j\in\Lambda}(I+\widehat N_j).
\]
We require
\begin{equation}\label{eqn:norm}
    \left\|
    W_{\Lambda_i}^{-1/2}
    O_i
    W_{\Lambda_i}^{-1/2}
    \right\|_\infty
    \le C_k^{\mathsf{obs}},
\end{equation}
where $C_k^{\mathsf{obs}}<\infty$ depends only on the fixed
locality bound $k$, and not on $n$ or on the instance. Together with suitable photon-number moment bounds on $\rho$, this condition controls the truncation error in the recovery
guarantee.\\
\\
Both protocols use a Fock cutoff $M$.
Let $P_{\Lambda,M}$ project onto the Fock states with at most
$M$ photons in each mode of $\Lambda$.
For each $i\in[m]$, define
\[
    O_i^{(M)}
    :=
    P_{\Lambda_i,M}O_iP_{\Lambda_i,M}.
\]
The corresponding cutoff snapshots are defined separately
for the two protocols.

\subparagraph{Heterodyne detection.}
The local heterodyne protocol performs heterodyne detection
on each mode. Each round produces a phase-space outcome
\[
    z_t
    =
    (q_{t,1},p_{t,1},\ldots,
     q_{t,n},p_{t,n})
    \in\mathbb R^{2n},
\]
and we take $s_t=z_t$.
The formal inverse shadow associated with $z_t$
has improper characteristic function
\[
    \chi_{\widehat\rho_{z_t}^{\mathrm{het}}}(\xi)
    =
    \exp\left(
        \frac{\|\xi\|_2^2}{4}
        -i\xi^{\mathsf T}\Omega z_t
    \right).
\]
Since this formal shadow generally does not define a
trace-class operator, heterodyne recovery uses both the
Fock cutoff $M$ and phase-space regularization.
For a support $\Lambda$ of $r:=|\Lambda|$ modes, let
$e_{ba}^{(r)}$ be Schwartz operators approximating the
Fock matrix units $\ket{b}\bra{a}$, with
$a,b\in\{0,\ldots,M\}^{r}$, chosen using the phase-space
regularization of \cite[Section~IV-A]{BDLR24}, so that
\[
    \chi_{e_{ba}^{(r)}}\in C_c^\infty(\mathbb R^{2r}).
\]
For each fixed locality $r$, cutoff $M$, and choice of
phase-space regularization parameters, these operators,
and hence their characteristic functions, are fixed as
part of the recovery procedure.
The regularization parameters are suppressed from the notation.\\
\\
The regularized finite-cutoff heterodyne snapshot is
\[
    \widehat\rho_{z_\Lambda}^{(M),\mathrm{het}}
    :=
    \sum_{a,b\in\{0,\ldots,M\}^{r}}
    \left[
        \int_{\mathbb R^{2r}}
        \chi_{e_{ba}^{(r)}}(\xi)
        \exp\left(
            \frac{\|\xi\|_2^2}{4}
            -iz_\Lambda^{\mathsf T}\Omega_\Lambda\xi
        \right)
        \frac{d\xi}{(2\pi)^r}
    \right]
    \ket{a}\bra{b},
\]
where $\Omega_\Lambda$ is the symplectic form on the
modes in $\Lambda$.

\subparagraph{Homodyne detection.}
For each round $t$ and mode $j$, we choose
$\theta_{t,j}$ independently and uniformly from $[-\pi,\pi]$
and measure
\[
    q_{\theta_{t,j}}
    :=
    q_j\cos\theta_{t,j}
    +
    p_j\sin\theta_{t,j},
\]
obtaining an outcome $x_{t,j}\in\mathbb R$.
The record includes both the phases and the outcomes:
\[
s_t=\bigl((\theta_{t,j},x_{t,j})\bigr)_{j=1}^{n}.
\]
As in heterodyne detection, the formal inverse homodyne
shadow does not define a trace-class operator
\cite[Section~V-A]{BDLR24}.
We therefore work with a finite Fock cutoff $M$.
The matrix entries of the resulting snapshot are expressed
through homodyne pattern functions, whose integral
representations converge absolutely and require no
additional phase-space regularization.\\
\\
Following \cite[Section~V-C]{BDLR24}, we define these
functions through their Fourier transforms.
In the canonical quadrature convention fixed above,
for Fock indices $0\le a\le b\le M$, set
\[
    \widetilde f_{ab}(u)
    :=
    \pi(-i)^{b-a}
    \sqrt{\frac{2^{a-b}a!}{b!}}\,
    |u|\,u^{b-a}e^{-u^2/4}
    L_a^{(b-a)}\!\left(\frac{u^2}{2}\right),
    \quad u\in\mathbb R,
\]
where $L_a^{(b-a)}$ is the associated Laguerre polynomial
of degree $a$. The pattern functions are symmetric in their Fock indices.
Thus, for $a>b$, we define $\widetilde f_{ab}(u)$ by interchanging the indices in the formula above:
$\widetilde f_{ab}(u):=\widetilde f_{ba}(u)$.
The pattern functions are obtained by inverse Fourier
transformation:
\[
    f_{ab}(x)
    :=
    \frac{1}{2\pi}
    \int_{\mathbb R}
    \widetilde f_{ab}(u)e^{iux}\,du.
\]
Here $u$ is the Fourier variable conjugate to the
quadrature outcome $x$.
\\
\\
For fixed Fock indices, $\widetilde f_{ab}(u)$ is
$|u|$ times a polynomial in $u$ times the Gaussian
$e^{-u^2/4}$. Consequently, the inverse Fourier integral defining
$f_{ab}$ converges absolutely.\\
\\
For a single-mode record $(\theta,x)$, define the
finite-cutoff homodyne snapshot by
\[
    \widehat\rho_{\theta,x}^{(M),\mathrm{hom}}
    :=
    \sum_{a,b=0}^{M}
    e^{i(a-b)\theta}f_{ab}(x)
    \ket{a}\bra{b}.
\]
For a record $s=((\theta_j,x_j))_{j=1}^{n}$ and a
support $\Lambda$, the local cutoff snapshot is
\[
    \widehat\rho_{s_\Lambda}^{(M),\mathrm{hom}}
    :=
    \bigotimes_{j\in\Lambda}
    \widehat\rho_{\theta_j,x_j}^{(M),\mathrm{hom}}.
\]
\smallskip
\noindent
\emph{Recovery.}
For $\nu\in\{\mathrm{het},\mathrm{hom}\}$, define
\[
    \widehat o_{t,\nu}^{(M)}(i)
    :=
    \Tr\left(
        O_i^{(M)}
        \widehat\rho_{s_{t,\Lambda_i}}^{(M),\nu}
    \right),
    \qquad
    A_\nu^{(M)}(S,i)
    :=
    \frac1L\sum_{t=1}^{L}
    \widehat o_{t,\nu}^{(M)}(i).
\]
The defining integrals are evaluated to the prescribed
numerical accuracy.
For heterodyne recovery, we use the discretisation scheme
of \cite[Section~VII]{BDLR24}.
For homodyne recovery, we use the Fourier representation
of the pattern functions discussed in
\cite[Section~V-C]{BDLR24}, evaluating the integrals by
truncating the Gaussian tails and applying deterministic quadrature separately on the positive and negative half-lines.

\section{Efficient Algorithm for Permutation-Invariant Observable Consistency}\label{scn:permutation}

We first show that permutation-invariant observable consistency
is solvable in polynomial time when the observables are supplied
through their Schur--Weyl blocks. We then apply this result to the permutation-invariant classical-shadow protocol of \cite{SL24}.
 \begin{definition}[$\ObsCon_{\mathsf{PI}}$]\label{def:ObsCon_PI}
  The input consists of $n$-qubit observables and target values
$(O_i,y_i)_{i=1}^{m}$, where $m=\poly(n)$ and $y_i\in[-1,1]$,
together with parameters $0\le\alpha<\beta$ satisfying
$\beta-\alpha\ge 1/\poly(n)$. Additionally, each observable $O_i$ is permutation invariant, $U_\pi O_i U_\pi^\dagger=O_i
\quad \forall \pi\in S_n$,
satisfies $\|O_i\|_\infty\le1$, and is supplied explicitly through its reduced
Schur--Weyl blocks, which give a polynomial-size
description (see \cref{par:PI}):
\[
O_i=\bigoplus_{\lambda}I_{d_\lambda}\otimes(O_i)_\lambda,
\quad
(O_i)_\lambda\in\operatorname{Herm}(m_\lambda).
\]
The task is to decide between the following cases:
\begin{itemize}
    \item \textbf{Yes}: $\exists$  $n$-qubit state $\rho$ s.t.\ $\forall$ $i\in [m]$, 
    $\left| \Tr\left(O_i\rho\right)-y_i\right|\leq \alpha$.
    \item \textbf{No}: $\forall$ $n$-qubit states $\rho$ $\exists\; i\in [m]$ s.t.\  $\left| \Tr \left(O_i\rho\right)-y_i\right|\geq \beta$.
\end{itemize} 
 \end{definition}
 \begin{theorem}\label{thm:ObsConPIinP}
     $\ObsCon_{\mathsf{PI}}$ is in $\p$.
 \end{theorem}
 \begin{proof}
Without loss of generality, we can restrict our search to permutation invariant (PI) states:
\[
\rho^{\mathsf{PI}}:=\frac{1}{n!}\sum_{\sigma\in S_n}U_{\sigma}\rho U_{\sigma}^{\dagger}.
\]
This is possible because
\[
\Tr(O_i\rho^{\mathsf{PI}})=\frac{1}{n!}\sum_{\sigma\in S_n}\Tr(O_iU_{\sigma}\rho U_{\sigma}^{\dagger})=\frac{1}{n!}\sum_{\sigma\in S_n}\Tr(U_{\sigma}^{\dagger}O_iU_{\sigma}\rho)=\Tr(O_i\rho),
\]
where the last equality comes from the fact that each $O_i$ is PI.\\
\\
By Schur--Weyl duality, every PI state has the form
\[
\rho=\bigoplus_{\lambda}I_{d_\lambda}\otimes\rho_\lambda,
\]
where $\rho_\lambda\succeq0$ is an
$m_\lambda\times m_\lambda$ matrix with
$m_\lambda:=\lambda_1-\lambda_2+1$.
For $n$ qubits, $\lambda$ ranges over partitions of $n$
with at most two rows, which we can write as
\[
\lambda=(n-k,k),
\quad k=0,\ldots,\lfloor n/2\rfloor.
\]
Thus there are $\lfloor n/2\rfloor+1=O(n)$ blocks,
each of dimension $m_\lambda=n-2k+1\le n+1$. Using the block decompositions of $O_i$ and $\rho$, we obtain
\[
\Tr(O_i\rho)=\Tr\left(\bigoplus_{\lambda}I_{d_{\lambda}}\otimes(O_i)_{\lambda}\rho_{\lambda}\right)=\sum_{\lambda}d_{\lambda}\Tr((O_i)_{\lambda}\rho_{\lambda}).
\]
We therefore obtain the following SDP feasibility problem:
\begin{alignat}{2}
    \left|\sum_\lambda d_\lambda
    \Tr\bigl((O_i)_\lambda\rho_\lambda\bigr)-y_i\right|
    &\le\alpha
    &\quad &\forall i\in[m],\\
    \rho_\lambda
    &\succeq0
    &\quad &\forall\lambda,\\
    \sum_\lambda d_\lambda\Tr(\rho_\lambda)
    &=1.
\end{alignat}
Each Hermitian block $\rho_\lambda$ has $m_\lambda^2$ real
parameters. Hence the total number of real variables is
\[
N_{\mathrm{var}}
=\sum_\lambda m_\lambda^2
=\sum_{k=0}^{\lfloor n/2\rfloor}(n-2k+1)^2
=O(n^3).
\]
Since this SDP has polynomial size and the promise gap $\beta-\alpha$ is inverse polynomial, the promised feasibility problem can be decided in polynomial time using the ellipsoid
method~\cite{GLS88}. Hence $\ObsCon_{\mathsf{PI}}\in\mathsf{P}$.
\end{proof}

\begin{definition}[Permutation-invariant classical shadow]
\label{def:pi-shadow}
We specialize the classical-shadow framework of
\cref{def:classical-shadow} to the permutation-invariant
protocol described in \cref{par:PI}. A permutation-invariant
classical shadow on $n$ qubits is a tuple
$(S,O,A_{\mathsf{PI}},\chi)$, where
\begin{itemize}
    \item The shadow $S=\{s_\ell\}_{\ell=1}^L$ consists of
    $L=\poly(n)$ strings encoding records of the form
    \[
    s_\ell=(W_\ell,h_\ell),
    \qquad
    W_\ell\in\mathrm{SU}(2),
    \quad h_\ell\in\{0,\ldots,n\},
    \]
    where $W_\ell$ specifies the collective rotation
    $W_\ell^{\otimes n}$ and $h_\ell$ is the Hamming-weight
    outcome.

    \item $O=\{O_i\}_{i=1}^m$ is a family of $m=\poly(n)$
    Hermitian observables satisfying
    \[
    U_\pi O_iU_\pi^\dagger=O_i
    \quad\forall\pi\in S_n,
    \qquad
    \|O_i\|_\infty\le1.
    \]
    Each observable is supplied through its Schur--Weyl
    blocks, as in \cref{def:ObsCon_PI}.

    \item The recovery algorithm $A_{\mathsf{PI}}$ computes
    the snapshot operators
    \[
    \widehat\rho_\ell
    =\mathcal M_{\mathsf{PI}}^{+}(R_{W_\ell,h_\ell})
    \]
    in their Schur--Weyl representation, as described in
    \cref{par:PI}. For each observable $O_i$, it computes
    the estimates $\Tr(O_i\widehat\rho_\ell)$ and aggregates
    them via the median-of-means technique. The resulting recovery value is returned to $\chi$ bits of precision in polynomial time.
\end{itemize}
\end{definition}

\begin{definition}[Permutation-invariant classical-shadow validity
($\CSV_{\mathsf{PI}}$)]
\label{def:csv-pi}
Given a permutation-invariant classical shadow
$(S,O,A_{\mathsf{PI}},\chi)$ as in \cref{def:pi-shadow}
and thresholds $0\le\alpha<\beta$ satisfying
$\beta-\alpha\ge1/\poly(n)$, decide between the following cases:
\begin{itemize}
    \item \textbf{Yes}: There exists an $n$-qubit state $\rho$
    such that, for every $i\in[m]$, $\left|\Tr(O_i\rho)-A_{\mathsf{PI}}(S,i)\right|
    \le\alpha$.

    \item \textbf{No}: For every $n$-qubit state $\rho$,
    there exists an $i\in[m]$ such that $\left|\Tr(O_i\rho)-A_{\mathsf{PI}}(S,i)\right|
    \ge\beta$.
\end{itemize}
\end{definition}

\begin{corollary}
$\CSV_{\mathsf{PI}}\in \p$.  
\end{corollary}
\begin{proof}
Since $A_{\mathsf{PI}}$ is computable in polynomial time,
setting $y_i:=A_{\mathsf{PI}}(S,i)$ gives a polynomial-time
reduction from $\CSV_{\mathsf{PI}}$ to
$\ObsCon_{\mathsf{PI}}$, preserving the YES and NO promises
(see \cite[Remark~4.1]{KRMEG25}).
The result follows from \cref{thm:ObsConPIinP}.
\end{proof}

\begin{remark}[General observable representations]
\label{rem:pi-general-representation}
The classical algorithm of \cref{thm:ObsConPIinP} uses
explicit access to the reduced Schur--Weyl blocks.
Under the succinct descriptions and uniform efficient quantum
measurement access of \cref{def:classical-shadow}, PI observable
consistency with $m=\poly(n)$ observables belongs to
$\mathsf{BQP}$.
\end{remark}
\begin{proof}[Proof sketch.]
We reconstruct the reduced blocks using efficiently prepared
probe states, then apply the classical block-input algorithm.
Let $U_{\mathrm{SW}}$ denote the Schur--Weyl change of basis,
in which
\[
U_{\mathrm{SW}}O_iU_{\mathrm{SW}}^\dagger
=
\bigoplus_\lambda I_{d_\lambda}\otimes(O_i)_\lambda.
\]
For each sector $\lambda$, select one copy of the reduced
block, identified by a label
$s_\lambda$. Every copy carries the same matrix
$(O_i)_\lambda$, so one copy suffices.
Keeping the sector and copy fixed, choose $\ket v$ to be
a basis state or a superposition of two basis states in
$\mathbb C^{m_\lambda}$.
The efficient inverse Schur transform~\cite{BCH05} prepares
approximations to the corresponding $n$-qubit probe states
\[
\ket{\psi_{\lambda,v}}
=
U_{\mathrm{SW}}^\dagger
\bigl(\ket{\lambda}\otimes\ket{s_\lambda}\otimes\ket v\bigr).
\]
Measuring $O_i$ on these probes gives access to
\[
\bra{\psi_{\lambda,v}}O_i\ket{\psi_{\lambda,v}}
=
\bra v(O_i)_\lambda\ket v.
\]
We reconstruct each block from these expectations, as in
quantum detector tomography~\cite{LFCPSREPW08}.
For a fixed block $B=(O_i)_\lambda$, basis-state probes
give the diagonal entries $B_{aa}$.
For $a<b$, define
\[
\ket{v_{\mathrm R}}=\frac{\ket a+\ket b}{\sqrt2},
\qquad
\ket{v_{\mathrm I}}=\frac{\ket a-\mathrm{i}\ket b}{\sqrt2}.
\]
The off-diagonal entries follow from
\[
\begin{aligned}
\operatorname{Re}B_{ab}
&=\bra{v_{\mathrm R}}B\ket{v_{\mathrm R}}
  -\frac{B_{aa}+B_{bb}}2,\\
\operatorname{Im}B_{ab}
&=\bra{v_{\mathrm I}}B\ket{v_{\mathrm I}}
  -\frac{B_{aa}+B_{bb}}2.
\end{aligned}
\]
Thus, reconstructing all observables requires estimating $m\sum_\lambda m_\lambda^2=O(mn^3)$ expectation values in total. For $m=\poly(n)$, repeated measurements therefore produce explicit classical descriptions of all reduced blocks
to any required inverse-polynomial precision in quantum
polynomial time. The classical block-input algorithm of
\cref{thm:ObsConPIinP} then applies, establishing the $\mathsf{BQP}$ upper bound.
\end{proof}

\section{Complexity of Real Classical Shadow Validity}\label{scn:real}

\begin{definition}[Real classical shadow]
\label{def:real-shadow}
The definition is the same as \cref{def:classical-shadow},
only now our classical shadow has the structure dictated by
the local real-shadow protocol described in \cref{par:real}.
That means the following:
\begin{itemize}
    \item The shadow $S=\{s_t\}_{t=1}^L$ consists of
    $L=\poly(n)$ strings encoding measurement bases and outcomes,
    \[
    s_t=((P_{t,j},b_{t,j}))_{j=1}^n,
    \; \text{with} \;\;
    P_{t,j}\in\{X,Z\},\; b_{t,j}\in\{-1,+1\}.
    \]

    \item $O=\{O_i\}_{i=1}^m$ is a family of $m=\poly(n)$
    Hermitian $k=O(1)$-local observables,
    satisfying
    \[
    O_i\in\operatorname{span}_{\mathbb R}\{I,X,Z\}^{\otimes n},
    \qquad
    \|O_i\|_\infty\le1.
    \]

    \item The recovery algorithm $A_{\mathrm{real}}$ applies the inverse of the measurement channel on $\mathcal V_{XZ}$ to obtain the snapshot operators $\widehat\rho_t$ from $S$, and aggregates the estimates $\Tr(O_i\widehat\rho_t)$ via the median-of-means technique, as described in \cref{par:real}.
\end{itemize}
\end{definition}
\begin{remark}
    Measurement records and their associated snapshots are equivalent descriptions. Thus, we freely identify a record with its associated snapshot operator when no confusion can arise.
\end{remark}
\begin{definition}[Real classical-shadow validity
($\CSV_{\mathrm{real}}$)]
\label{def:csv-real}
Given a real classical shadow
$(S,O,A_{\mathrm{real}},\chi)$ as in
\cref{def:real-shadow} and thresholds $0\le\alpha_{\mathrm{CSV}}<\beta_{\mathrm{CSV}}$
satisfying $\beta_{\mathrm{CSV}}-\alpha_{\mathrm{CSV}}\ge\frac1{\poly(n)}$,
decide between the following cases:

\begin{itemize}
    \item \textbf{YES}: $\exists$ an $n$-qubit state
    $\sigma$ such that $\left|
    \Tr(O_i\sigma)-A_{\mathrm{real}}(S,i)
    \right|
    \le\alpha_{\mathrm{CSV}}
    \;\forall i\in[m]$.

    \item \textbf{NO}: $\forall$ $n$-qubit state
    $\sigma$, there exists $i\in[m]$ such that $\left|
    \Tr(O_i\sigma)-A_{\mathrm{real}}(S,i)
    \right|
    \ge\beta_{\mathrm{CSV}}$.
\end{itemize}
\end{definition}

\begin{definition}[$\iDXZLH$]
Fix a constant $p$. Given a Hamiltonian on a chain of
$n$ sites, each of dimension $2^p$, $H = \sum_{i=1}^{n-1} H_{i,i+1}$, with $H_{i,i+1}\in
\operatorname{span}_{\mathbb R}\{I,X,Z\}^{\otimes 2p},\;\forall\;i\in[n-1]$ and thresholds $\alpha,\beta$ with $\beta-\alpha\ge 1/\poly(n)$, decide
  \begin{itemize}
    \item YES: $\lmin(H) \le \alpha$.
    \item NO: $\lmin(H) \ge \beta$.
  \end{itemize}
\end{definition}

\begin{lemma}[$\QMA$-hardness of $\iDXZLH$]
\label{lem:1d-xz-lh}
There is a constant $p$ for which $\iDXZLH$ on a chain of qudits, with dimension $2^p$, is
$\QMA$-hard. 
\end{lemma}
\begin{proof}[Proof sketch.]
    The proof starts from the HNN Hamiltonian for a verifier with real gates. We encode configurations by storing overlapping pairs of neighboring symbols, with each pair label represented by a bit string of Hamming weight one. This representation allows label changes using products of \(X\) operators, while diagonal penalties enforce valid labels and agreement between neighboring registers. Using this encoding, we construct a Hamiltonian containing only \(I\), \(X\), and \(Z\) Pauli factors whose compression to the code space reproduces the original HNN Hamiltonian. A sufficiently large penalty and the Projection Lemma (\Cref{lem:projection}) preserve the inverse-polynomial promise gap. Finally, grouping consecutive registers gives nearest-neighbor interactions on constant-dimensional sites. The full proof is deferred to Appendix~\ref{app:1d-xz-lh}.
\end{proof}

\begin{definition}[$XZ$-certified $1D\text{-}\CLDM$]
Fix a constant $p$.  An $XZ$-certified $1D\text{-}\CLDM$ instance consists of
nearest-neighbor target states $\{\sigma_e\}$ on a line
of $n$ sites, each of dimension $2^p$, thresholds $\alpha_{\mathrm{CLDM}}<\beta_{\mathrm{CLDM}}$ with
$\beta_{\mathrm{CLDM}}-\alpha_{\mathrm{CLDM}}\ge 1/\poly(n)$, and an edge-local Hamiltonian
\[
H^{XZ}=\sum_{\ell=1}^M g_\ell,
\quad
g_\ell\in
\operatorname{span}_{\mathbb R}\{I,X,Z\}^{\otimes 2p},
\]
together with $\Delta\ge 1/\poly(n)$, such that $\sum_{\ell=1}^M
\Tr(g_\ell\sigma_{e_\ell})=0$,
where $e_\ell$ is the edge supporting $g_\ell$. Decide between the following two cases:
\begin{itemize}
    \item \textbf{YES:} $\exists\rho\;\text{such that}\;
\|\Tr_{\bar e}(\rho)-\sigma_e\|_1\le\alpha_{\mathrm{CLDM}}
\quad\forall e$

    \item \textbf{NO:} $\lambda_{\min}(H^{XZ})\ge\Delta$, and $\forall\rho\;\exists e$ with $\|\Tr_{\bar e}(\rho)-\sigma_e\|_1\ge\beta_{\mathrm{CLDM}}$.
\end{itemize}
\end{definition}

\begin{lemma}
\label{lem:xz-certified-cldm}
$XZ$-certified $1$D-CLDM on a line of constant-dimensional sites is
$\QMA$-hard.
\end{lemma}
\begin{proof}[Proof sketch.]
Using the simulatable verifier construction of \cite{BG22} and our overlap encoding from \Cref{lem:1d-xz-lh}, we efficiently construct local target states that, on YES instances, approximate the marginals of an encoded history state. On every input, these targets assign total energy zero to the encoded $XZ$-Hamiltonian. On NO instances, however, every global state has energy at least an inverse-polynomial amount. A state sufficiently close to all targets would therefore have energy below this lower bound, giving a contradiction. Polynomial bounds on the interaction norms ensure that the resulting consistency gap remains inverse polynomial. The full proof is deferred to Appendix~\ref{sapp:xz-certified-cldm}.
\end{proof}

\begin{lemma}
\label{lem:csv-real-hard}
$\CSV_{\mathrm{real}}$ is $\QMA$-hard.
\end{lemma}

\begin{proof}
We use the $XZ$-certified $1D$-CLDM instances constructed
in the proof of \Cref{lem:xz-certified-cldm}. Throughout this proof, $n$ denotes the number of sites in the encoded chain. Each site consists of $p$ qubits,
where $p$ is a fixed constant, so the constructed global shadow acts on $np$ qubits.\\
\\
For an input $x$, that construction produces a Hamiltonian
\[
H_x^{XZ}=\sum_{\ell=1}^M g_\ell
\]
and edge targets $\{\sigma_e\}$.
For each $\ell\in[M]$, let $e_\ell$ be an edge containing
the support of $g_\ell$. Different terms may have the
same supporting edge. The targets satisfy
\[
\sum_{\ell=1}^M\Tr(g_\ell\sigma_{e_\ell})=0.
\]
Without loss of generality, divide $H_x^{XZ}$ and $\Delta_x$ by a common positive, efficiently computable polynomially bounded factor, retaining the same notation, so that
\[
\|g_\ell\|_\infty \le \frac{1}{2}
\quad \forall\,\ell,
\quad\text{and}\quad
0<\Delta_x\le 1.
\]
The NO-case bound $\;\lambda_{\min}\left(H_x^{XZ}\right)\ge \Delta_x$
is preserved.\\
\\
We adapt the shadow construction of \cite{KRMEG25} to
real measurements. Each qubit is measured uniformly at random in either the
$X$ or $Z$ basis, with two possible outcomes, giving four
single-qubit snapshot types
$2\ket{\psi}\bra{\psi}-I/2$, where
$\ket{\psi}\in\{\ket0,\ket1,\ket+,\ket-\}$. Recall that each site of the encoded chain is a block of
$p$ qubits, where $p$ is a fixed constant. Under the local
real-shadow protocol, a block snapshot is the tensor product of its $p$ single-qubit snapshots. Since each qubit has four possible snapshot types, there are $r=4^p$ possible block snapshots, which we enumerate as
$\{\widehat\eta_j\}_{j\in[r]}$. In a YES instance, \Cref{lem:xz-certified-cldm} gives
a state $\rho^*$ such that
\[
\left\|
\Tr_{\overline{\{i,i+1\}}}(\rho^*)-\sigma_{i,i+1}
\right\|_1\le \alpha_{\mathrm{CLDM}}
\quad\forall i.
\] 
Let $\widehat\rho$ be a global real-shadow snapshot of
$\rho^*$, and define
\[
\widehat\rho_{i,i+1}
:=
\Tr_{\overline{\{i,i+1\}}}(\widehat\rho).
\]
Let $q_{i,j,k}$ be the probability that this snapshot equals
$\widehat\eta_j\otimes\widehat\eta_k$. By unbiasedness,
\[
\mathbb E[\widehat\rho_{i,i+1}]
=
\sum_{j,k\in[r]}q_{i,j,k}\,
\widehat\eta_j\otimes\widehat\eta_k
=
\Pi_{XZ}\operatorname{Tr}_{\overline{\{i,i+1\}}}(\rho^*).
\]
Now consider $L=\poly(|x|)$ independent global real-shadow samples of
$\rho^*$. Since $p$ is constant, a sufficiently large
polynomial $L$ ensures that, with positive probability,
all empirical edge averages are simultaneously within
$\epsilon/2$ of their expectations in trace norm.
On an edge of $2p$ qubits, Hilbert--Schmidt orthogonality gives
\[
\|\Pi_{XZ}(B)\|_1
\le 2^p\|\Pi_{XZ}(B)\|_2
\le 2^p\|B\|_2
\le 2^p\|B\|_1
\]
for every Hermitian operator $B$, where $\|\cdot\|_2$ is the Hilbert-Schmidt norm. Choosing
$\alpha_{\mathrm{CLDM}}\le\epsilon/2^{p+1}$ therefore makes these averages $\epsilon$-close to the projected targets, since their trace-norm error is at most
$\epsilon/2+2^p\alpha_{\mathrm{CLDM}}\le\epsilon$.\\
\\
We seek integer counts $n_{i,j,k}$ for the occurrences
of $\widehat\eta_j\otimes\widehat\eta_k$ on edge $(i,i+1)$.
The counts on shared blocks must agree so that the local
samples can be combined into global samples.
The global samples considered above satisfy these matching
conditions automatically, proving that suitable counts
exist on YES instances. We find them by solving the
following integer program:
\begin{align}
\left\|
\Pi_{XZ}\sigma_{i,i+1}
-\frac{1}{L}\sum_{j,k\in[r]}
n_{i,j,k}\widehat\eta_j\otimes\widehat\eta_k
\right\|_1
&\le\epsilon,
&& i\in[n-1],                                      \label{eq:csv-ip-approx}\\
\sum_{j\in[r]}n_{i,j,t}
&=
\sum_{k\in[r]}n_{i+1,t,k},
&& i\in[n-2],\ t\in[r],                            \label{eq:csv-ip-consistency}\\
n_{i,j,k}&\in\mathbb Z_{\ge0},
&& i\in[n-1],\ j,k\in[r],                          \label{eq:csv-ip-integrality}\\
\sum_{j,k\in[r]}n_{i,j,k}&=L,
&& i\in[n-1].                                      \label{eq:csv-ip-size}
\end{align}
Since $r$ is constant and the matching constraints couple
only neighboring edges, we solve this system in polynomial
time using \cite[Algorithm 1]{KRMEG25}.
If it is infeasible, we output a fixed NO instance.
\\
\\
We now define ``local shadows'' $S_{i} = \{s_{i,l}\}_{l\in[L]}$ by taking $n_{i,j,k}$ copies of $\widehat\eta_j\otimes\widehat\eta_k$ and compute permutations $f_i\in \textup S_L$, such that $\Tr_{1}(s_{i,l})=\Tr_2(s_{i+1,f_i(l)})$ via a perfect matching. Finally, we construct our global shadow:
  \begin{equation}
    S_0 = \{s_{l}\}_{l\in [L]}, \quad s_l = s_{1,l}\otimes \Tr_1(s_{2,f_1(l)}) \otimes \Tr_1(s_{3,f_2(f_1(l))})\otimes \dotsm\otimes \Tr_1(s_{n-1,(f_{n-2}\circ\dotsm\circ f_1)(l)}).
  \end{equation}
We form the final shadow $S$ by placing an identical copy
of $S_0$ in each MoM bucket. All bucket means are therefore equal, so the recovery value equals the empirical average over $S_0$. The observables are the terms of our $XZ$-Hamiltonian, $O_\ell:=g_\ell$ with $\ell\in [M]$, and the recovery algorithm is $A_{\mathrm{real}}$. If the recovery algorithm uses $K$ buckets, the final shadow $S$ contains $KL$ records.\\
\\
For each edge $e=\{i,i+1\}$, define
\[
\widetilde{\sigma}_e
:=
\frac{1}{L}
\sum_{j,k\in[r]}
n_{i,j,k}\widehat{\eta}_j\otimes\widehat{\eta}_k.
\]
The stitching construction preserves each local edge multiset, and
therefore 
\[
A_{\mathrm{real}}(S,\ell)
=
\Tr\left(
g_\ell\widetilde{\sigma}_{e_\ell}
\right).
\]
Since $g_\ell$ is $XZ$-only,
\[
\Tr\left(
g_\ell\Pi_{XZ}\sigma_{e_\ell}
\right)
=
\Tr(g_\ell\sigma_{e_\ell}).
\]
Consequently $\forall\ell\in[M]$, 
\begin{align}
\left|
A_{\mathrm{real}}(S,\ell)-\Tr(g_\ell\sigma_{e_\ell})
\right|
&=
\left|
\Tr\left(
g_\ell
\left[
\widetilde{\sigma}_{e_\ell}
-\Pi_{XZ}\sigma_{e_\ell}
\right]
\right)
\right|
\nonumber\\
&\leq
\|g_\ell\|_\infty
\left\|
\widetilde{\sigma}_{e_\ell}
-\Pi_{XZ}\sigma_{e_\ell}
\right\|_1
\nonumber\\
&\leq\epsilon.
\label{eq:csv-recovery-approx}
\end{align}
\textbf{Completeness}: Suppose we start with a YES instance. By \Cref{lem:xz-certified-cldm}, there exists a state $\rho^*$ such that 
\[
\|\Tr_{\overline{\{i,i+1\}}}(\rho^*)-\sigma_{{i,i+1}}\|_1\le\alpha_{\mathrm{CLDM}},\;\;\forall i\in[n-1]
\]
From that we get:
\[
\begin{aligned}
|\Tr(g_\ell\rho^*)-A_{\mathrm{real}}(S,\ell)|&=|\Tr(g_{\ell}\Tr_{\overline{e_\ell}}(\rho^*))-A_{\mathrm{real}}(S,\ell)|\\
&\le|\Tr(g_{\ell}\Tr_{\overline{e_\ell}}(\rho^*))-\Tr(g_\ell\sigma_{e_\ell})|+|\Tr(g_\ell\sigma_{e_\ell})-A_{\mathrm{real}}(S,\ell)|\\
&\le\|g_\ell\|_\infty\|\Tr_{\overline{e_\ell}}(\rho^*)-\sigma_{e_\ell}\|_1+\epsilon\le\alpha_{\mathrm{CLDM}}+\epsilon.
\end{aligned}
\]
Thus the constructed instance is a YES instance with $\alpha_{\mathrm{CSV}}:=\alpha_{\mathrm{CLDM}}+\epsilon$.\\
\\
\textbf{Soundness}: Suppose we start with a NO instance. By \Cref{lem:xz-certified-cldm} we have 
\[
\lambda_{\min}(H_x^{XZ})\ge\Delta_x.
\]
Now assume for contradiction that there exists a global
state $\rho$ such that
\[
\left|
\Tr\left(g_\ell\Tr_{\overline{e_\ell}}(\rho)\right)
-A_{\mathrm{real}}(S,\ell)
\right|
<\beta_{\mathrm{CSV}}
\quad\forall\ell\in[M].
\]
Therefore
\begin{align*}
\Tr(H_x^{XZ}\rho)
&=
\sum_{\ell=1}^{M}
\Tr\left(
g_\ell\,
\Tr_{\overline{e_\ell}}(\rho)
\right)<
\sum_{\ell=1}^{M}A_{\mathrm{real}}(S,\ell)+M\beta_{\mathrm{CSV}}\\
&=\sum_{\ell=1}^{M}
\left(A_{\mathrm{real}}(S,\ell)-\Tr(g_\ell\sigma_{e_\ell})\right)+\sum_{\ell=1}^{M}
\Tr(g_\ell\sigma_{e_\ell})+M\beta_{\mathrm{CSV}}\\
&\le M\epsilon+\underbrace{
\sum_{\ell=1}^{M}
\Tr(g_\ell\sigma_{e_\ell})}_{=\,0}+M\beta_{\mathrm{CSV}}=
M(\epsilon+\beta_{\mathrm{CSV}}).
\end{align*}
Choose
\[
\epsilon=\frac{\Delta_x}{16M},
\quad
\beta_{\mathrm{CSV}}=\frac{\Delta_x}{4M}.
\]
Since $\alpha_{\mathrm{CLDM}}$ is negligible, our choice
$\epsilon=\Delta_x/(16M)$ is sufficiently large compared
with $\alpha_{\mathrm{CLDM}}$ to guarantee feasibility
in the YES case. Then
\[
\Tr(H_x^{XZ}\rho)
<
M(\epsilon+\beta_{\mathrm{CSV}})
=
\frac{5\Delta_x}{16}
<
\Delta_x,
\]
contradicting
\[
\Tr(H_x^{XZ}\rho)
\geq
\lambda_{\min}(H_x^{XZ})
\geq
\Delta_x.
\]
Therefore, for every global state $\rho$, there exists some
$\ell\in[M]$ such that
\[
\left|
\Tr\left(
g_\ell\,
\Tr_{\overline{e_\ell}}(\rho)
\right)
-
A_{\mathrm{real}}(S,\ell)
\right|
\geq\beta_{\mathrm{CSV}}.
\]
Hence the constructed $\CSV_{\mathrm{real}}$ instance is a NO
instance.\\
\\
Finally, since $\alpha_{\mathrm{CSV}}=\alpha_{\mathrm{CLDM}}+\epsilon$ and $\alpha_{\mathrm{CLDM}}\leq\epsilon$,
\[
\alpha_{\mathrm{CSV}}
\leq
2\epsilon
=
\frac{\Delta_x}{8M}.
\]
Therefore
\[
\beta_{\mathrm{CSV}}-\alpha_{\mathrm{CSV}}
\geq
\frac{\Delta_x}{4M}
-
\frac{\Delta_x}{8M}
=
\frac{\Delta_x}{8M}
\geq
\frac{1}{\poly(|x|)}.
\]
Thus $\CSV_{\mathrm{real}}$ is $\QMA$-hard.

\end{proof}

\section{Complexity of Fermionic Classical-Shadow Validity}\label{scn:fermion}

\begin{definition}[Matchgate fermionic classical shadow]
\label{def:mfs-shadow}
We specialize the classical-shadow framework of
\cref{def:classical-shadow} to the matchgate protocol described
in \cref{par:MFS}. A matchgate fermionic classical shadow on
$n$ modes is a tuple $(S,O,A_{\mathsf{MFS}},\chi)$, where
\begin{itemize}
    \item The shadow $S=\{s_t\}_{t=1}^L$ consists of
    $L=\poly(n)$ strings of the form
    \[
    s_t=(Q_t,b_t),
    \;\text{where}\;\;
    Q_t\in B(2n),\; b_t\in\{0,1\}^n,
    \]
    encoding the sampled Clifford-matchgate unitary and the
    occupation-number measurement outcome.

    \item $O=\{O_i\}_{i=1}^m$ is a family of $m=\poly(n)$
    succinctly described Hermitian observables satisfying
    $\|O_i\|_\infty\le1$, drawn from the classes described in
    \cref{par:MFS}: local even fermionic observables of constant
    locality, fermionic Gaussian density operators, and
    projectors onto Slater determinants.

    \item The recovery algorithm $A_{\mathsf{MFS}}$ applies the inverse of the measurement channel restricted to its image $\Gamma_{\mathrm{even}}$ to obtain the snapshot operators $\widehat\rho_t$ from $S$. It computes the estimates $\Tr(O_i\widehat\rho_t)$ using the efficient post-processing rules described in \cref{par:MFS} and aggregates them via the median-of-means technique.
\end{itemize}
\end{definition}

For $0\le N\le n$, define the $N$-particle subspace by
\[
    \mathcal H_{n,N}
    :=
    \operatorname{span}
    \left\{
        \ket{x_1,\ldots,x_n}:
        x_j\in\{0,1\},\
        \sum_{j=1}^n x_j=N
    \right\}.
\]
Thus, $\mathcal H_{n,N}$ is the subspace spanned by occupation-number
basis states containing exactly $N$ occupied modes. Let $\mathcal D(\mathcal H_{n,N}):=\left\{\rho\succeq0:\Tr(\rho)=1,\;\rho\text{ is supported on }\mathcal H_{n,N}\right\}$.

\begin{definition}
    [$\CSV^{(N)}_{\mathsf{MFS}}$]\label{def:CSV_MFS} Given integers $n,N$ with $0\le N\le n$, a matchgate fermionic shadow $(S,O,A_{\mathsf{MFS}},\chi)$, and thresholds $0\le\alpha<\beta$, with $\Delta:=\beta-\alpha\ge1/\poly(n)$, decide between the following two cases:
    \begin{itemize}
        \item \textbf{Yes:} $\exists \rho\in\mathcal D(\mathcal H_{n,N})\;\text{such that}\;\forall i,\left|\Tr(O_i\rho)-A_{\mathsf{MFS}}(S,i)\right|\le\alpha$
        \item \textbf{No:} $\forall \rho\in \mathcal D(\mathcal H_{n,N}), \exists\;i\;\text{such that}\;\left|\Tr(O_i\rho)-A_{\mathsf{MFS}}(S,i)\right|\ge\beta. $
    \end{itemize}
\end{definition}
Our reduction starts from the following exact version of the fermionic $N$-representability problem.

\begin{definition}
    [Exact N-representability]\label{def:Exact_Representability} Given integers $N,d$ with $2\le N\le d\le \poly(N)$, a two-fermion density matrix $\rho^{(2)}\in \C^{\binom{d}{2}\times \binom{d}{2}} $, specified using $\poly(d)$ bits, and threshold $\beta\ge1/\poly(N)$, decide between the following:
    \begin{itemize}
        \item \textbf{Yes}: there exists an N-fermion state $\sigma$ on d modes such that $\;\Tr_{3,\dots,N}(\sigma)=\rho^{(2)}$.
        \item\textbf{No}: for every N-fermion state $\sigma$ on d modes, $\|\Tr_{3,\dots,N}(\sigma)-\rho^{(2)}\|_1\ge\beta$.
    \end{itemize}
\end{definition}
\begin{theorem}
    [{\cite[Theorem~5.7]{KR25}}] Exact $N$-representability is $\QMA$-complete under Karp reductions.
\end{theorem}
We next express the entries of a two-fermion reduced density matrix as expectation values of Hermitian fermionic observables. This gives
an intermediate observable-consistency problem for our reduction.
Let
\[
\mathcal P_d:=\{\{p,q\}:1\le p<q\le d\},\;|\mathcal P_d|=\binom{d}{2}
\]
and fix a total ordering on $\mathcal P_d$. For $P=\{p,q\}$ and $ Q=\{r,s\}$ with $p<q$ and $r<s$, define
\[
E_{P,Q}:=a_p^\dagger a_q^\dagger a_s a_r.
\]
Writing $\ket{P}:=a_p^\dagger a_q^\dagger\ket{\mathrm{vac}}$, the restriction of $E_{P,Q}$ to the two-particle sector is the matrix
unit $\ket{P}\bra{Q}$. More generally, for an $N$-fermion
state $\sigma$, let $\rho_\sigma^{(2)}:=\Tr_{3,\ldots,N}(\sigma)$ denote its
normalized two-particle reduced density matrix. Then
\[
\Tr(E_{P,Q}\sigma)
=
C_N(\rho_\sigma^{(2)})_{Q,P},
\quad
C_N:=\binom N2.
\]
The factor $C_N$ counts the unordered
particle pairs under the normalization $\Tr(\rho_\sigma^{(2)})=1$.

Since $E_{P,Q}^\dagger=E_{Q,P}$, we obtain Hermitian observables by defining
\[
D_P:=E_{P,P},\;\; R_{P,Q}:=\frac{E_{P,Q}+E_{Q,P}}{2},\;\;I_{P,Q}:=\frac{E_{Q,P}-E_{P,Q}}{2i}.
\]
Their expectation values recover, up to a factor $C_N$, the diagonal entries and the real and imaginary parts of the off-diagonal entries of $\rho_\sigma^{(2)}$, respectively. We collect these observables into the family $\;\;\mathcal O_{2\mathrm{RDM}}:=\{D_P:P\in\mathcal P_d\}\cup \{R_{P,Q}, I_{P,Q}:P,Q\in\mathcal P_d,P<Q\}$.\\
\\
Since creation and annihilation operators have operator
norm one, $\|E_{P,Q}\|_\infty\le1$. The triangle inequality therefore gives $\|O\|_\infty\le1$ for every 
$O\in\mathcal O_{2\mathrm{RDM}}$. Each of these observables preserves the particle number, is even, and is supported on at most four fermionic modes.
\begin{definition}[$\FermObsCon^{(N)}_{\mathcal O_{2\mathrm{RDM}}}$]\label{def:FermObsCon}
     The input consists of integers $N,d$, with $2\le N\le d\le \poly(N)$ and $d\ge4$, polynomially many pairs $\{(O_i,y_i)\}_{i=1}^{m}$ where $O_i\in \mathcal O_{2\mathrm{RDM}},\; y_i\in[-1,1]$, and threshold $ \frac{1} {\poly(N)}\le\beta_{\mathsf{Ferm}}\le2$. Decide between the following cases:
\begin{itemize}
    \item \textbf{Yes}: $\exists$  $N$-fermion state $\sigma$ on $d$ modes such that $\forall$ $i\in [m]$, 
    $ \Tr\left(O_i\sigma\right)=y_i$.
    \item \textbf{No}: $\forall$ $N$-fermion states $\sigma$ on $d$ modes $\exists i\in [m]$ such that $\left| \Tr \left(O_i\sigma\right)-y_i\right|\ge \beta_{\mathsf{Ferm}}$.
\end{itemize}
\end{definition}

\begin{lemma}\label{lem:fermobscon-hard}
   $\FermObsCon^{(N)}_{\mathcal O_{2\mathrm{RDM}}}$ is $\QMA$-hard. 
\end{lemma}

\begin{proof}
We reduce from Exact $N$-representability restricted
to $d=2N\ge4$, which remains $\QMA$-hard by the
construction in \cite[Section~5.1]{KR25}. Let $(N,d,\rho^{(2)},\beta)$ be the input instance, and set $D:=\binom{d}{2}$. For an $N$-fermion state $\sigma$, the 2-RDM coordinate observables recover the entries of its normalized two-particle RDM:
\[
\Tr(D_P\sigma)=\binom{N}{2}(\rho_\sigma^{(2)})_{P,P},\;\; \Tr(R_{P,Q}\sigma)=\binom{N}{2}\operatorname{Re}(\rho_\sigma^{(2)})_{P,Q},\;\;\Tr(I_{P,Q}\sigma)=\binom{N}{2}\operatorname{Im}(\rho_\sigma^{(2)})_{P,Q}.
\]
Therefore, we introduce the following constraints $(O,y)$:
\[
\left(D_P,\binom{N}{2}(\rho^{(2)})_{P,P}\right),\;\; \left(R_{P,Q},\binom{N}{2}\operatorname{Re}(\rho^{(2)})_{P,Q}\right),\;\;\left(I_{P,Q},\binom{N}{2}\operatorname{Im}(\rho^{(2)})_{P,Q}\right).
\]
There are $D+2\binom{D}{2}=D^2$ constraints, all computable
in polynomial time. If any target lies outside $[-1,1]$,
output a fixed NO instance. This cannot occur on a YES
input, since $\|O_i\|_\infty\le1$ for every $i$. Otherwise, set the output threshold to
\[
\beta_{\mathsf{Ferm}}
:=
C_N\frac{\beta}{\sqrt{2}D^{3/2}}.
\]
\textbf{Completeness:} If the Exact $N$-representability instance is YES, there exists an $N$-fermion state $\sigma$ whose two-particle RDM equals $\rho^{(2)}$, i.e., $\rho_\sigma^{(2)}=\rho^{(2)}$. Hence every constructed constraint is satisfied exactly.\\
\\
\textbf{Soundness:} Suppose the Exact $N$-representability instance is NO, then $\forall$ $N$-fermion state $\sigma$,
\[
\|\rho_{\sigma}^{(2)}-\rho^{(2)}\|_1\ge\beta.
\]
Set $X:=\rho_{\sigma}^{(2)}-\rho^{(2)}$ and let  $\delta(X)$ be the maximum absolute value among the diagonal coordinates $X_{P,P}$ and the real and imaginary parts of the off-diagonal coordinates $X_{P,Q}$. Since $X$ is a $D\times D$ Hermitian matrix,
\[
\|X\|_1\le\sqrt D\|X\|_2\le \sqrt2D^{3/2}\delta(X).
\]
Because $\|X\|_1\ge\beta$, we obtain $\delta(X)\ge\frac{\beta}{\sqrt2D^{3/2}}$. By construction of $\mathcal O_{2\mathrm{RDM}}$, the corresponding observable constraint is therefore violated by at least 
\[
\binom{N}{2}\frac{\beta}{\sqrt2 D^{3/2}}=\beta_{\mathsf{Ferm}}.
\]
Since $d\le\poly(N)$, this remains inverse polynomial.
\end{proof}

\begin{theorem}\label{thm:fermobscon4-to-csvmfs}
$\FermObsCon^{(N)}_{\mathcal O_{2\mathrm{RDM}}}
\;\le_m\;
\CSV^{(N')}_{\mathsf{MFS}}$.
\end{theorem}
To prove the theorem, fix an instance $\mathcal I=\{(O_i,y_i)\}_{i=1}^m$ of
$\FermObsCon^{(N)}_{\mathcal O_{2\mathrm{RDM}}}$ on $d$ modes, with threshold $\beta_{\mathsf{Ferm}}$. We construct a matchgate shadow whose recovery values encode the target expectations $y_i$, using auxiliary modes and observables to treat each constraint separately.\\
\\
For each constraint $i\in[m]$, introduce a fresh tag mode with Majorana operators $\eta_{i,1},\eta_{i,2}$ and define
  \[
  A_i^{\mathsf{tag}}:=i\eta_{i,1}\eta_{i,2},\;\;\; T_i:=O_iA_i^{\mathsf{tag}}
  \]
The fresh tag distinguishes constraint $i$ even when
physical supports overlap. Since $[O_i,A_i^{\mathsf{tag}}]=0$, the
observable $T_i$ is Hermitian and has norm at most one.
On the $+1$ eigenspace of $A_i^{\mathsf{tag}}$, it agrees with $O_i$.\\
\\
Since every $O_i\in\mathcal O_{2\mathrm{RDM}}$ acts on at most four modes, choose a set of exactly four physical modes 
 \[
 \Omega_i\subseteq [d],\;\;|\Omega_i|=4,\;\; \operatorname{supp}_{\text{mode}}(O_i)\subseteq \Omega_i.
 \]
 Let $\Gamma(\Omega_i)$ denote the eight Majorana labels associated with the modes in $\Omega_i$, and define the local Majorana region
  \[
  R_i:=\Gamma(\Omega_i)\cup\{\eta_{i,1},\eta_{i,2}\}
  \]
For the discrete matchgate sample $s=(Q,b)$, let $M_s$ denote the perfect matching on the Majorana labels induced by the signed permutation $Q$. We use the following consequence of the matchgate estimator formula: for every even Majorana monomial $G_F$, where $F$ is its set of Majorana labels, the matchgate estimator satisfies
\[
\hat{o}_s(G_F)=0\;\;\text{unless $F$ is a union of complete edges of $M_s$}.
\]
We say that $s$ is $i$-isolating if, for every $j\neq i$, at least one of the two tag Majoranas $\eta_{j,1},\eta_{j,2}$ is matched by $M_s$ to a Majorana label outside $R_j$.

\begin{lemma}[Isolation lemma]\label{lem:isolation}
If $s$ is $i$-isolating, then for every $j\neq i$,
\[
\hat{o}_s(A_j^{\mathsf{tag}})=0,\;\;\hat{o}_s(T_j)=0.
\]
  
\end{lemma}
\begin{proof}
    Fix $j\neq i$. Every Majorana monomial appearing in $T_j=O_jA_j^{\mathsf{tag}}$ contains both tag Majoranas $\eta_{j,1},\eta_{j,2}$ and its support is contained in $R_j$. Since $s$ is $i$-isolating, at least one of these tag Majoranas is matched by $M_s$ to a label outside $R_j$. Hence the support of the monomial cannot be a union of complete edges of $M_s$, and therefore its estimator vanishes. By linearity,
    \[
    \hat{o}_s(T_j)=0
    \]
    The same argument applied to $A_j^{\mathsf{tag}}=i\eta_{j,1}\eta_{j,2}$ gives $\hat{o}_s(A_j^{\mathsf{tag}})=0$.
\end{proof}

Isolation lets us handle each constraint separately.
We now construct the local sample families needed to
build a single shadow whose average estimates are
close to $y_i$ for $T_i$ and to $1$ for $A_i^{\mathsf{tag}}$,
for every $i$.\\
\\
Recall that $\Omega_i\subseteq [d]$ is the four-mode set containing the mode support of $O_i$. Let $\mathcal P_i:=\binom{\Omega_i}{2}$. Since $|\Omega_i|=4$, the set $\mathcal P_i$ contains six unordered pairs of modes. Define 
\[
\mathcal W_i:=\operatorname{span}_{\mathbb R}\left(\{D_P:P\in \mathcal P_i\}\cup\{R_{P,Q},I_{P,Q}:P,Q\in\mathcal P_i, P<Q\}\right).
\]
Thus $\mathcal W_i$ is the space of two-particle coordinate observables on the four modes $\Omega_i$. The six operators $D_P$, fifteen operators $R_{P,Q}$, and fifteen operators $I_{P,Q}$ form a basis of $\mathcal W_i$. Hence $\operatorname{dim}_{\mathbb R}\mathcal W_i=36$. Now introduce the tagged local space 
    \[
    \widetilde{V}_i:=\operatorname{span}_{\mathbb R}\left(\{A_i^{\mathsf{tag}}\}\cup\{OA_i^{\mathsf{tag}}:O\in\mathcal W_i\}\right).
    \]
Thus $\widetilde V_i$ consists of the tag observable together with all 36 tagged coordinate observables. Every operator in $\mathcal W_i$ annihilates the
four-mode vacuum, so $I\notin\mathcal W_i$.
Since $(A_i^{\mathsf{tag}})^2=I$, multiplication by $A_i^{\mathsf{tag}}$ is injective,
and $A_i^{\mathsf{tag}}$ is independent of
$\{OA_i^{\mathsf{tag}}:O\in\mathcal W_i\}$. Thus
$\dim_{\mathbb R}\widetilde V_i=37$. Set $B_1^{(i)}:=A_i^{\mathsf{tag}}$, and let
$B_2^{(i)},\ldots,B_{37}^{(i)}$ be the
36 tagged coordinate observables obtained by
multiplying each of the above basis operators
of $\mathcal W_i$ by $A_i^{\mathsf{tag}}$, using the same
ordering under relabeling for every $i$.\\
\\
For the remainder of the construction, we may assume that either $m=1$ or $m\ge 3$; if $m=2$, duplicate one of the two constraints.
Every perfect matching on the ten Majorana labels in $R_i$ can be extended to an $i$-isolating perfect matching on the full system. Indeed, after fixing the matching inside $R_i$, pair the unused physical Majoranas arbitrarily. Arrange the inactive tag modes in a cycle and pair
the second Majorana of each tag with the first
Majorana of the next tag. Then every inactive tag Majorana is paired with a label outside its corresponding region $R_j$, so the resulting global matching is $i$-isolating. Any padding Majoranas introduced later are paired internally.\\
\\
Let $n'$ denote the total number of modes in the
ambient system. By \Cref{rem:mfs-zero-outcome},
we may take $b=0^{n'}$ for all samples constructed below.\\
\\
For each perfect matching $M=\{e_1,\ldots,e_5\}$ on $R_i$,
fix an $i$-isolating global extension $\widetilde M$,
with arbitrary fixed signs on the edges outside $R_i$.
We vary the signs of the five edges inside $R_i$ over
all $\omega\in\{\pm1\}^5$. Let $s_{M,\omega}$ be a sample
realizing $\widetilde M$ with sign $\omega_k$ on $e_k$
and the fixed signs outside $R_i$. Define 
\[
\Sigma_i:=\{s_{M,\omega}: M\;\text{a perfect matching on $R_i$}, \omega\in\{\pm 1\}^5\}.
\]
There are ($9!!$) perfect matchings on the ten local
Majorana labels and $2^5$ sign choices, so
$|\Sigma_i|=(9!!)\cdot2^5=O(1)$.
Every sample in $\Sigma_i$ is $i$-isolating. Consequently, by \Cref{lem:isolation},
\[
\hat{o}_s(T_j)=\hat{o}_s(A_j^{\mathsf{tag}})=0,\;\; s\in\Sigma_i,\;j\neq i.
\]
For each $s\in\Sigma_i$, define the local estimator vector
\[
u_s^{(i)}:=\bigl(\hat{o}_s(B_1^{(i)}),\ldots,\hat{o}_s(B_{37}^{(i)})\bigr)
\in\mathbb R^{37}.
\]
We next show that these vectors are symmetric under sign
reversal and that their convex hull contains a ball centered
at the origin whose radius grows linearly with $n'$.

\begin{lemma}[Sign flip lemma]\label{lem:sign-flip}
For every $s\in\Sigma_i$, there exists $s^-\in\Sigma_i$ such that
\[
\hat{o}_{s^-}(B)=-\hat{o}_s(B)
\quad
\forall B\in\widetilde V_i.
\]
In particular, $\;u_{s^-}^{(i)}=-u_s^{(i)}$.
\end{lemma}
\begin{proof}
Fix $s=s_{M,\omega}=(Q,b)\in\Sigma_i$.
Let $\omega^-$ be obtained from $\omega$ by switching
the sign of the edge containing $\eta_{i,1}$, and set
\[
s^-:=s_{M,\omega^-}\in\Sigma_i.
\]
Let $D_i\in B(2n')$ be the diagonal signed permutation
that changes the sign of $\eta_{i,1}$ and fixes every
other Majorana label. The sample $(QD_i,b)$ has the
same signed matching as $s^-$, so their snapshots,
and hence their estimator values, coincide.\\
\\
Now fix $B\in\widetilde{V}_i$. By the definition of the tagged local space, every Majorana monomial appearing in $B$ contains both tag labels $\eta_{i,1},\eta_{i,2}$. Thus we can write
\[
 B=\sum_Fb_FG_F,\;\;\{\eta_{i,1},\eta_{i,2}\}\subseteq F,
\]
where $G_F$ is the Hermitian Majorana monomial supported on $F$. Its single-shot estimator is 
\[
\hat{o}_{(Q,b)}(G_F)=i^{|F|/2}c_{|F|}(n')\operatorname{pf}(i(Q^{\mathsf T}C_{\ket{b}}Q)|_{F}),
\]
where $c_{2r}(n'):=\binom{2n'}{2r}/\binom{n'}{r}$. Replacing $Q$ by $QD_i$ gives

\[
    \hat{o}_{(QD_i,b)}(G_F)=i^{|F|/2}c_{|F|}(n')\operatorname{pf}(i(D_i^{\mathsf T}Q^{\mathsf T}C_{\ket{b}}QD_i)|_{F})
\]
Now applying the following Pfaffian identity:
\[
 \operatorname{pf}(\Lambda A\Lambda^{\mathsf T})=\operatorname{det}(\Lambda)\operatorname{pf}(A)
\]
and the fact that $\operatorname{det}(D_i|_{F})=-1$ we get that
\[
\hat{o}_{(QD_i,b)}(G_F)=-\hat{o}_{(Q,b)}(G_F)
\]
for every monomial in the expansion of $B$. By linearity of the single-shot estimator,
\[
\hat{o}_{s^-}(B)=\hat{o}_{(QD_i,b)}(B)=\sum_Fb_F\hat{o}_{(QD_i,b)}(G_F)=-\sum_Fb_F\hat{o}_{(Q,b)}(G_F)=-\hat{o}_s(B).
\]
Since this holds for every basis element $B_r^{(i)}$, it follows that $u_{s^-}^{(i)}=-u_s^{(i)}$.
\end{proof}

\begin{lemma}[Inradius lemma]\label{lem:inradius}
     There exists a constant $\kappa>0$, independent of the constraint $i$ and the mode number $n'$, such that
    \[
    B_2\left(0,\kappa(2n'-1)\right)\subseteq\operatorname{conv}\{u_s^{(i)}:s\in\Sigma_i\}.
    \]
    Here $B_2(0,r):=\{x\in\mathbb R^{37}:\|x\|_2\le r\}$ denotes the Euclidean ball centered at the origin. In particular, the convex hull of the local estimator vectors has inradius $\varepsilon_i(n')\ge\kappa(2n'-1)=\Omega(n').$
\end{lemma}
\begin{proof}
    Every $B\in\widetilde V_i$ is a sum of monomials
using two, four, or six Majoranas from $R_i$:
\[
B=\sum_{\substack{F\subseteq R_i\\|F|\in\{2,4,6\}}}
b_F(B)G_F.
\]
For each support $F$, define
\[
v_F^{(i)}
:=
\bigl(b_F(B_1^{(i)}),\ldots,b_F(B_{37}^{(i)})\bigr)
\in\mathbb R^{37}.
\]
The vector $v_F^{(i)}$ collects the coefficients of
one monomial $G_F$ across the 37 basis observables.
The vector $u_s^{(i)}$ collects the estimates of
those observables produced by one sample $s$.\\
    \\
    The coefficient vectors $v_F^{(i)}$ span $\mathbb R^{37}$. Indeed, if $x\in\mathbb R^{37}$ is orthogonal to every $v_F^{(i)}$, then $B:=\sum_{r=1}^{37}x_rB_r^{(i)}$ has Majorana coefficient
    \[
    b_F(B)=\sum_{r=1}^{37}x_rb_F(B_r^{(i)})=x\cdot v_F^{(i)}=0,\;\;\forall F.
    \]
    Hence $B=0$, and linear independence of the basis implies $x=0$.\\
    \\
    Therefore the symmetric convex hull of these vectors contains a ball around the origin. Since all local bases are obtained from the same  basis by relabeling, there exists a constant $\kappa>0$, independent of $i$, such that $B_2(0,\kappa)\subseteq\operatorname{conv}\{\pm v_F^{(i)}\}$.\\
    \\
    Now fix a support $F$ and choose a perfect matching $M=\{e_1,\dots,e_5\}$ on $R_i$ and $J\subseteq[5]$ such that $F=F_J:=\bigcup_{j\in J}e_j$. Define $\omega_J:=\prod_{j\in J}\omega_j$. For any $K\subseteq[5]$, write $F_K:=\bigcup_{k\in K}e_k$ and $\omega_K:=\prod_{k\in K}\omega_k$. The estimator of $G_{F_K}$ satisfies 
    \[
      \hat{o}_{s_{M,\omega}}(G_{F_K})= \operatorname{sgn}(M,K)c_{2|K|}(n')\,\omega_K,
    \]
   where $\operatorname{sgn}(M,K)\in\{\pm1\}$ is determined by the ordering and phase conventions, independently of $\omega$. Monomials whose supports are not unions of complete matching edges have zero estimator. The orthogonality identity
   \[
      2^{-5}\sum_{\omega\in\{\pm1\}^5} \omega_J\omega_K=\delta_{J,K}
   \]
   therefore cancels all contributions with $K\neq J$ and isolates the coefficient vector for $F=F_J$:
    \[
    \begin{aligned}2^{-5}\sum_{\omega\in\{\pm1\}^5}\omega_J\cdot u_{s_{M,\omega}}^{(i)}&=2^{-5}\sum_{\omega\in\{\pm1\}^5}\omega_J\left(\hat{o}_{s_{M,\omega}}(B_r^{(i)})\right)_{r=1}^{37}\\
    &=\operatorname{sgn}(M,J)c_{2|J|}(n')v_F^{(i)}.
    \end{aligned}
    \]
    By \Cref{lem:sign-flip}, each vector $\omega_J u_{s_{M,\omega}}^{(i)}$ equals $u_t^{(i)}$ for some $t\in\Sigma_i$. Hence the displayed average is a convex combination of estimator vectors. Since their convex hull is symmetric about the origin, we obtain
    \[
    \pm c_{2|J|}(n')v_F^{(i)}\in\operatorname{conv}\{u_s^{(i)}:s\in\Sigma_i\}.
    \]
    For $|F|=2|J|\in\{2,4,6\}$ and $n'\ge 5$ we have:
    \[
    c_{2|J|}(n')=\frac{\binom{2n'}{2|J|}}{\binom{n'}{|J|}}\ge2n'-1.
    \]
    Explicitly:
    \[
    c_2(n')=2n'-1,\;\;c_4(n')=\frac{(2n'-1)(2n'-3)}{3}\ge2n'-1,\;\;c_6(n')=\frac{(2n'-1)(2n'-3)(2n'-5)}{15}\ge 2n'-1
    \]
    Since the convex hull contains the origin and $\pm c_{2|J|}(n')v_F^{(i)}$, it also contains the line segments joining these points to the origin. As $c_{2|J|}(n')\ge 2n'-1$, we conclude that
    \[
       \pm(2n'-1)v_F^{(i)} \in \operatorname{conv}\{u_s^{(i)}:s\in\Sigma_i\}.
    \] 
    So we get
    \[
    \begin{aligned}
    B_2(0,\kappa(2n'-1))&=(2n'-1)B_2(0,\kappa)\subseteq (2n'-1)\operatorname{conv}\{\pm v_F^{(i)}\}\\&=\operatorname{conv}\{\pm(2n'-1)v_F^{(i)}\}\subseteq \operatorname{conv}\{u_s^{(i)}:s\in\Sigma_i\}
    \end{aligned}
    \]
    Hence, $\varepsilon_i(n')\ge\kappa(2n'-1)=\Omega(n')$.
\end{proof}
We now combine the isolation property and the inradius
bound to construct a single matchgate shadow with the
desired average estimator values.\\
\\
We will combine the $m$ local sample families with equal
weight, so each family must realize target averages
scaled by $m$. By \Cref{lem:inradius}, adding modes
increases the radius of the ball contained in each local
estimator convex hull. We therefore introduce $p$ padding
modes, with $p$ chosen large enough to cover these
rescaled targets. Denote their Majorana operators by
$\xi_{h,1},\xi_{h,2}$ for $h\in[p]$, and set
\[
n':=d+m+p,\quad A_h^{\mathsf{pad}}:=i\xi_{h,1}\xi_{h,2}.
\]
\begin{proposition}[Shadow synthesis]\label{prop:shadow-synthesis}
There exists $p=\poly(d+m)$ such that, for every
inverse-polynomial $\delta>0$, one can construct in
polynomial time a multiset $S_0$ of discrete $n'$-mode
matchgate samples satisfying, for every $i\in[m]$,
\begin{equation}
\left|\frac{1}{|S_0|}\sum_{s\in S_0}\hat{o}_s(T_i)- y_i\right|\le \delta,
\qquad
\left|\frac{1}{|S_0|}\sum_{s\in S_0}\hat{o}_s(A_i^{\mathsf{tag}})-1\right|\le \delta,
\end{equation}
and, for every $h\in[p]$
\[
\frac{1}{|S_0|}\sum_{s\in S_0}\hat{o}_s(A_h^{\mathsf{pad}})=1.
\]
\end{proposition}

\begin{proof} 
Let $\kappa>0$ be the constant from \Cref{lem:inradius}.
Since the finite set of vectors $\{\pm v_F^{(i)}\}$ is fixed up to relabeling
and explicitly computable, fix once and for all a rational constant
$0<\kappa_0\le\kappa$, which is hard-coded into the reduction. Add $p$ padding modes and define the total number of modes $n':=d+m+p$. Choose $p=\poly(d+m)$ sufficiently large that $\sqrt2m\le\kappa_0(2n'-1)$. For this choice of $n'$, use the $i$-isolating sample
families $\Sigma_i$ constructed above, with every
padding pair $\{\xi_{h,1},\xi_{h,2}\}$ internally matched. By \Cref{lem:inradius} and $\kappa_0\le\kappa$,
\[
B_2\left(0,\kappa_0(2n'-1)\right)\subseteq\operatorname{conv}\{u_s^{(i)}:s\in\Sigma_i\}.
\]
For each $i$, let $r(i)\in[36]$ satisfy $T_i=B_{r(i)+1}^{(i)}$. Let $e_1,\dots,e_{37}$ be the standard coordinate
vectors of $\mathbb R^{37}$.
Each local family will contribute a fraction $1/m$
of the samples, so we target the local mean
$m(y_i e_{r(i)+1}+e_1)$. Since $T_i=B_{r(i)+1}^{(i)},\; A_i^{\mathsf{tag}}=B_{1}^{(i)}$ and $|y_i|\le 1$, we have
\[
\bigl\|m(y_i e_{r(i)+1}+e_1)\bigr\|_2
=m\sqrt{y_i^2+1}
\le\sqrt{2}\,m
\le\kappa_0(2n'-1).
\]
Therefore $m\bigl(y_i e_{r(i)+1}+e_1\bigr)\in\operatorname{conv}\{u_s^{(i)}:s\in\Sigma_i\}$, and so there exists a probability distribution $p_i$ on $\Sigma_i$ such that
\[
\sum_{s\in\Sigma_i} p_i(s)\,u_s^{(i)}
=
m(y_i e_{r(i)+1}+e_{1}).
\]
To compute such a distribution, we solve the following
linear feasibility problem:

\begin{subequations}\label{eq:local-convex-program}
\begin{align}
\sum_{s\in\Sigma_i} p_i(s)\,u_s^{(i)}
&=
m\bigl(y_i e_{r(i)+1}+e_1\bigr),
\label{eq:local-convex-program-a}\\
\sum_{s\in\Sigma_i} p_i(s)
&=1,
\label{eq:local-convex-program-b}\\
p_i(s)
&\ge 0
\qquad
\text{for all }s\in\Sigma_i.
\label{eq:local-convex-program-c}
\end{align}
\end{subequations}
For each fixed gadget $i$, this is a constant-size linear system. Solving these systems for all $i\in[m]$ therefore takes polynomial time.\\
\\
The distribution $p_i$ specifies the sample weights
within family $i$. In the full shadow, each family will
account for $1/m$ of all samples. We therefore define
$q_i(s):=p_i(s)/m$, the intended fraction of the full
shadow assigned to sample $s$ from family $i$.
Consequently,
\[
\sum_{s\in\Sigma_i} q_i(s)=\frac1m,
\quad
\sum_{s\in\Sigma_i} q_i(s)u_s^{(i)}
= y_i e_{r(i)+1}+e_{1}.
\]
Each sample in $\Sigma_i$ is already a matchgate sample
on the full $n'$-mode system. Thus samples from different
families can be combined without overlap-count constraints. We therefore round the weights $q_i(s)$ to integer sample multiplicities separately within each family, preserving each family's fraction $1/m$ of the total samples. Choose a common denominator
$L=\poly(n',1/\delta)$, divisible by $m$, and nonnegative integers $N_{i,s}$ such that
\[ 
\left|\frac{N_{i,s}}{L}-q_i(s)\right|\le\frac{1}{L}
\quad
\forall s\in\Sigma_i,
\quad
\sum_{s\in\Sigma_i}N_{i,s}=\frac{L}{m}.
\]
Let $S_i$ contain exactly $N_{i,s}$ copies of each
sample $s\in\Sigma_i$, and define the intermediate
multiset
\[
S^*:=\bigsqcup_{i=1}^m S_i.
\]
By construction, $|S_i|=L/m$ and $|S^*|=L$.
The contribution of family $i$ to the average over
$S^*$ is therefore
\[
\frac1L\sum_{s\in S_i}u_s^{(i)}
=
\sum_{s\in\Sigma_i}\frac{N_{i,s}}L u_s^{(i)}.
\]
Rounding changes each weight by at most $1/L$.
Since there are only constantly many weights and the
estimator vectors have polynomially bounded norm,
choosing $L=\poly(n',1/\delta)$ sufficiently large
makes the resulting error at most $\delta$. Thus
\begin{equation}
\left\|
\frac1L\sum_{s\in S_i}u_s^{(i)}
-\bigl(y_i e_{r(i)+1}+e_1\bigr)
\right\|_2
\le\delta.
\end{equation}
Fix now $j\in[m]$. By \Cref{lem:isolation}, every sample in $\Sigma_i$ with $i\neq j $ has zero contribution to both $T_j$ and $A_j^{\mathsf{tag}}$. Hence
\begin{equation}
\frac1L\sum_{s\in S^*}\hat o_s(T_j)
=
\frac1L\sum_{s\in S_j}\hat o_s( T_j),
\quad
\frac1L\sum_{s\in S^*}\hat o_s(A_j^{\mathsf{tag}})
=
\frac1L\sum_{s\in S_j}\hat o_s(A_j^{\mathsf{tag}}).
\end{equation}
Since $T_j=B_{r(j)+1}^{(j)}$ and $A_j^{\mathsf{tag}}=B_1^{(j)}$,
the corresponding coordinate errors are each at most
$\delta$ by the vector bound above with $i=j$.
Combining this with the preceding isolation equalities
gives
\begin{equation}
\left|
\frac1L\sum_{s\in S^*}\hat o_s( T_j)- y_j
\right|
\le \delta,
\quad
\left|
\frac1L\sum_{s\in S^*}\hat o_s(A_j^{\mathsf{tag}})-1
\right|
\le \delta.
\end{equation}
It remains to make the average estimator of each
padding observable $A_h^{\mathsf{pad}}$ equal to $1$. Since each padding pair
$ \{\xi_{h,1},\xi_{h,2}\}$ is internally matched, the
degree-$2$ estimation formula gives
\[
\hat o_s(A_h^{\mathsf{pad}})=\pm c_2(n'),
\quad
c_2(n'):=\binom{2n'}2\binom{n'}1^{-1}=2n'-1.
\]

For each $s=(Q,b)\in S^*$, choose diagonal signed permutations $D_s^+,D_s^-\in B(2n')$, acting only on padding Majoranas, such that $s_{\text{pad}}^{\pm}:=(QD_s^{\pm},b)$ have the same underlying matching as
$s$ and satisfy
\[
\hat o_{s_{\text{pad}}^\pm}(A_h^{\mathsf{pad}})=\pm(2n'-1),
\;\;
\forall h\in[p].
\]
Such sign choices are possible because the sign of each internally matched padding edge can be changed independently by flipping one Majorana in that pair.\\
\\
Now let $B$ be an observable supported outside the
padding modes. Every Majorana monomial appearing in $B$ has disjoint support from the padding labels. Since $D_s^{\pm}$ acts only on padding Majoranas, it does not change the estimator of $B$. Hence
\[
\hat o_{s_{\text{pad}}^-}(B)=\hat o_{s_{\text{pad}}^+}(B)=\hat o_s(B)
\]
Construct $S_0$ by replacing every occurrence of $s\in S^*$ with $n'$ copies of $s_{\text{pad}}^+$ and $n'-1$ copies of $s_{\text{pad}}^-$. Then 
\[
|S_0|=(2n'-1)|S^*|.
\]
For every padding observable $A_h^{\mathsf{pad}}$, we have
\[
\frac{1}{|S_0|}\sum_{t\in S_0}\hat{o}_t(A_h^{\mathsf{pad}})=\frac{1}{(2n'-1)|S^*|}\sum_{s\in S^*}[n'(2n'-1)-(n'-1)(2n'-1)]=1
\]
For any observable supported outside the padding modes,
all $2n'-1$ samples replacing $s$ have the same estimator
value as $s$. Since every occurrence in $S^*$ is replaced
by the same number of samples, the average estimator
remains unchanged. Consequently, for every $i\in[m]$,
\[
\left|
\frac1{|S_0|}\sum_{s\in S_0}\hat{o}_s(T_i)-y_i
\right|\le\delta,
\quad
\left|
\frac1{|S_0|}\sum_{s\in S_0}\hat{o}_s(A_i^{\mathsf{tag}})-1
\right|\le\delta.
\]
\end{proof}
Using \Cref{prop:shadow-synthesis}, we now construct
the target shadow-validity instance and establish
completeness and soundness of the reduction.
\begin{proof}[Proof of \Cref{thm:fermobscon4-to-csvmfs}.]
Let $\mathcal I=\{(O_i,y_i)\}_{i=1}^m$ be an instance of
$\FermObsCon^{(N)}_{\mathcal O_{2\mathrm{RDM}}}$
on $d$ modes, with threshold $ \frac{1} {\poly(N)}\le\beta_{\mathsf{Ferm}}\le2$.\\
\\
Let $p$ be the number of padding modes supplied by
\Cref{prop:shadow-synthesis}, and set
\[
n':=d+m+p,\quad N':=N+m+p.
\]
Choose
\[
\lambda:=\frac{1}{2(m+p)},
\quad
\delta:=\frac{\lambda\beta_{\mathsf{Ferm}}}{16}.
\]
These parameters are inverse polynomial and satisfy
\[
(m+p)\frac{\lambda\beta_{\mathsf{Ferm}}}{2}
=\frac{\beta_{\mathsf{Ferm}}}{4}
\le\frac12.
\]

Apply \Cref{prop:shadow-synthesis} with error parameter
$\delta$ to obtain $S_0$. Form $S$ by repeating $S_0$
in $K=\poly(n')$ identical recovery blocks.
All block means coincide, so the raw median-of-means
estimates equal the corresponding averages over $S_0$.\\
\\
To ensure that the reported recovery values lie in $[-1,1]$, we include an overall factor of $1/2$ in the output
observable family:
\[
O':=
\left\{\frac{\lambda}{2}T_i,\frac12 A_i^{\mathsf{tag}}:i\in[m]\right\}
\cup
\left\{\frac12 A_h^{\mathsf{pad}}:h\in[p]\right\}.
\]
These observables are even, act on at most five modes,
and have operator norm at most one, hence belong to the supported matchgate-shadow observable class.\\
\\
Positive scaling commutes with the single-shot estimators,
block averages, and median-of-means aggregation:
\[
A_{\mathsf{MFS}}(S,cO)
=cA_{\mathsf{MFS}}(S,O)
\quad(c>0).
\]
The synthesis guarantees therefore imply, for every $i\in[m]$,
\[
\left|
A_{\mathsf{MFS}}\left(S,\frac{\lambda}{2}T_i\right)
-\frac{\lambda y_i}{2}
\right|
\le\frac{\lambda\delta}{2},
\qquad
\left|
A_{\mathsf{MFS}}\left(S,\frac12 A_i^{\mathsf{tag}}\right)-\frac12
\right|
\le\frac{\delta}{2},
\]
and, for every $h\in[p]$,
\[
A_{\mathsf{MFS}}\left(S,\frac12 A_h^{\mathsf{pad}}\right)=\frac12.
\]
Finally, set
\[
\alpha_{\CSV}:=\frac{\delta}{2},
\quad
\beta_{\CSV}:=
\frac{\lambda\beta_{\mathsf{Ferm}}}{4}
-\frac{\delta}{2}\quad \text{with}\quad
\beta_{\CSV}-\alpha_{\CSV}
=\frac{3\lambda\beta_{\mathsf{Ferm}}}{16}
\ge\frac{1}{\poly(N')}.
\]
The resulting instance consists of the shadow $S$, observable family $O'$, recovery algorithm $A_{\mathsf{MFS}}$, and thresholds
$\alpha_{\CSV}$ and $\beta_{\CSV}$.\\
\\
\textbf{Completeness:} Assume $\mathcal{I}$ is a YES instance. Then there exists an $N$-fermion state $\sigma$ on the original $d$ physical modes such that
\[
\Tr(O_i\sigma)=y_i\;\;\forall i\in [m].
\]
Place all $m+p$ auxiliary modes in their occupied states,
which are the $+1$ eigenstates of the corresponding tag
and padding observables and define $\tau=\sigma\otimes\tau_{\mathsf{tag,+}}\otimes\tau_{\mathsf{pad},+}$. Then $\tau$ has particle number $N'=N+m+p$ and
\[
\Tr(A_i^{\mathsf{tag}}\tau)=1,\;\;\;\Tr(A_h^{\mathsf{pad}}\tau)=1,\;\;\;\Tr(T_i\tau)=\Tr(O_i\sigma)=y_i.
\]
Therefore, for every $i\in[m]$,
\[
\begin{aligned}
\left|
\Tr\left(\frac{\lambda}{2}T_i\tau\right)
-A_{\mathsf{MFS}}\left(S,\frac{\lambda}{2}T_i\right)
\right|
&\le\frac{\lambda\delta}{2}
\le\frac{\delta}{2}
=\alpha_{\CSV},\\
\left|
\Tr\left(\frac12 A_i^{\mathsf{tag}}\tau\right)
-A_{\mathsf{MFS}}\left(S,\frac12 A_i^{\mathsf{tag}}\right)
\right|
&\le\frac{\delta}{2}
=\alpha_{\CSV}.
\end{aligned}
\]
For every $h\in[p]$,
\[
\left|
\Tr\left(\frac12 A_h^{\mathsf{pad}}\tau\right)
-A_{\mathsf{MFS}}\left(S,\frac12 A_h^{\mathsf{pad}}\right)
\right|
=0\le\alpha_{\CSV}.
\]
Hence the constructed $\CSV_{\mathsf{MFS}}^{(N')}$ instance is a YES instance.\\
\\
\textbf{Soundness:} Assume $\mathcal{I}$ is a NO instance. Then for every $N$-fermion state $\sigma$ there exists a constraint $i\in[m]$ such that:
\[
|\Tr(O_i\sigma)-y_i|\ge\beta_{\mathsf{Ferm}}.
\]
Suppose, toward contradiction, that there exists an
$N'$-fermion state $\tau$ satisfying, for every $i\in[m]$
and $h\in[p]$,
\[
\begin{aligned}
\left|
\Tr\left(\frac12 A_i^{\mathsf{tag}}\tau\right)
-A_{\mathsf{MFS}}\left(S,\frac12 A_i^{\mathsf{tag}}\right)
\right|
&<\beta_{\CSV},\\
\left|
\Tr\left(\frac12 A_h^{\mathsf{pad}}\tau\right)
-A_{\mathsf{MFS}}\left(S,\frac12 A_h^{\mathsf{pad}}\right)
\right|
&<\beta_{\CSV},\\
\left|
\Tr\left(\frac{\lambda}{2}T_i\tau\right)
-A_{\mathsf{MFS}}\left(S,\frac{\lambda}{2}T_i\right)
\right|
&<\beta_{\CSV}.
\end{aligned}
\]
Combining these inequalities with the scaled synthesis guarantees, we obtain
\[
|\Tr(A_i^{\mathsf{tag}}\tau)-1|<\varepsilon,\quad
|\Tr(A_h^{\mathsf{pad}}\tau)-1|<\varepsilon,\quad
|\Tr(\lambda T_i\tau)-\lambda y_i|<\varepsilon,
\]
where $\varepsilon:=2\beta_{\CSV}+\delta
=\frac{\lambda\beta_{\mathsf{Ferm}}}{2}$.\\
\\
We now use the auxiliary constraints to recover a
state in the original $N$-particle sector.
To fix all auxiliary modes in their occupied states,
define the projector onto the joint $+1$ eigenspace
of all tag and padding observables:
\[
 P:=\prod_{i=1}^m \frac{I+A_i^{\mathsf{tag}}}{2}\prod_{h=1}^p\frac{I+A_h^{\mathsf{pad}}}{2}  
\]
It is straightforward to see that for every $i\in[m]$ and $h\in[p]$,
\[
\Tr\left(\left(I-\frac{I+A_i^{\mathsf{tag}}}{2}\right)\tau\right)=\frac{1-\Tr(A_i^{\mathsf{tag}}\tau)}{2}<\frac{\varepsilon}{2},\;\;\Tr\left(\left(I-\frac{I+A_h^{\mathsf{pad}}}{2}\right)\tau\right)<\frac{\varepsilon}{2}.
\]
Since the projectors commute, the union bound gives:
\[
1-\Tr(P\tau)<\frac{m+p}{2}\varepsilon.
\]
With our choice of $\lambda$ we get
\[
(m+p)\varepsilon=(m+p)\frac{\lambda\beta_{\mathsf{Ferm}}}{2}\le\frac{1}{2}
\]
and therefore $q:=\Tr(P\tau)>0$. Let $\tau^+:=P\tau P/q$ be the normalized state obtained by
projecting onto the joint $+1$ eigenspace of all tag and padding observables. Since these $m+p$ auxiliary modes are now occupied and the total particle number is $N'=N+m+p$, the reduced state
\[
\sigma:=\Tr_{\mathsf{aux}}(\tau^+)
\]
is an $N$-fermion state on the original $d$ modes.
Since $A_i^{\mathsf{tag}}P=P$, we have $A_i^{\mathsf{tag}}\tau^+=\tau^+$, and hence
\[
\begin{aligned}
\Tr(T_i\tau^+)
&=\Tr\left[(O_i\otimes I_{\mathsf{aux}})A_i^{\mathsf{tag}}\tau^+\right]\\
&=\Tr\left[(O_i\otimes I_{\mathsf{aux}})\tau^+\right]\\
&=\Tr\left[O_i\,\Tr_{\mathsf{aux}}(\tau^+)\right]\\
&=\Tr(O_i\sigma).
\end{aligned}
\]
It remains to bound the change in the expectation
of $T_i$ under this projection. By the definition of $\tau^+$, we have
\[
\begin{aligned}
 \Tr(T_i\tau^+)=\frac{1}{q}\Tr(T_iP\tau P)
\end{aligned}
\]
Now because $I=P+(I-P)$ we can write:
\[
\begin{aligned}
\Tr(T_i\tau)&=\Tr(IT_iI\tau)=\Tr((P+
(I-P))T_i(P+(I-P))\tau)\\
&=\Tr(PT_iP\tau)+\Tr(((I-P)T_i(I-P))\tau)\\
&=x_1+x_2,
\end{aligned}
\]
where
\[
x_1:=\Tr(PT_iP\tau),
\quad
x_2:=\Tr\!\left[(I-P)T_i(I-P)\tau\right].
\]
Here the cross terms vanish because $[T_i,P]=0$. By cyclicity of the trace,
\[
\Tr(T_i\tau^+)
=\frac{1}{q}\Tr(T_iP\tau P)
=\frac{1}{q}\Tr(PT_iP\tau)
=\frac{x_1}{q}.
\]
We now bound $x_1$ and $x_2$:
\[
|x_1|=|\Tr(T_iP\tau P)|\le\|T_i\|_\infty\Tr(P\tau P)\le\Tr(P\tau P)=q
\]
\[
|x_2|=|\Tr(T_i(I-P)\tau(I-P))|\le \Tr((I-P)\tau (I-P))=1-q
\]
So finally we get:
\[
\begin{aligned}
|\Tr(T_i\tau^+)-\Tr(T_i\tau)|&=\left|x_1\left(\frac{1}{q}-1\right)-x_2\right|\\
&\le |x_1|\left(\frac{1}{q}-1\right)+|x_2|\\
&\le 2(1-q)<(m+p)\varepsilon.
\end{aligned}
\]
Combining all the above we get:
\[
\begin{aligned}
|\Tr(O_i\sigma)-y_i|&=|\Tr(T_i\tau^+)-y_i|\\
&\le|\Tr(T_i\tau^+)-\Tr(T_i\tau)|+|\Tr(T_i\tau)-y_i|\\
&<(m+p)\varepsilon+\varepsilon/\lambda\\
&=(m+p)\frac{\lambda\beta_{\mathsf{Ferm}}}{2}+\frac{\beta_{\mathsf{Ferm}}}{2}\\
&\le \frac{3\beta_{\mathsf{Ferm}}}{4}<\beta_{\mathsf{Ferm}}.
\end{aligned}
\]
This holds for every $i\in[m]$. Hence $\sigma$ is an $N$-fermion state satisfying all original constraints with error strictly smaller than $\beta_{\mathsf{Ferm}}$, contradicting the NO case.
\end{proof}
\begin{theorem}
    $\CSV^{(N')}_{\mathsf{MFS}}$ is $\QMA$-hard.
\end{theorem}
\begin{proof}
    This follows from \Cref{thm:fermobscon4-to-csvmfs,lem:fermobscon-hard}
\end{proof}

\section{Complexity of Continuous-Variable Classical-Shadow Validity}\label{scn:continuous}

\subsection{Heterodyne shadows}\label{sscn:heterodyne}

\begin{definition}[Heterodyne CV classical shadow]
\label{def:CShetero}
We adapt the classical-shadow framework of
\cref{def:classical-shadow} to the heterodyne protocol described
in \cref{par:HetCV}. For a fixed Fock cutoff $M$, a heterodyne
CV classical shadow on $n$ modes is a tuple
$(S,O,A_{\mathrm{het}}^{(M)},\chi)$, where
\begin{itemize}
    \item The shadow $S=\{z_t\}_{t=1}^L$ consists of
    $L=\poly(n)$ phase-space records
    \[
    z_t=(q_{t,1},p_{t,1},\ldots,q_{t,n},p_{t,n})
    \in\mathbb R^{2n},
    \]
    with each coordinate represented using $\poly(n)$ bits.

    \item $O=\{O_i\}_{i=1}^m$ is a family of $m=\poly(n)$
    succinctly described Hermitian observables. Each $O_i$
    is supported on at most $k=O(1)$ modes, is a constant-degree
    polynomial in the canonical observables, and satisfies
    the weighted photon-number bound specified in
    \cref{par:HetCV}, \Cref{eqn:norm}.

    \item $A_{\mathrm{het}}^{(M)}$ is the cutoff-$M$ recovery
    algorithm specified in \cref{par:HetCV}; given $S$ and
    $i\in[m]$, it returns the estimate
    $A_{\mathrm{het}}^{(M)}(S,i)$ for $O_i$ to $\chi$ bits of precision, in polynomial time.
\end{itemize}
\end{definition}

Now let $\mathcal H_n:=L^2(\mathbb R)^{\otimes n},
\;
\mathcal D(\mathcal H_n):=\{\sigma\in\mathcal T(\mathcal H_n):\sigma\ge 0, \Tr(\sigma)=1\}$. Here $\mathcal T(\mathcal H_n)$ denotes the space of trace-class operators acting on $\mathcal H_n$.
For a moment order $\ell\ge k$ and bound $B$, define 
\[
\mathcal D_{\ell,B}(\mathcal H_n):=\left\{\sigma\in\mathcal D(\mathcal H_n):\Tr((I+\widehat N_j)^\ell\sigma)\le B\ \text{for all }j\in[n]\right\}.
\]
\begin{definition}[$\CSV^{\ell,B,M}_{\mathrm{HetCV}}$]\label{def:CSV_het}
Given a heterodyne CV classical shadow $(S,O,A_{\mathrm{het}}^{(M)},\chi)$, as in \Cref{def:CShetero}, thresholds
$0\le\alpha_{\mathrm{CV}}<\beta_{\mathrm{CV}}$, with $\Delta_{\mathrm{CV}}:=\beta_{\mathrm{CV}}-\alpha_{\mathrm{CV}}\ge1/\poly(n)$, and moment parameters $\ell,B$ with $\ell\ge k$, decide between the following two cases:
\begin{itemize}

\item\textbf{YES}: $\exists \sigma\in\mathcal D_{\ell,B}(\mathcal H_n)$ such that $\left|\Tr(O_i\sigma)-A_{\mathrm{het}}^{(M)}(S,i)\right|\le \alpha_{\mathrm{CV}}
\;\;\forall i$.

\item\textbf{NO}: $\forall \sigma\in\mathcal D_{\ell,B}(\mathcal H_n), \exists i$
such that $\left|\Tr(O_i\sigma)-A_{\mathrm{het}}^{(M)}(S,i)\right|\ge \beta_{\mathrm{CV}}$.
\end{itemize}
\end{definition}

\begin{setup}[Reduction setup]\label{setup:reduction}
Fix a qubit $\ObsCon$ instance $\{(O_i,y_i)\}_{i=1}^m$, see \Cref{def:obscon}, where each $O_i$ is a Pauli observable supported on $\Lambda_i\subseteq [n],\;1\le|\Lambda_i|\le k=O(1)$. Encode one qubit into the first two Fock levels of one bosonic mode by
\[
V\ket{0}=\ket{0}_{\text{Fock}},\;\; V\ket{1}=\ket{1}_{\text{Fock}}
\]
Next define the one-mode cutoff projector, the projector to the high Fock levels and the shifted number operator by
\[
P_1:=\ketbra00+\ketbra11,\;\;H_1:=I-P_1,\;\; W_j:=I+\widehat N_j
\]
For a local set of modes $\Lambda$, define
\[
V_\Lambda:=\bigotimes_{j\in\Lambda}V_j,\;\; P_\Lambda:=\bigotimes_{j\in\Lambda}P_1^{(j)},\;\; H_\Lambda:=I-P_\Lambda,\;\; W_\Lambda:=\bigotimes_{j\in\Lambda}W_j
\]
For each Pauli factor, use the polynomial CV representatives
\[
X_j\rightarrow \widetilde{X}_j:=a_j+a_j^\dagger=\sqrt2 q_j,\;\;Y_j\rightarrow \widetilde{Y}_j:= -i(a_j-a_j^\dagger)=\sqrt2 p_j,\;\;Z_j\rightarrow \widetilde{Z}_j:= I-2\widehat N_j=2I-q_j^2-p_j^2.
\]
Write $O_i=\bigotimes_{j\in\Lambda_i}\tau_{i,j}$,
with $\tau_{i,j}\in\{X,Y,Z\}$. Define its polynomial
CV representative by replacing each Pauli factor
with the corresponding operator introduced above:
\[
Q_i:=\bigotimes_{j\in\Lambda_i}\widetilde{\tau}_{i,j},\;\; \widetilde{\tau}_{i,j}\in\{\widetilde X,\widetilde Y,\widetilde Z\}.
\]
The encoded qubit observable is $E_i:=V_{\Lambda_i}O_iV_{\Lambda_i}^\dagger$ and by construction $P_{\Lambda_i}Q_iP_{\Lambda_i}=E_i$.\\
\\
We shift each target to zero and rescale so that
the cutoff block has operator norm at most one.
Define the physical CV observable by
\[
G_i:=\frac12(Q_i-y_iI)
\]
and define the $M=1$ cutoff block by
\[
F_i:=P_{\Lambda_i}G_iP_{\Lambda_i}=\frac12(E_i-y_iP_{\Lambda_i}).
\]
Moreover, $ \|F_i\|_\infty
\le \frac12(\|E_i\|_\infty+|y_i|)
\le 1$. The finite cutoff heterodyne recovery only sees $F_i$, while the physical CV observable is $G_i$.
\end{setup}
We first show that sufficiently large local heterodyne outcomes produce arbitrarily small single-shot estimates. This will allow us to construct records whose recovery values are simultaneously close to zero.
\begin{lemma}\label{lem:snapshotbound}
    Let $F$ be a finite-cutoff observable supported on a set of modes $\Lambda$, with $|\Lambda|=r$, and $z_{\Lambda}\in \mathbb{R}^{2r}$ denote the restriction of a heterodyne outcome to the coordinates in $\Lambda$. The single-shot estimator before numerical approximation,
$\widehat o_{F,\mathrm{het}}^{(M)}(z_\Lambda):=\Tr(F\widehat\rho_{z_\Lambda}^{(M),\mathrm{het}})$,
satisfies, for every $z_\Lambda\neq0$,
    \[
    |\widehat o_{F,\mathrm{het}}^{(M)}(z_\Lambda)|\le\frac{C_F}{\|z_\Lambda\|_2}
    \]
    where $C_F<\infty$ depends on $F$, the cutoff $M$, and the locality $r$ but not on $z_\Lambda$.
\end{lemma}
\begin{proof}
All sums below run over $a,b\in\{0,\ldots,M\}^r$. Define
\[
    g_{ab}^{\mathrm{het}}(\xi)
    :=
    \frac{1}{(2\pi)^r}
    \chi_{e_{ba}^{(r)}}(\xi)e^{\|\xi\|_2^2/4}.
\]
Each $g_{ab}^{\mathrm{het}}$ is smooth and compactly supported. By the definition of the snapshot,
\[
\widehat o_{F,\mathrm{het}}^{(M)}(z_\Lambda)
=
\sum_{a,b}F_{ba}
\int_{\mathbb R^{2r}}
g_{ab}^{\mathrm{het}}(\xi)e^{-iz_\Lambda^{\mathsf T}\Omega_\Lambda\xi}\,d\xi,
\]
where $F_{ba}:=\bra{b}F\ket{a}$. It therefore suffices to bound each integral in this sum.\\
\\
Fix $z_\Lambda\neq0$ and write $\zeta_\Lambda:=\Omega_\Lambda^{\mathsf T}z_\Lambda$. For the standard symplectic form, $\|\zeta_\Lambda\|_2=\|z_\Lambda\|_2$. Choose a coordinate $j$ for which $|(\zeta_\Lambda)_{j}|$ is largest. Since $\zeta_\Lambda$ has $2r$ coordinates,
    \[
    |(\zeta_\Lambda)_{j}|\ge\frac{\|z_\Lambda\|_2}{\sqrt {2r}}.
    \]
    Now we integrate by parts in this coordinate. Using
    \[
\frac{\partial}{\partial\xi_j}e^{-i\zeta_\Lambda^{\mathsf T}\xi}
=
-i(\zeta_\Lambda)_j e^{-i\zeta_\Lambda^{\mathsf T}\xi},
\]
we obtain
\[
\int_{\mathbb R^{2r}}
g_{ab}^{\mathrm{het}}(\xi)e^{-i\zeta_\Lambda^{\mathsf T}\xi}\,d\xi
=
\frac{1}{i(\zeta_\Lambda)_j}
\int_{\mathbb R^{2r}}
\frac{\partial g_{ab}^{\mathrm{het}}}{\partial\xi_j}(\xi)
e^{-i\zeta_\Lambda^{\mathsf T}\xi}\,d\xi.
\]
    The boundary term vanishes because $g_{ab}^{\mathrm{het}}$ has compact support. Taking absolute values gives
\[
\left|
\int_{\mathbb R^{2r}}
g_{ab}^{\mathrm{het}}(\xi)e^{-i\zeta_\Lambda^{\mathsf T}\xi}\,d\xi
\right|
\le
\frac{1}{|(\zeta_\Lambda)_j|}
\int_{\mathbb R^{2r}}
\left|\frac{\partial g_{ab}^{\mathrm{het}}}{\partial\xi_j}(\xi)\right|
\,d\xi
\le
\frac{C_{ab}^{(r)}}{\|z_\Lambda\|_2},
\]
where
\[
C_{ab}^{(r)}
:=
\sqrt{2r}\max_{1\le s\le2r}
\int_{\mathbb R^{2r}}
\left|\frac{\partial g_{ab}^{\mathrm{het}}}{\partial\xi_s}(\xi)\right|
\,d\xi.
\]
These constants are finite because the derivatives are
continuous and compactly supported. They are also independent
of $z_\Lambda$. Finally, applying these bounds to the finite sum defining
$\widehat o_{F,\mathrm{het}}^{(M)}$ yields
\[
|\widehat o_{F,\mathrm{het}}^{(M)}(z_\Lambda)|
\le
\frac{\sum_{a,b}|F_{ba}|C_{ab}^{(r)}}{\|z_\Lambda\|_2}
=
\frac{C_F}{\|z_\Lambda\|_2},
\qquad
C_F:=\sum_{a,b}|F_{ba}|C_{ab}^{(r)}<\infty.
\]
\end{proof}
We next make the decay bound uniform over all observables $F_i$.
Since $M=1$ and $r_i:=|\Lambda_i|\le k=O(1)$, the indices
$r,a,b$ range over a fixed finite set. The regularizers are
fixed by the recovery algorithm, so the corresponding constants
$C_{ab}^{(r)}$ have a finite maximum independent of the input size.
Moreover, each $F_i$ has $4^{r_i}\le4^k$ matrix entries,
all of magnitude at most one because $\|F_i\|_\infty\le1$.
Hence
\[
    C_{F_i}=\sum_{a,b\in\{0,1\}^{r_i}}
    |(F_i)_{ba}|C_{ab}^{(r_i)}\le 4^k
    \max_{\substack{1\le r\le k\\
                    a,b\in\{0,1\}^r}}
    C_{ab}^{(r)}
    =:C_k^{\mathrm{het}}.
\]
We now use this uniform decay bound to construct, in polynomial time, a heterodyne shadow whose recovery values are simultaneously close to zero.
\begin{corollary}\label{cor:zerorecovery}
Let $F_1,\ldots,F_m$, with $m=\poly(n)$, be finite-cutoff
observables with cutoff $M=1$, supported on nonempty sets
of at most $k=O(1)$ modes, and satisfying
$\|F_i\|_\infty\le1$.
For any $\eta\ge1/\poly(n)$ and prescribed shadow length
$L=\poly(n)$, one can construct in polynomial time an outcome
\[
    z_R:=((R,0),\ldots,(R,0))\in\mathbb R^{2n}
\]
such that the shadow $S_R$ consisting of $L$ copies of $z_R$
satisfies
\[
    \left|A_{\mathrm{het}}^{(1)}(S_R,i)\right|<\eta
    \quad \forall i\in[m].
\]
\end{corollary}

\begin{proof}
For $R\ge1$, define $z_R:=((R,0),\ldots,(R,0))\in\mathbb R^{2n}$. Since each support $\Lambda_i$ is nonempty, $\|z_{R,\Lambda_i}\|_2=R\sqrt{|\Lambda_i|}\ge R$. Let $\widetilde o_{F_i,\mathrm{het}}^{(1)}$ denote the numerical single-shot
estimator used by the recovery algorithm, with precision
chosen so that
\[
    \left|
    \widehat o_{F_i,\mathrm{het}}^{(1)}(z_{R,\Lambda_i})
    -\widetilde o_{F_i,\mathrm{het}}^{(1)}(z_{R,\Lambda_i})
    \right|
    \le\frac{\eta}{4}.
\]

Starting from $R=1$, evaluate all
$\widetilde o_{F_i,\mathrm{het}}^{(1)}(z_{R,\Lambda_i})$.
If every value has magnitude at most $\eta/2$, stop;
otherwise replace $R$ by $2R$ and repeat.
This procedure does not require knowing $C_k^{\mathrm{het}}$. By \Cref{lem:snapshotbound} and the uniform bound
$C_{F_i}\le C_k^{\mathrm{het}}$,
\[
    |\widehat o_{F_i,\mathrm{het}}^{(1)}(z_{R,\Lambda_i})|
    \le\frac{C_k^{\mathrm{het}}}{\|z_{R,\Lambda_i}\|_2}
    \le\frac{C_k^{\mathrm{het}}}{R}.
\]
Thus, once $R\ge4C_k^{\mathrm{het}}/\eta$, every exact value has magnitude at most $\eta/4$, and every numerical estimate has magnitude at most $\eta/2$. The procedure therefore terminates with $R\le\max\{1,8C_k^{\mathrm{het}}/\eta\}$.\\
\\
At termination, let $S_R$ consist of $L$ copies of $z_R$.
For each $i$, all single-shot estimates coincide, so
\[
    \left|A_{\mathrm{het}}^{(1)}(S_R,i)\right|
    =
    \left|\widetilde o_{F_i,\mathrm{het}}^{(1)}(z_{R,\Lambda_i})\right|
    \le\frac{\eta}{2}<\eta.
\]

Finally, $C_k^{\mathrm{het}}$ is independent of the input size and
$\eta^{-1}$ is polynomially bounded. Hence $R$ is
polynomially bounded and the search uses $O(\log n)$
iterations. Each iteration evaluates $m=\poly(n)$
estimators using the fixed polynomial-time recovery
routine. Since $L=\poly(n)$ as well, the entire
construction takes polynomial time.
\end{proof}
For soundness, we must also consider bosonic witnesses
that are not supported on the encoded qubit subspace.
The next lemma uses the moment promise to bound the
weight outside this subspace on each mode.
\begin{lemma}\label{lem:momentbound}
Let $\sigma\in\mathcal D(\mathcal H_n)$. Suppose that, with $W_j:=I+\widehat N_j$,
\[
    \Tr(W_j^\ell\sigma)\le B
    \quad\text{for every }j\in[n].
\]
Let $P_1:=\ketbra00+\ketbra11$ and $H_1:=I-P_1$.
Then the probability of finding at least two photons
in any given mode $j$ satisfies
\[
    \Tr\!\left(H_1^{(j)}\sigma\right)
    \le\frac{B}{3^\ell}.
\]
\end{lemma}
\begin{proof}
Fix $j\in[n]$. Both $W_j^\ell$ and $H_1^{(j)}$
are diagonal in the Fock basis of mode $j$.
On levels $m=0,1$, $H_1^{(j)}$ vanishes.
On every level $m\ge2$, it has eigenvalue one,
while $W_j^\ell$ has eigenvalue
$(1+m)^\ell\ge3^\ell$. Hence
\[
    W_j^\ell\succeq3^\ell H_1^{(j)}.
\]
Taking expectations in $\sigma$ and applying the
moment bound gives
\[
    3^\ell\Tr\!\left(H_1^{(j)}\sigma\right)
    \le\Tr(W_j^\ell\sigma)
    \le B.
\]
Dividing by $3^\ell$ proves the claim.
\end{proof}

\begin{corollary}\label{cor:localmomentbound}
Under the assumptions of \Cref{lem:momentbound},
let $\Lambda\subseteq[n]$ with $|\Lambda|\le k$.
Then
\[
    \Tr(H_\Lambda\sigma)
    \le \frac{|\Lambda|B}{3^\ell}
    \le \frac{kB}{3^\ell}.
\]
Here $H_\Lambda=I-P_\Lambda$ projects onto the subspace
in which at least one mode in $\Lambda$ has two or more
photons.
\end{corollary}

\begin{proof}
Since the one-mode projectors commute,
\[
H_\Lambda=I-\prod_{j\in\Lambda}P_1^{(j)}\preceq\sum_{j\in\Lambda}(I-P_1^{(j)})=\sum_{j\in\Lambda}H_1^{(j)}.
\]
Taking expectations in $\sigma$ and applying \Cref{lem:momentbound} gives
\[
\Tr(H_\Lambda\sigma)\le\sum_{j\in\Lambda}\Tr(H_1^{(j)}\sigma)\le\frac{|\Lambda|B}{3^\ell}\le\frac{kB}{3^\ell}.
\]
\end{proof}
Since $Q$ may be unbounded, a small weight outside the
encoded subspace does not by itself guarantee a small
error in its expectation value. The next lemma uses
the moment bound to control the difference between
the expectations of $Q$ and its encoded block $E=PQP$.
\begin{lemma}\label{lem:QE-close}
Let $\Lambda\subseteq[n]$ satisfy $|\Lambda|\le k$, and let
\[
Q=\bigotimes_{j\in\Lambda}A_j,
\quad
A_j\in
\left\{
I,\,
a_j+a_j^\dagger,\,
-i(a_j-a_j^\dagger),\,
I-2\widehat N_j
\right\}.
\]
Let $\sigma\in\mathcal D(\mathcal H_n)$ satisfy
\[
    \Tr(W_j^\ell\sigma)\le B
    \quad\forall j\in\Lambda,
\]
where $\ell\ge k$.
Set $P:=P_\Lambda$ and let $E:=PQP$ be the encoded block of $Q$. Then
\[
\left|\Tr\big((Q-E)\sigma\big)\right|
\le \varepsilon_\ell=
2\cdot 3^k\sqrt{\frac{kB}{3^\ell}}
+
2^k kB\,3^{k-\ell}.
\]
\end{lemma}

\begin{proof}[Proof sketch]
Set $H:=I-P$. We decompose the difference as
\[
    Q-E=PQH+HQP+HQH.
\]
The first two terms couple the encoded subspace to its
complement. Cauchy-Schwarz, together with
$\|PQ\|_\infty\le3^k$ and the leakage bound from
\Cref{cor:localmomentbound}, gives
\[
    |\Tr(PQH\sigma)|+|\Tr(HQP\sigma)|
    \le
    2\cdot3^k\sqrt{\frac{kB}{3^\ell}}.
\]

To control the remaining term, we use the weighted bound
$\|W_\Lambda^{-1/2}QW_\Lambda^{-1/2}\|_\infty\le2^k$.
Since the range of $H$ contains only configurations with
at least two photons in some mode, the moment assumption
and $\ell\ge k$ yield
\[
    |\Tr(HQH\sigma)|
    \le 2^k\Tr(HW_\Lambda H\sigma)
    \le 2^k kB\,3^{k-\ell}.
\]
Adding these bounds gives the claimed $\varepsilon_\ell$.
The full proof is given in \Cref{sapp:QE-close}.
\end{proof}

\begin{theorem}\label{thm:Bosonic_shadows_hard}
  $\CSV^{\ell,B, M}_{\mathrm{HetCV}}$ is $\QMA$-hard even when $M=1$, $\ell=O(\log n)$, and $B=2^\ell$.  
\end{theorem}
\begin{proof}
Start from an instance $\{(O_i,y_i)\}_{i=1}^m$
of the $\QMA$-hard restriction of qubit $\ObsCon$
to constant-local Pauli observables, with
$m=\poly(n)$ and thresholds $\alpha,\beta$,
as established in \cite[Appendix~A]{KRMEG25}.
Set $\Delta:=\beta-\alpha\ge1/\poly(n)$.
The output observable family is $\{G_i\}_{i=1}^m$,
defined in \Cref{setup:reduction}. These observables satisfy the required weighted norm bound. In particular, in \Cref{lem:QE-close} we got
\[
\left\|
W_{\Lambda_i}^{-1/2}
Q_i
W_{\Lambda_i}^{-1/2}
\right\|_\infty
\le 2^k.
\]
Now since
\[
G_i=\frac12(Q_i-y_iI),
\qquad
|y_i|\le1,
\]
and $W_{\Lambda_i}\succeq I$, we obtain
\[
\left\|
W_{\Lambda_i}^{-1/2}
G_i
W_{\Lambda_i}^{-1/2}
\right\|_\infty
\le
\frac12(2^k+1).
\]
Since $k=O(1)$, the observables $G_i$ belong to the promised weighted-norm observable class.
We also need to construct the heterodyne measurement data. Apply \Cref{cor:zerorecovery} to the finite-cutoff observables $F_i$. By the fixed-recovery convention and $\|F_i\|_\infty\le1$, we have $C_{F_i}\le C_k^{\mathrm{het}}$, where $C_k^{\mathrm{het}}$ is finite and independent of $n$ and of the source instance. Set $\eta:=\Delta/16\ge1/\poly(n)$. Then \Cref{cor:zerorecovery} constructs in polynomial time a heterodyne shadow $S_R:=\{z_R,\ldots,z_R\}$ such that $\left|A_{\mathrm{het}}^{(1)}(S_R,G_i)\right|<\eta,\;\forall i$.\\
\\
Set the remaining parameters as follows:
\[
\alpha_{\mathrm{CV}}:=\frac{\alpha}{2}+\eta,\quad\beta_{\mathrm{CV}}:=\frac{\alpha}{2}+\frac{\Delta}{4}
\]
Thus $\;\beta_{\mathrm{CV}}-\alpha_{\mathrm{CV}}
=\frac{3\Delta}{16}
\ge\frac{1}{\poly(n)}$.\\
\\
Finally choose an integer $\ell\ge k$ and set $B:=2^\ell$ so that
\[
\frac{kB}{3^\ell}=k\left(\frac{2}{3}\right)^\ell\le\frac\Delta8,\quad \varepsilon_\ell\le\frac\Delta8.
\]
Since $k=O(1)$ and $\Delta\ge1/\poly(n)$, such a choice has
$\ell=O(\log n)$, and consequently $B=\poly(n)$.\\
\\
\textbf{Completeness}: Assume that the qubit $\ObsCon$ instance is a YES instance. Then there exists an $n$-qubit state $\rho$ such that $|\Tr(O_i\rho)-y_i|\le\alpha\;\forall i$. Now the prover can encode this into the $CV$ Hilbert space: $\sigma:=(V^{\otimes n})\rho (V^{\otimes n})^ \dagger$. This state lies entirely in the first two Fock levels, so $\sigma\in \mathcal D_{\ell,B}(\mathcal H_n)$ is a legitimate proof state. Let us now compute the expectation of $G_i$:
\begin{align}
    \begin{aligned}
        \left|\Tr(G_i\sigma_{\Lambda_i})\right|&=\frac12\left(\left|\Tr(Q_i\sigma_{\Lambda_i})-y_i\right|\right)=\frac12\left(\left|\Tr(P_{\Lambda_i}Q_iP_{\Lambda_i}\sigma_{\Lambda_i})-y_i\right|\right)\\
        &=\frac12\left(\left|\Tr(E_iV_{\Lambda_i}\rho_{\Lambda_i} V_{\Lambda_i}^\dagger)-y_i\right|\right)=\frac12\left(\left|\Tr(O_i\rho)-y_i\right|\right)\le\frac\alpha2.
    \end{aligned}
\end{align}
By construction $\left|A_{\mathrm{het}}^{(1)}(S_R,i)\right|\le \eta$. Hence
\[
\left|\Tr(G_i\sigma_{\Lambda_i})-A_{\mathrm{het}}^{(1)}(S_R,i)\right|\le \frac\alpha2+\eta=\alpha_{\mathrm{CV}}\;\;\forall i.
\]
So the constructed $\CSV^{\ell,B,1}_{\mathrm{HetCV}}$ instance is a YES instance.\\
\\
\textbf{Soundness}: Assume the qubit $\ObsCon$ instance is a NO instance. Then by the soundness condition $\forall$ $n$-qubit states $\rho$ there exists a constraint $i$ such that $|\Tr(O_i\rho)-y_i|\ge\beta$. Let's assume now for the sake of contradiction that there exists a state $\sigma\in \mathcal D_{\ell,B}(\mathcal H_n)$ such that $\left|\Tr(G_i\sigma_{\Lambda_i})-A_{\mathrm{het}}^{(1)}(S_R,i)\right|< \beta_{\mathrm{CV}}\;\forall i$. Now because by construction we have that $\left|A_{\mathrm{het}}^{(1)}(S_R,i)\right|\le \eta$, we get $\left|\Tr(G_i\sigma_{\Lambda_i})\right|<\beta_{\mathrm{CV}}+\eta$ and so $\left|\Tr(Q_i\sigma_{\Lambda_i})-y_i\right|<2(\beta_{\mathrm{CV}}+\eta)$.\\
\\
From \Cref{lem:momentbound}, for each mode we have $\Tr((I-P_1^{(j)})\sigma)\le\frac{B}{3^\ell}$. Therefore, on the local support $|\Lambda_{i}|\le k$, by a union bound we get, $\Tr((I-P_{\Lambda_i})\sigma)\le\frac{kB}{3^\ell}$.\\
\\
For each mode $j$, define the decoding channel
\[
\mathcal D_j(\omega):=V_j^\dagger P_1^{(j)}\omega P_1^{(j)}V_j+\Tr\left((I-P_1^{(j)})\omega\right)\ketbra{0}{0}.
\]
This map is completely positive and trace preserving: it preserves the
$\ket0,\ket1$ block and maps every higher Fock level to the qubit state
$\ket0$. Define the global qubit state
\[
\rho:=(\mathcal D_1\otimes\cdots\otimes\mathcal D_n)(\sigma).
\]
For an observable $O_i$ supported on $\Lambda_i$, the
adjoint decoded observable is block diagonal with respect
to $P_{\Lambda_i}$ and $I-P_{\Lambda_i}$. It agrees with
$E_i$ on the first block and has operator norm at most
one on the second. Therefore
\[
\left|
\Tr(O_i\rho)-\Tr(E_i\sigma_{\Lambda_i})
\right|\le\Tr\left((I-P_{\Lambda_i})\sigma\right)
\le\sum_{j\in\Lambda_i}
\Tr\left((I-P_1^{(j)})\sigma\right)
\le\frac{kB}{3^\ell}.
\]
By \Cref{lem:QE-close} we have $|\Tr(E_i\sigma_{\Lambda_i})-\Tr(Q_i\sigma_{\Lambda_i})|\le\varepsilon_\ell$.\\
\\
Then adding everything together gives us:
\[
|\Tr(O_i\rho)-y_i|<2(\beta_{\mathrm{CV}}+\eta)+\varepsilon_\ell+\frac{kB}{3^\ell}
\]
With the choice of our parameters: $|\Tr(O_i\rho)-y_i|< \alpha+\frac{7\Delta}{8}<\alpha+\Delta=\beta$. Since this holds for every $i$, it contradicts the NO condition of
the source instance for the decoded state $\rho$. Hence the constructed instance is a NO instance.

\end{proof}
\subsection{Homodyne shadows}\label{sscn:homodyne}

\label{subsec:homodyne-hardness}

We now extend the hardness result to local homodyne shadows.
The encoding, observable construction, and moment estimates
from the preceding subsection also apply here.
The only protocol-dependent step to replace is the
construction of records whose recovery values are
simultaneously close to zero.

\begin{definition}[Homodyne CV classical shadow]
\label{def:CShomo}
We adapt the classical-shadow framework of
\cref{def:classical-shadow} to the homodyne protocol
described in \cref{par:HetCV}.
For a fixed Fock cutoff $M$, a homodyne CV classical
shadow on $n$ modes is a tuple
$(S,O,A_{\mathrm{hom}}^{(M)},\chi)$, where
\begin{itemize}
    \item The shadow $S=\{s_t\}_{t=1}^{L}$ consists of
    $L=\poly(n)$ records
    \[
      s_t=\bigl((\theta_{t,j},x_{t,j})\bigr)_{j=1}^{n},
        \quad
        \theta_{t,j}\in[-\pi,\pi],
        \quad
        x_{t,j}\in\mathbb R,
    \]
    with each phase and outcome represented using
    at most $\poly(n)$ bits.

    \item $O=\{O_i\}_{i=1}^{m}$ is a family of
    $m=\poly(n)$ Hermitian observables with
    polynomial-bit descriptions.
    Each $O_i$ is supported on a set
    $\Lambda_i\subseteq[n]$ of at most $k=O(1)$ modes,
    is a constant-degree polynomial in the canonical
    observables, and satisfies the weighted
    photon-number bound specified in
    \cref{par:HetCV}, \Cref{eqn:norm}.

    \item $A_{\mathrm{hom}}^{(M)}$ is the cutoff-$M$
    recovery algorithm specified in \cref{par:HetCV};
    given $S$ and $i\in[m]$, it returns the estimate
    $A_{\mathrm{hom}}^{(M)}(S,i)$ for $O_i$ to
    $\chi$ bits of precision, in polynomial time.
\end{itemize}
\end{definition}

\begin{definition}[$\CSV^{\ell,B,M}_{\mathrm{HomCV}}$]
Given a homodyne CV classical shadow $(S,O,A_{\mathrm{hom}}^{(M)},\chi)$, as in \Cref{def:CShomo}, thresholds
$0\le\alpha_{\mathrm{CV}}<\beta_{\mathrm{CV}}$, with $\Delta_{\mathrm{CV}}:=\beta_{\mathrm{CV}}-\alpha_{\mathrm{CV}}\ge1/\poly(n)$, and moment parameters $\ell,B$ with $\ell\ge k$, decide between the following two cases:
\begin{itemize}

\item\textbf{YES}: $\exists \sigma\in\mathcal D_{\ell,B}(\mathcal H_n)$ such that $\left|\Tr(O_i\sigma)-A_{\mathrm{hom}}^{(M)}(S,i)\right|\le \alpha_{\mathrm{CV}}
\;\;\forall i$.

\item\textbf{NO}: $\forall \sigma\in\mathcal D_{\ell,B}(\mathcal H_n), \exists i$
such that $\left|\Tr(O_i\sigma)-A_{\mathrm{hom}}^{(M)}(S,i)\right|\ge \beta_{\mathrm{CV}}$.
\end{itemize}
\end{definition}

\begin{corollary}
\label{cor:homodyne-hard}
$\CSV^{\ell,B,M}_{\mathrm{HomCV}}$ is $\QMA$-hard
even when $M=1$, $\ell=O(\log n)$, and $B=2^\ell$.
\end{corollary}

\begin{proof}
We establish the homodyne counterpart of
\Cref{cor:zerorecovery} and then apply the reduction
from \Cref{thm:Bosonic_shadows_hard}.
Use the physical observables $G_i$ and their cutoff
blocks $F_i=P_{\Lambda_i}G_iP_{\Lambda_i}$ from
\Cref{setup:reduction}.
Recall that
$1\le r_i:=|\Lambda_i|\le k$ and
$\|F_i\|_\infty\le1$.\\
\\
For $a,b\in\{0,1\}$, define
\[
    g_{ab}^{\mathrm{hom}}(u):=\frac{\widetilde f_{ab}(u)}{\pi}.
\]
By the Fourier representation of the homodyne pattern
functions given in the preliminaries,
\[
    f_{ab}(x)
    =
    \frac12\int_{\mathbb R}g_{ab}^{\mathrm{hom}}(u)e^{iux}\,du.
\]
The Laguerre-polynomial formula in the preliminaries gives,
for $a\le b$,
\[
    g_{ab}^{\mathrm{hom}}(u)
    =
    |u|e^{-u^2/4}
    \underbrace{
        (-i)^{b-a}
        \sqrt{\frac{2^{a-b}a!}{b!}}\,
        u^{b-a}
        L_a^{(b-a)}\!\left(\frac{u^2}{2}\right)
    }_{=:p_{ab}(u)}.
\]
For $a>b$, set $p_{ab}:=p_{ba}$, using the symmetry
of the pattern functions.
Since $p_{ab}$ is a polynomial, it is continuous at zero.
Moreover, the Gaussian decays faster than any polynomial
grows. Hence
\[
    \lim_{u\to0^\pm}g_{ab}^{\mathrm{hom}}(u)
    =0,
    \quad
    \lim_{u\to\pm\infty}g_{ab}^{\mathrm{hom}}(u)=0.
\]
On each of the intervals $(-\infty,0)$
and $(0,\infty)$, $g_{ab}^{\mathrm{hom}}$ is smooth, and both
$g_{ab}^{\mathrm{hom}}$ and its derivative are absolutely integrable,
being polynomial multiples of a Gaussian. Since the factor $|u|$ may prevent differentiability
at zero, we apply integration by parts separately on
these two intervals, taking one-sided limits at zero.\\
\\
We now evaluate the pattern functions at a positive quadrature outcome $x=R>0$. Applying integration by parts on the two intervals and taking the endpoint limits gives
\[
\begin{aligned}
    \int_{-\infty}^{0}g_{ab}^{\mathrm{hom}}(u)e^{iuR}\,du
    &=
    \left[\frac{g_{ab}^{\mathrm{hom}}(u)e^{iuR}}{iR}\right]_{-\infty}^{0^-}
    -
    \frac{1}{iR}
    \int_{-\infty}^{0}(g_{ab}^{\mathrm{hom}})'(u)e^{iuR}\,du,\\
    \int_{0}^{\infty}g_{ab}^{\mathrm{hom}}(u)e^{iuR}\,du
    &=
    \left[\frac{g_{ab}^{\mathrm{hom}}(u)e^{iuR}}{iR}\right]_{0^+}^{\infty}
    -
    \frac{1}{iR}
    \int_{0}^{\infty}(g_{ab}^{\mathrm{hom}})'(u)e^{iuR}\,du.
\end{aligned}
\]
By the endpoint limits established above, all boundary terms vanish and so we get:
\[
f_{ab}(R)=-\frac{1}{2iR}\int_{\mathbb R}(g_{ab}^{\mathrm{hom}})'(u)e^{iuR}\,du,
\]
where $(g_{ab}^{\mathrm{hom}})'(u)$ denotes the derivative for $u\ne0$;
its value at zero does not affect the integral. Taking absolute values and using $|e^{iuR}|=1$, we obtain
\[
|f_{ab}(R)|\le\frac{1}{2R}\int_{\mathbb R}|(g_{ab}^{\mathrm{hom}})'(u)|\,du.
\]
Define
\[
    D:=\max\left\{1,\,
        \frac12\max_{a,b\in\{0,1\}}
        \int_{\mathbb R}|(g_{ab}^{\mathrm{hom}})'(u)|\,du
    \right\}.
\]
Each integral is finite, and the four functions $g_{ab}^{\mathrm{hom}}$ are fixed independently of the input instance. Hence $D$ is finite and independent of the instance. Thus
\[
    |f_{ab}(R)|\le\frac{D}{R}
    \qquad
    \forall a,b\in\{0,1\}.
\]
Choose the homodyne record
\[
    s_R:=((0,R),\ldots,(0,R)),
    \quad R\ge1.
\]
By the tensor-product formula for the homodyne snapshot,
each of its matrix entries on $\Lambda_i$ has magnitude
at most $(D/R)^{r_i}$.
Moreover, $F_i$ has $4^{r_i}$ matrix entries, each of
magnitude at most one.
Define the exact single-shot estimator for $G_i$ by
\[
    \widehat o_{i,\mathrm{hom}}^{(1)}(R)
    :=
    \Tr\left(
        F_i
        \widehat\rho_{s_{R,\Lambda_i}}^{(1),\mathrm{hom}}
    \right).
\]
The preceding entrywise bounds give
\[
    \left|\widehat o_{i,\mathrm{hom}}^{(1)}(R)\right|
    \le
    4^{r_i}\left(\frac{D}{R}\right)^{r_i}
    \le
    \frac{C_k^{\mathrm{hom}}}{R},
    \qquad
    C_k^{\mathrm{hom}}:=(4D)^k.
\]
Here we used $1\le r_i\le k$, $D\ge1$, and $R\ge1$.
In particular, $C_k^{\mathrm{hom}}$ is independent of
$n$ and of the source instance.\\
\\
For any $\eta\ge1/\poly(n)$, we can now use the same
doubling search as in \Cref{cor:zerorecovery}.
Let $\widetilde o_{i,\mathrm{hom}}^{(1)}(R)$ be the numerical estimator used
by the recovery algorithm, with precision chosen so that
\[
    |\widehat o_{i,\mathrm{hom}}^{(1)}(R)-\widetilde o_{i,\mathrm{hom}}^{(1)}(R)|\le\frac{\eta}{4}.
\]
Starting from $R=1$, double $R$ until
$|\widetilde o_{i,\mathrm{hom}}^{(1)}(R)|\le\eta/2$ for every $i$.
The decay bound guarantees termination once
$R\ge4C_k^{\mathrm{hom}}/\eta$, so the returned value satisfies
\[
    R\le
    \max\left\{1,\frac{8C_k^{\mathrm{hom}}}{\eta}\right\}.
\]
Since $C_k^{\mathrm{hom}}$ is constant and $1/\eta$ is polynomially bounded, the search uses
$R=\poly(n)$ and $O(\log n)$ iterations.
At these values of $R$, each iteration evaluates $m=\poly(n)$ estimators to accuracy $\eta/4$ in polynomial time. Thus the search runs in polynomial time.\\
\\
For any prescribed $L=\poly(n)$, let $S_R$ consist of $L$ copies of $s_R$. Since recovery averages identical numerical single-shot
estimates, the stopping criterion gives
\[
\left|A_{\mathrm{hom}}^{(1)}(S_R,i)\right|=\left|\widetilde o_{i,\mathrm{hom}}^{(1)}(R)\right|\le\frac{\eta}{2}<\eta.
\]
We now apply the reduction from
\Cref{thm:Bosonic_shadows_hard}.
Set $\eta=\Delta/16$, where $\Delta=\beta-\alpha$
is the source-instance promise gap, and choose
\[
    \alpha_{\mathrm{CV}}:=\frac{\alpha}{2}+\eta,
    \qquad
    \beta_{\mathrm{CV}}:=\frac{\alpha}{2}+\frac{\Delta}{4}.
\]
Choose $\ell=O(\log n)$ and $B=2^\ell$ as in that proof.
Its completeness and soundness arguments depend on
the records only through the bound
$|A_{\mathrm{hom}}^{(1)}(S_R,i)|<\eta$.
They therefore apply to $S_R$ with the same observables,
encoding, decoding channel, and moment estimates,
establishing the claimed hardness.
\end{proof}

\section{Complexity of All-Pauli Observable Consistency}\label{scn:allpauli}

The following version of $\ObsCon$ is for the classical shadow scheme of~\cite{kingTriplyEfficientShadow2025} for the observable set of all $n$-qubit Pauli strings.

\begin{definition}[$\pauliObsCon$]\label{def:allpauliobscon}
Let $\mathcal P_n:=\{I,X,Y,Z\}^{\otimes n}$ be the set of all $n$-qubit Pauli strings. The input consists of target expectation values $\{y_P\}_{P\in\mathcal P_n}$, for which we assume succinct access, and parameters $\alpha,\beta$ satisfying $\beta-\alpha\ge 1/\poly(n)$. We assume $y_P\in[-1,1]$ and $0\le\alpha<\beta\le2$. Decide between the following cases:
\begin{itemize}
    \item \textbf{Yes}: $\exists$ an $n$-qubit state $\rho$ such that $\forall P\in\mathcal P_n\;\;|\Tr(P\rho)-y_P|\le\alpha$.
    \item \textbf{No}: $\forall$ $n$-qubit states $\rho$, $\exists P\in\mathcal P_n$ such that $|\Tr(P\rho)-y_P|\ge\beta$.
\end{itemize}
\end{definition}

\noindent We begin by establishing two hardness results for $\pauliObsCon$: $\coNP$-hardness and $\QMA$-hardness.

\subsection{\texorpdfstring{$\pauliObsCon$ is $\mathrm{coNP}$-hard}{AllPauliObsConexp is coNP-hard}}\label{sscn:coNPhard}
\begin{definition}[$\mathrm{UNSAT}$]
Given a Boolean formula $\phi$ in conjunctive normal form on $n$ variables, decide between the following cases:
\begin{itemize}
    \item \textbf{YES:} $\phi$ is unsatisfiable, meaning
    $\phi(x)=0$ for every $x\in\{0,1\}^n$.
    \item \textbf{NO:} $\phi$ is satisfiable, meaning
    $\phi(x)=1$ for some $x\in\{0,1\}^n$.
\end{itemize}
\end{definition}
\begin{theorem}[\cite{AB09}]\label{thm:UNSAT_coNP}
    $\mathrm{UNSAT}$ is $\mathrm{coNP}$-complete. 
\end{theorem}
\begin{lemma}\label{lem:AllPauli_is_coNP}
    $\text{UNSAT}\le\pauliObsCon$.
\end{lemma}
\begin{proof}
    Let $\phi(a)$ be a Boolean formula on $n$-bits $a=(a_1,\dots ,a_n)\in\{0,1\}^n$. We construct an all-Pauli $\ObsCon$ instance on $N=n+2$ qubits as follows:\\
    \\
    For $u=(u_1,\ldots,u_n)\in\{0,1\}^n$, define $Z_u:=Z^{u_1}\otimes\cdots\otimes Z^{u_n}$, where $Z^0=I$ and $Z^1=Z$. Now for each assignment $a\in\{0,1\}^n$ define
    \[
    A:=Z\otimes I\otimes I^{\otimes n},\;\; B_a:=I\otimes Z\otimes Z_a,\;\;AB_a=Z\otimes Z\otimes Z_a.
    \]
    For every Pauli string $P\in\mathcal{P}_N$ define $y_P$ as follows: $y_I=1,\;\;y_A=1$.\\
    \\
    For each satisfying assignment $a$, i.e., $\phi(a)=1$ we set: $y_{B_a}=1,\;\;y_{AB_a}=-1$\\
    \\
    Lastly, for every other Pauli string, set $y_P=0$.\\
    \\
    Given $P$, its target can be computed in polynomial time by checking its form and, if $P=B_a$ or $P=AB_a$, evaluating $\phi(a)$.\\
    \\
    Set the promise parameters to $\;\alpha:=0,\; \beta:=\frac12$.\\
    \\
    \textbf{Completeness}: Assume $\phi$ is unsatisfiable. Then there is no satisfying assignment $a$. So the only non-zero targets are $y_I=1,\;y_A=1$. Now consider the state $\rho:=\frac{I+A}{2^N}$. This is a valid density matrix since
    \begin{equation*}
        \Tr(\rho)=\frac{1}{2^N}\Tr(I)+\frac{1}{2^N}\Tr(A)=1,\quad\rho\ge0.
    \end{equation*}
    It is easy to see that this density matrix matches all target values exactly:
    \[
    \Tr(I\rho)=\Tr(\rho)=1,\;\;\Tr(A\rho)=\Tr\left(\frac{I+A}{2^N}\right)=\Tr(\rho)=1.
    \]
    And for every other Pauli string $P\neq I,A$ we get:
    \[
    \Tr(P\rho)=\Tr\left(P\frac{I+A}{2^N}\right)=\Tr\left(\frac{P+PA}{2^N}\right)=0.
    \]
    Therefore $\rho$ matches every target exactly.\\
    \\
    \textbf{Soundness}: Let $a$ be a satisfying assignment. Then by construction
    \[
    y_A=1,\;\; y_{B_a}=1,\;\; y_{AB_a}=-1.
    \]
    We show that no quantum state satisfies these even approximately within error $<1/2$. Because $A,B_a$ have eigenvalues $\pm1$ and they commute we get 
    \[ 
    (I-A)(I-B_a)\succeq0\iff AB_a\succeq A+B_a-I
    \]
    So for any $\rho$ we have
    \[
    \Tr(AB_a\rho)\ge\Tr(A\rho)+\Tr(B_a\rho)-1
    \]
    Now suppose for contradiction that
    \begin{align*}
        |\Tr(A\rho)-1|< \frac12&\implies \Tr(A\rho)>\frac12\\
        |\Tr(B_a\rho)-1|<\frac12 &\implies \Tr(B_a\rho)>\frac12
    \end{align*}
    And so $\Tr(AB_a\rho)>\frac12+\frac12-1=0$.\\
    \\
    But the target value for $AB_a$ is $-1$ and so
    \[
    |\Tr(AB_a\rho)-(-1)|=|\Tr(AB_a\rho)+1|>1.
    \]
    So $AB_a$ constraint is violated by more than $1/2$.   
\end{proof}
\begin{theorem}\label{thm:coNPhard}
 $\pauliObsCon$ is $\coNP$-hard.   
\end{theorem}
\begin{proof}
    This follows from \Cref{thm:UNSAT_coNP,lem:AllPauli_is_coNP}.
\end{proof}

\subsection{\texorpdfstring{$\pauliObsCon$ is $\QMA$-hard}{AllPauliObsConexp is $\QMA$-hard}}\label{sscn:QMAhard}
We first define the source problem for our reduction and establish its $\QMA$-hardness.
\begin{definition}[$\text{FullPauli}\ObsCon_{k}$]\label{def:FullPauliObsCon} The input is the set $\mathcal P_n^{[1,k]}:=\{P\in\mathcal P_n:1\le\operatorname{wt}(P)\le k\}$ of all the up to $k$-local Pauli observables along with their target expectation values $y_P\in[-1,1]$ and parameters $0\le\alpha_{\mathsf{FP}}<\beta_{\mathsf{FP}}\le2$, with $\beta_{\mathsf{FP}}-\alpha_{\mathsf{FP}}\ge1/\poly(n)$. The output is to decide between the following two cases:
\begin{itemize}
    \item \textbf{Yes}: $\exists \rho\;$such that $\;\forall P\in\mathcal P_n^{[1,k]},\;|\Tr(P\rho)-y_P|\le\alpha_{\mathsf{FP}}$
    \item \textbf{No}: $\forall\rho\;\;\exists P\in\mathcal P_n^{[1,k]}$ such that $|\Tr(P\rho)-y_P|\ge\beta_{\mathsf{FP}}.$
\end{itemize}
    
\end{definition}
\begin{proposition}\label{prop:fullpauli-hardness}
 $\text{FullPauli}\ObsCon_{2}$ is $\QMA$-hard even for $\alpha_{\mathsf{FP}}=0$ and $\beta_{\mathsf{FP}}\ge\frac{1}{\poly(n)}.
 $   
\end{proposition}
\begin{proof}
    We use the $\QMA$-hard family of \emph{exact} $2$-$\CLDM$ instances from \cite[Sections~4.4-4.5]{KR25}, in which a marginal $\rho_{ij}$ is supplied for every pair of qubits. First we check that the marginals agree on their one-qubit overlap. If not we output a fixed NO instance; otherwise let $\rho_i$ denote the common one-qubit marginal and for every Pauli string $P$ of weight one or two define
    \[
    y_P=\begin{cases}
        \Tr(P\rho_i),\;\;&\operatorname{supp}(P)=\{i\}\\
        \Tr(P\rho_{ij}), &\operatorname{supp}(P)=\{i,j\}.
    \end{cases}
    \]
    If the original instance is a YES instance, there exists a global state whose two-qubit marginals equal every $\rho_{ij}$ exactly. Therefore every Pauli target is satisfied exactly, and we set $\alpha_{\mathsf{FP}}=0$.\\
    \\
    For soundness, let $\sigma$ be any global state. In a NO instance, there is a pair $i,j$ such that 
    \[
    \|\sigma_{ij}-\rho_{ij}\|_1\ge\beta_{\mathsf{CLDM}}.
    \]
    Set $\Delta_{ij}:=\sigma_{ij}-\rho_{ij}$. Its two-qubit Pauli expansion is
    \[
    \Delta_{ij}=\frac{1}{4}\sum_{Q\in\mathcal P_2}\Tr(Q\Delta_{ij})Q.
    \]
    By the Pauli expansion and the triangle inequality,
    \[
    \begin{aligned}
    \|\Delta_{ij}\|_1&\le\frac{1}{4}\sum_{Q\in\mathcal P_2}|\Tr(Q\Delta_{ij})|\|Q\|_1
    =\sum_{Q\in\mathcal P_2}|\Tr(Q\Delta_{ij})|\le16\max_{Q\in\mathcal P_2}|\Tr(Q\Delta_{ij})|
    \end{aligned}
    \]
    Consequently, there exists $Q\in\mathcal P_2$ such that
\[
|\Tr(Q\Delta_{ij})|
\ge \frac{\beta_{\mathsf{CLDM}}}{16}>0.
\]
Since $\Tr(\Delta_{ij})=0$, this $Q$ cannot be the identity. Finally we get,
\[
\begin{aligned}
|\Tr(P\sigma)-y_P|
&=|\Tr(Q\sigma_{ij})-\Tr(Q\rho_{ij})|=|\Tr(Q\Delta_{ij})|\\&\ge \frac{\beta_{\mathsf{CLDM}}}{16}:=\beta_{\mathsf{FP}}.
\end{aligned}
\]
\end{proof}
We now introduce a gadget that suppresses the physical Pauli expectations associated with unknown high-weight data moments, allowing us to assign them zero targets, while keeping the prescribed low-weight moments recoverable.
\paragraph{The Pauli-spreading gadget.}
Fix an arbitrary positive \emph{even} integer $m$ and set $c:=2^{-m}$. The reduction will later choose $m$ as a function of the source soundness parameter $\beta_{\mathsf{FP}}$. For each data qubit $j\in[n]$, introduce a $2m$-qubit key register $K_j$. A value in this register is a pair $\kappa_j=(u_j,v_j)\in\mathbb F_2^m\times \mathbb F_2^m$. A complete key is therefore $\boldsymbol\kappa=(\kappa_1,\dots,\kappa_n)\in (\mathbb F_2^{2m})^n$.
Now define the matrix
\[
J:=\bigoplus_{r=1}^{m/2}
\begin{pmatrix}
   0&1\\1&1 
\end{pmatrix}
\in\mathbb F_2^{m\times m}.
\]
With this choice, $I_m$, $J$, and $I_m+J$ are invertible
over $\mathbb F_2$. As shown below, this property
allows us to suppress the encoded Pauli expectations
uniformly for all three data Paulis $X$, $Y$, and $Z$.\\
\\
We now use each key $(u,v)$ to choose a Pauli operator
that acts as a mask on the associated data qubit. Define the two bits
\[
a(u,v):=u^{\mathsf T}Jv,\quad d(u,v):=u^{\mathsf T}v,
\]
where arithmetic is modulo $2$. These bits determine
whether to apply $X$ and $Z$, respectively:
\[
W_{u,v}:=i^{a(u,v)d(u,v)}X^{a(u,v)}Z^{d(u,v)}.
\]
And so for a complete key we get $W_{\boldsymbol\kappa}:=\bigotimes_{j=1}^nW_{\kappa_j}$. The encoding chooses a complete key uniformly at random,
stores it in the key register, and applies the corresponding Pauli operator to each data qubit. Thus,
\[
\mathcal E(\rho):=2^{-2mn}\sum_{\boldsymbol\kappa\in(\mathbb F_2^{2m})^n}\ket{\boldsymbol\kappa}\bra{\boldsymbol\kappa}_K\otimes W_{\boldsymbol\kappa}\rho W^\dagger_{\boldsymbol\kappa}.
\]
The encoded state acts on $n$ data qubits and $2mn$ key qubits, for a total of $N=n(2m+1)$ qubits. Notice that because the encoding is a classical mixture over computational-basis key states, any Pauli observable containing $X$ or $Y$ on a key qubit has zero expectation in the encoded state.\\
\\
We therefore focus on $Z$-type Pauli observables on the key register. Each such observable assigns a sign $\pm1$ to every key value. When combined with a data Pauli, the resulting observable has an expectation equal to a signed average of the data Pauli's expectations in the states obtained for the different keys. We now show how these averages spread the original data Pauli expectation across many physical Pauli observables.\\
\\
To make this sign assignment explicit, consider a single
key register $K_j$. We label its $Z$-type Pauli observables by $s_j=(\xi_j,\zeta_j)\in\mathbb F_2^{2m}$. Define $Z^{s_j}:=Z^{\xi_j}\otimes Z^{\zeta_j}$, where $Z^{\xi_j}:=Z^{(\xi_j)_1}\otimes\cdots\otimes Z^{(\xi_j)_m}$.
On a computational-basis key state, this observable acts as
\[
Z^{s_j}\ket{u_j,v_j}=(-1)^{\xi_j^{\mathsf T}u_j+\zeta_j^{\mathsf T}v_j}\ket{u_j,v_j}.
\]
For $A\in\{I,X,Y,Z\}$, let $f_A(u_j,v_j)\in\{\pm1\}$
denote the sign induced by conjugation with $W_{u_j,v_j}$:
\[
W_{u_j,v_j}^\dagger A W_{u_j,v_j}=f_A(u_j,v_j)A.
\]
Define the coefficient $\lambda_A(\xi_j,\zeta_j)$ by
averaging the product of this sign and the eigenvalue
of $Z^{s_j}$ over all key values:
\[
\lambda_A(\xi_j,\zeta_j):=
2^{-2m}\sum_{u_j,v_j\in\mathbb F_2^m}
(-1)^{\xi_j^{\mathsf T}u_j+\zeta_j^{\mathsf T}v_j}
f_A(u_j,v_j).
\]
The following lemma shows that
$|\lambda_A(\xi_j,\zeta_j)|=c$
for every $A\in\{X,Y,Z\}$ and every key-observable
label $(\xi_j,\zeta_j)$.
\begin{lemma}\label{lem:coefficients}
Define $M_X:=I_m$, $M_Z:=J$, and $M_Y:=I_m+J$.
For every $A\in\{X,Y,Z\}$, the matrix $M_A$ is invertible
over $\mathbb F_2$, and
$f_A(u_j,v_j)=(-1)^{u_j^{\mathsf T}M_Av_j}$
for all $u_j,v_j\in\mathbb F_2^m$.
Furthermore, for every $A\in\{X,Y,Z\}$ and every
$\xi_j,\zeta_j\in\mathbb F_2^m$,
\[
|\lambda_A(\xi_j,\zeta_j)|=c,
\quad
\lambda_I(\xi_j,\zeta_j)=
\begin{cases}
1,&\xi_j=\zeta_j=0,\\
0,&\text{otherwise}.
\end{cases}
\]
\end{lemma}
\begin{proof}
The phase in $W_{u_j,v_j}$ cancels under conjugation.
Using the Pauli commutation relations, we obtain
\[
\begin{aligned}
W_{u_j,v_j}^\dagger XW_{u_j,v_j}
&=(-1)^{d(u_j,v_j)}X
 =(-1)^{u_j^{\mathsf T}v_j}X,\\
W_{u_j,v_j}^\dagger ZW_{u_j,v_j}
&=(-1)^{a(u_j,v_j)}Z
 =(-1)^{u_j^{\mathsf T}Jv_j}Z,\\
W_{u_j,v_j}^\dagger YW_{u_j,v_j}
&=(-1)^{a(u_j,v_j)+d(u_j,v_j)}Y
 =(-1)^{u_j^{\mathsf T}(I_m+J)v_j}Y.
\end{aligned}
\]
Hence $f_A(u_j,v_j)=(-1)^{u_j^{\mathsf T}M_Av_j}$
for every $A\in\{X,Y,Z\}$. Moreover, each $2\times2$
block of $J$ and $I_m+J$ has determinant $1$ over
$\mathbb F_2$, so these matrices, as well as $I_m$,
are invertible.\\
\\
Substituting the expression for $f_A$ into the
definition of $\lambda_A$ gives
\[
\begin{aligned}
\lambda_A(\xi_j,\zeta_j)
&=2^{-2m}\sum_{u_j,v_j\in\mathbb F_2^m}
(-1)^{\xi_j^{\mathsf T}u_j+\zeta_j^{\mathsf T}v_j
+u_j^{\mathsf T}M_Av_j}\\
&=2^{-2m}\sum_{u_j\in\mathbb F_2^m}
(-1)^{\xi_j^{\mathsf T}u_j}
\sum_{v_j\in\mathbb F_2^m}
(-1)^{(M_A^{\mathsf T}u_j+\zeta_j)^{\mathsf T}v_j}.
\end{aligned}
\]
For any $t\in\mathbb F_2^m$,
\[
\sum_{v_j\in\mathbb F_2^m}(-1)^{t^{\mathsf T}v_j}
=
\begin{cases}
2^m,&t=0,\\
0,&t\neq0,
\end{cases}
\]
since all terms equal $1$ when $t=0$, while they
cancel in pairs when $t\neq0$.\\
\\
Thus the inner sum is nonzero only when
$M_A^{\mathsf T}u_j+\zeta_j=0$. Since $M_A$ is
invertible, this has the unique solution
$u_j=M_A^{-\mathsf T}\zeta_j$. Consequently,
\[
\lambda_A(\xi_j,\zeta_j)
=2^{-2m}\,2^m
(-1)^{\xi_j^{\mathsf T}M_A^{-\mathsf T}\zeta_j}
=2^{-m}(-1)^{\xi_j^{\mathsf T}M_A^{-\mathsf T}\zeta_j},
\]
and therefore $|\lambda_A(\xi_j,\zeta_j)|=2^{-m}=c$.\\
\\
Finally, for $A=I$, we have $f_I(u_j,v_j)=1$, so
\[
\lambda_I(\xi_j,\zeta_j)
=2^{-2m}
\left(\sum_{u_j\in\mathbb F_2^m}
(-1)^{\xi_j^{\mathsf T}u_j}\right)
\left(\sum_{v_j\in\mathbb F_2^m}
(-1)^{\zeta_j^{\mathsf T}v_j}\right).
\]
Applying the same cancellation identity to each factor
gives $1$ when $\xi_j=\zeta_j=0$ and $0$ otherwise.
\end{proof}
We extend the notation to the full system. For
$\mathbf s=(s_1,\dots,s_n)$, where
$s_j=(\xi_j,\zeta_j)\in\mathbb F_2^{2m}$, define
$Z_K(\mathbf s):=\bigotimes_{j=1}^n Z^{s_j}$.
For a data Pauli string $P=P_1\otimes\cdots\otimes P_n$,
define
\[
\lambda_P(\mathbf s):=
\prod_{j=1}^n\lambda_{P_j}(\xi_j,\zeta_j).
\]
\begin{corollary}\label{cor:coefficients}
Let $P\in\mathcal P_n$ and $w:=\operatorname{wt}(P)$. Then
\[
\max_{\mathbf s}|\lambda_P(\mathbf s)|=c^w,
\quad
\sum_{\mathbf s}\lambda_P(\mathbf s)^2=1,
\quad
\sum_{\mathbf s}|\lambda_P(\mathbf s)|=c^{-w}.
\]
\end{corollary}

\begin{proof}
 At each nonidentity data position, every one of the $2^{2m}$ labels $s_j$ gives a nonzero coefficient of magnitude $c$. At each identity position only $s_j=0$ survives and that contributes a factor $1$. Consequently, $\lambda_P(\mathbf s)$ is nonzero for exactly $2^{2mw}=c^{-2w}$ values of $\mathbf s$ and every nonzero coefficient has magnitude $\left|\lambda_P(\mathbf s)\right|=c^w$. The maximum is therefore $c^w$, and
\[
\sum_{\mathbf s}\lambda_P(\mathbf s)^2
=c^{-2w}c^{2w}=1,
\qquad
\sum_{\mathbf s}|\lambda_P(\mathbf s)|
=c^{-2w}c^w=c^{-w}.
\]
\end{proof}
The next lemma gives an exact expression for the
Pauli expectations of the encoded state.
\begin{lemma}\label{lem:encoded-expectations}
    For every data state $\rho$, every data Pauli $P$, and every $\mathbf s$ we have
    \[
    \Tr[(Z_K(\mathbf s)\otimes P)\mathcal E(\rho)]=\lambda_P(\mathbf s)\Tr(P\rho).
    \]
\end{lemma}
\begin{proof}
By the definition of $\mathcal E$,
\[
\begin{aligned}
\Tr[(Z_K(\mathbf s)\otimes P)\mathcal E(\rho)]
&=2^{-2mn}\sum_{\boldsymbol\kappa}
(-1)^{\sum_j(\xi_j^{\mathsf T}u_j+\zeta_j^{\mathsf T}v_j)}
\Tr[PW_{\boldsymbol\kappa}\rho W_{\boldsymbol\kappa}^\dagger]\\
&=2^{-2mn}\sum_{\boldsymbol\kappa}
(-1)^{\sum_j(\xi_j^{\mathsf T}u_j+\zeta_j^{\mathsf T}v_j)}
\left(\prod_{j=1}^n f_{P_j}(u_j,v_j)\right)\Tr(P\rho)\\
&=\lambda_P(\mathbf s)\Tr(P\rho).
\end{aligned}
\]
The second equality follows from cyclicity of the trace
and the conjugation identities, while the last follows
by factoring the sum over the individual key blocks
and using the definition of $\lambda_P(\mathbf s)$.
\end{proof}
The stored key allows the Pauli mask to be undone.
Define the decoding map by
\[
\mathcal D(\omega):=\sum_{\boldsymbol\kappa}W^\dagger_{\boldsymbol\kappa}\bra{\boldsymbol\kappa}\omega\ket{\boldsymbol\kappa}_KW_{\boldsymbol\kappa}.
\]
Operationally, $\mathcal D$ measures the key register in the computational basis, applies the inverse Pauli mask associated with the observed key, and discards the key register.\\
\\
It is straightforward to see that $\mathcal D(\mathcal E(\rho))=\rho$.

\begin{theorem}\label{thm:QMAhard}
$\pauliObsCon$  is $\QMA$-hard.
\end{theorem}

\begin{proof}
By \Cref{prop:fullpauli-hardness}, it suffices to reduce
from $n$-qubit $\text{FullPauli}\ObsCon_2$ instances
whose parameters satisfy
\[
0\le\alpha_{\mathsf{FP}}<\beta_{\mathsf{FP}}\le2,\quad
\alpha_{\mathsf{FP}}\le
\left(\frac{\beta_{\mathsf{FP}}}{32}\right)^2,\quad
\beta_{\mathsf{FP}}-\alpha_{\mathsf{FP}}
\ge\frac{1}{\poly(n)}.
\]
Choose $m$ to be the smallest positive even integer
such that $2^{-m}\le\beta_{\mathsf{FP}}/8$, and set
$c:=2^{-m}$. The output instance acts on
$N=n(2m+1)$ qubits and has thresholds
\[
\alpha_{\mathrm{AP}}:=c^3,\qquad
\beta_{\mathrm{AP}}:=\frac{\beta_{\mathsf{FP}}c^2}{2}
\]
with a succinct target function
$x:\mathcal P_N\rightarrow[-1,1]$ defined below.
Recall the encoding $\mathcal E$, the coefficients
$\lambda_P(\mathbf s)$, and the decoder $\mathcal D$
defined earlier. Every output Pauli can be uniquely
written as $Q=A_K\otimes P$, where $A_K$ acts on the
key register and $P$ acts on the data register. Define
\[
x(A_K\otimes P)=
\begin{cases}
0,
&A_K\;\text{contains}\;X\;\text{or}\;Y,\\
\delta_{\mathbf s,0},
&A_K=Z_K(\mathbf s),\ P=I,\\
\lambda_P(\mathbf s)y_P,
&A_K=Z_K(\mathbf s),\ 1\le\operatorname{wt}(P)\le2,\\
0,
&A_K=Z_K(\mathbf s),\ \operatorname{wt}(P)\ge3.
\end{cases}
\]
Since $m$ is the smallest positive even integer such that
$2^{-m}\le\beta_{\mathsf{FP}}/8$, we have
$\beta_{\mathsf{FP}}/32<c\le\beta_{\mathsf{FP}}/8$ and
$m=O(\log(8/\beta_{\mathsf{FP}}))$.
As $\beta_{\mathsf{FP}}\ge1/\poly(n)$, it follows that
$m=O(\log n)$, and hence $N=n(2m+1)=\poly(n)$.
Moreover, given an output Pauli $Q=A_K\otimes P$,
the value $x(Q)$ can be computed in polynomial time
by inspecting the key part $A_K$, computing
$\operatorname{wt}(P)$, and evaluating
$\lambda_P(\mathbf s)
=\prod_{j=1}^n\lambda_{P_j}(\xi_j,\zeta_j)$
using the closed form from \Cref{lem:coefficients}.
Finally,
\[
\beta_{\mathrm{AP}}-\alpha_{\mathrm{AP}}
=c^2\left(\frac{\beta_{\mathsf{FP}}}{2}-c\right)
\ge3c^3.
\]
Since $c>\beta_{\mathsf{FP}}/32$ and
$\beta_{\mathsf{FP}}\ge1/\poly(n)$, the promise gap
is inverse polynomial in $N$. Moreover,
\[
\alpha_{\mathsf{FP}}
\le\left(\frac{\beta_{\mathsf{FP}}}{32}\right)^2
<c^2.
\]
 \textbf{Completeness:} Assume that the source instance is
a YES instance. Then there exists an $n$-qubit state $\rho$
such that
\[
|\Tr(P\rho)-y_P|\le\alpha_{\mathsf{FP}},
\quad \forall P\in\mathcal P_n^{[1,2]}.
\]
We claim that the encoded state $\omega:=\mathcal E(\rho)$
is a valid witness for our output instance.\\
\\
\textbf{Case 1}: The key part contains $X$ or $Y$.
This gives us
\[
\Tr[(A_K\otimes P)\mathcal E(\rho)]
=
2^{-2mn}\sum_{\boldsymbol\kappa}
\bra{\boldsymbol\kappa}A_K\ket{\boldsymbol\kappa}
\Tr(PW_{\boldsymbol\kappa}\rho
W_{\boldsymbol\kappa}^{\dagger})
=0,
\]
because the key register is diagonal in the computational
basis. By definition, for this type of $Q$ we get $x(Q)=0$,
so the target is matched exactly:
$|\Tr(Q\omega)-x(Q)|=0$.\\
\\
\textbf{Case 2}: $A_K=Z_K(\mathbf s)$.
Let $w:=\operatorname{wt}(P)$. By \Cref{lem:encoded-expectations},
\[
\Tr[(Z_K(\mathbf s)\otimes P)\omega]
=\lambda_P(\mathbf s)\Tr(P\rho).
\]
Moreover, \Cref{cor:coefficients} gives
$|\lambda_P(\mathbf s)|\le c^w$.
We distinguish the following cases according to $w$.\\
\\
\textbf{Case 2a}: $P=I$. Then
\[
\Tr[(Z_K(\mathbf s)\otimes I)\omega]
=\lambda_I(\mathbf s)\Tr(I\rho)
=\delta_{\mathbf s,0}
=x(Z_K(\mathbf s)\otimes I).
\]
\textbf{Case 2b}: $1\le w\le2$. The target is
$x(Z_K(\mathbf s)\otimes P)=\lambda_P(\mathbf s)y_P$.
Thus,
\[
\begin{aligned}
|\Tr[(Z_K(\mathbf s)\otimes P)\omega]
-x(Z_K(\mathbf s)\otimes P)|&=
|\lambda_P(\mathbf s)|\,|\Tr(P\rho)-y_P|\\
&\le c^w\alpha_{\mathsf{FP}}
\le c^3=\alpha_{\mathrm{AP}}.
\end{aligned}
\]
\textbf{Case 2c}: $w\ge3$. Here the target is defined
to be zero, and
\[
|\Tr[(Z_K(\mathbf s)\otimes P)\omega]|
=|\lambda_P(\mathbf s)|\,|\Tr(P\rho)|
\le c^w\le c^3=\alpha_{\mathrm{AP}}.
\]
Hence
$|\Tr[(Z_K(\mathbf s)\otimes P)\omega]
-x(Z_K(\mathbf s)\otimes P)|\le\alpha_{\mathrm{AP}}$
for all $P$, and so a YES instance is mapped to
a YES instance.\\
\\
\textbf{Soundness:} Assume that the source instance is a NO instance and so
\[
\forall\rho\;\exists P\in\mathcal P_n^{[1,2]}
\;\text{s.t.}\; |\Tr(P\rho)-y_P|\ge\beta_{\mathsf{FP}}.
\]
Suppose, toward contradiction, that there exists an $N$-qubit state $\omega$ satisfying
\[
\left|\Tr(Q\omega)-x(Q)\right|<\beta_{\mathrm{AP}},
\quad \forall Q\in\mathcal P_N.
\]
Define the decoded state $\tau:=\mathcal D(\omega)$.
For a data Pauli string $P=P_1\otimes\cdots\otimes P_n$,
define $f_P(\boldsymbol\kappa):=
\prod_{j=1}^n f_{P_j}(u_j,v_j)$.
By the definition of $\lambda_P(\mathbf s)$ and the
cancellation identity used in the proof of
\Cref{lem:coefficients},
\[
\begin{aligned}
&\sum_{\mathbf s}\lambda_P(\mathbf s)
\bra{\boldsymbol\kappa}Z_K(\mathbf s)
\ket{\boldsymbol\kappa}\\
&\quad=2^{-2mn}\sum_{\boldsymbol\kappa'}
f_P(\boldsymbol\kappa')
\sum_{\mathbf s}
\bra{\boldsymbol\kappa'}Z_K(\mathbf s)
\ket{\boldsymbol\kappa'}
\bra{\boldsymbol\kappa}Z_K(\mathbf s)
\ket{\boldsymbol\kappa}\\
&\quad=2^{-2mn}\sum_{\boldsymbol\kappa'}
f_P(\boldsymbol\kappa')\,2^{2mn}
\delta_{\boldsymbol\kappa',\boldsymbol\kappa}
=f_P(\boldsymbol\kappa).
\end{aligned}
\]
Write $\omega_{\boldsymbol\kappa}
:=\bra{\boldsymbol\kappa}\omega
\ket{\boldsymbol\kappa}_K$.
For any data Pauli $P$, the definition of the decoder gives
\[
\Tr(P\tau)
=\sum_{\boldsymbol\kappa}
\Tr[PW_{\boldsymbol\kappa}^\dagger
\omega_{\boldsymbol\kappa}W_{\boldsymbol\kappa}]
=\sum_{\boldsymbol\kappa}
f_P(\boldsymbol\kappa)\Tr[P\omega_{\boldsymbol\kappa}].
\]
Substituting the identity above gives
\[
\begin{aligned}
\Tr(P\tau)
&=\sum_{\mathbf s}\lambda_P(\mathbf s)
\sum_{\boldsymbol\kappa}
\bra{\boldsymbol\kappa}Z_K(\mathbf s)
\ket{\boldsymbol\kappa}
\Tr[P\omega_{\boldsymbol\kappa}]\\
&=\sum_{\mathbf s}\lambda_P(\mathbf s)
\Tr[(Z_K(\mathbf s)\otimes P)\omega].
\end{aligned}
\]
Now fix any source Pauli $P\in\mathcal P_n^{[1,2]}$,
and let $w:=\operatorname{wt}(P)$.
For every $\mathbf s$, the corresponding output target is
$x(Z_K(\mathbf s)\otimes P)=\lambda_P(\mathbf s)y_P$.
By \Cref{cor:coefficients},
$\sum_{\mathbf s}\lambda_P(\mathbf s)^2=1$
and $\sum_{\mathbf s}|\lambda_P(\mathbf s)|=c^{-w}$.
Therefore,
\[
\begin{aligned}
\left|\Tr(P\tau)-y_P\right|
&=\left|
\sum_{\mathbf s}\lambda_P(\mathbf s)
\Tr[(Z_K(\mathbf s)\otimes P)\omega]
-y_P\sum_{\mathbf s}\lambda_P(\mathbf s)^2
\right|\\
&=\left|
\sum_{\mathbf s}\lambda_P(\mathbf s)
\left(
\Tr[(Z_K(\mathbf s)\otimes P)\omega]
-x(Z_K(\mathbf s)\otimes P)
\right)
\right|\\
&\le\sum_{\mathbf s}
\left|
\lambda_P(\mathbf s)
\left(
\Tr[(Z_K(\mathbf s)\otimes P)\omega]
-x(Z_K(\mathbf s)\otimes P)
\right)
\right|\\
&<\beta_{\mathrm{AP}}\sum_{\mathbf s}|\lambda_P(\mathbf s)|
=c^{-w}\beta_{\mathrm{AP}}
\le c^{-2}\beta_{\mathrm{AP}}
=\frac{\beta_{\mathsf{FP}}}{2}
<\beta_{\mathsf{FP}}.
\end{aligned}
\]
This is a contradiction and so a NO instance is mapped
to a NO instance.
\end{proof}

\subsection{\texorpdfstring{$\pauliObsCon$ is contained in $\p^{\PP}$}{AllPauliObsConexp is contained in PPP}}\label{sscn:inPPP}

We now turn to upper bounds for
$\pauliObsCon$. The general problem $\ObsCon_{\exp}$ is
$\qc$-complete \cite{KRMEG25}, so its all-Pauli
restriction is trivially contained in $\qc$. The same work gives the upper bound
$\qc\subseteq\mathsf{PSPACE}$. In this section we improve this bound by showing $\qc\subseteq \p^{\PP}$.
\begin{theorem}\label{thm:qcinP^PP}
    $\qc\subseteq \p^{\PP}.$
\end{theorem}
\noindent The proof adapts the postselection-based halving argument
used by Aaronson to simulate quantum advice
\cite[Theorems~3.4 and~3.5]{Aar05}. We implement the search for a halving challenge using
a $\mathsf{PP}$ oracle.\\
\\
Fix $L\in\qc$ (see \Cref{def:qc-cs}). By the weak amplification of $\qc$
\cite[Lemma~3.13]{KRMEG25}, for any polynomial $r=r(n)$ we obtain a verifier $V_x$
with error $\varepsilon=2^{-r(n)}$,
acting on a $q$-qubit proof and an $m$-bit challenge,
where $q,m=\poly(n,r)$. The verifier satisfies
\[
\begin{aligned}
    x\in L_{\text{yes}}
    &\implies
    \exists\rho\;\;\text{such that}\;\;\forall y:
    \Pr[V_x(\rho,y)=1]\ge 1-\varepsilon,\\
    x\in L_{\text{no}}
    &\implies
    \forall \rho\;\;\exists y\;\;\text{such that}\;\;
    \Pr[V_x(\rho,y)=1]\le\varepsilon.
\end{aligned}
\]
Choose $r$ sufficiently large that
$\varepsilon\le 1/[8(q+2)]$, and set $T:=q+2$.\\
\\
For each challenge $y$, we perform a two-outcome projective
measurement on the witness register and a fresh workspace
initialized to $\ket{0}$: we run the verifier circuit,
measure the output qubit, and apply the inverse circuit.
Previously used workspaces are left untouched.\\
\\
Starting from $I/2^q$ on the proof register, let $p_b$
be the probability that all projective measurements in a sequence
$b=(y_1,\ldots,y_t)$ accept. Whenever $p_b>0$, let $\rho_b$
be the reduced state of the witness register conditioned on this event.
Since the next test uses a fresh workspace,
\[
p_{b\circ y}
=
p_b\,\Pr[V_x(\rho_b,y)=1].
\]
We call $y$ a halving challenge for $b$ if
 $\Pr[V_x(\rho_b,y)=1]<1/2$. Thus a halving challenge satisfies
\[
p_{b\circ y}<\frac12p_b.
\]
We now show that, whenever $p_b>0$, a halving challenge
can be found on NO instances using polynomially many
queries to a $\mathsf{PP}$ oracle.
\begin{lemma}\label{lem:PPoracle_proced}
There is a deterministic polynomial-time procedure with
access to a $\mathsf{PP}$ oracle that, given $x$ and a
sequence $b$ of at most $T$ challenges with $p_b>0$,
returns either a halving challenge for $b$ or $\perp$.
If $x\in L_{\mathrm{no}}$, the procedure always returns
a halving challenge.
\end{lemma}
\begin{proof}
Fix a sequence $b$ with $p_b>0$. We use $\mathsf{PQP}=\mathsf{PP}$ \cite[Theorem~7]{Watrous08}, where $\mathsf{PQP}$ is defined as $\mathsf{PP}$ with polynomial-time quantum computations in place of classical probabilistic ones. Hence a $PP$ oracle can compare the acceptance probabilities of two polynomial-size quantum circuits. Indeed, if $C_1$ and $C_2$ have acceptance probabilities $p_1$ and $p_2$, respectively, consider the circuit that with probability $1/2$ runs $C_1$ and accepts iff $C_1$ accepts, and with probability $1/2$ runs $C_2$ and accepts iff $C_2$ rejects. Its acceptance probability is 
\[
\frac{1}{2}+\frac{p_1-p_2}{2},
\]
which is greater than $1/2$ exactly when $p_1>p_2$.\\
\\
On NO instances, amplified soundness guarantees a challenge
$y$ with $\Pr[V_x(\rho_b,y)=1]\le\varepsilon\le1/8$,
whereas every non-halving challenge has acceptance
probability at least $1/2$. We separate these two ranges
by checking whether at most a quarter of the runs accept
in repeated independent executions of the verifier.\\
\\
Set $\delta:=2^{-(m+3)}$ and $K:=32(m+3)$.
Let $X_y$ count the acceptances in $K$ independent runs
of $V_x(\rho_b,y)$, each on a separate copy of $\rho_b$,
and define
\[
r_b(y):=\Pr[X_y\le K/4].
\]
Since
$\mathbb{E}[X_y]=K\Pr[V_x(\rho_b,y)=1]$,
Hoeffding's inequality gives
\[
\begin{aligned}
\Pr[V_x(\rho_b,y)=1]\le\frac18
&\implies
1-r_b(y)
\le e^{-2K(1/4-1/8)^2}
=e^{-K/32}\le\delta,\\
\Pr[V_x(\rho_b,y)=1]\ge\frac12
&\implies
r_b(y)
\le e^{-2K(1/2-1/4)^2}
=e^{-K/8}\le\delta.
\end{aligned}
\]
Consequently, $r_b(y)>\delta$ guarantees that $y$ is
a halving challenge.\\
\\
We now construct the circuits used to choose the bits
of $y$. For a prefix $u\in\{0,1\}^{\ell}$ with
$0\le\ell\le m$, define $C_{b,u}$ as follows.
It prepares $K$ independent copies of $I/2^q$ and
applies the projective measurements corresponding to
the sequence $b$ to each copy, rejecting if any
measurement rejects. It then chooses a uniformly random completion $z\in\{0,1\}^{m-\ell}$ and runs the verifier with the same challenge $uz$ on each of the $K$ witness registers. The circuit accepts if and only if at most $K/4$ of these verifier runs accept.\\
\\
All $K$ copies pass the previous measurements with
probability $p_b^K$. Conditioned on this event, the
witness registers have joint reduced state
$\rho_b^{\otimes K}$. Thus
\[
S_b(u):=\Pr[C_{b,u}\text{ accepts}]
=
p_b^K\cdot 2^{-(m-\ell)}
\sum_{z\in\{0,1\}^{m-\ell}}r_b(uz).
\]
Starting with $u=\varnothing$, use the $\mathsf{PP}$
oracle to compare $S_b(u0)$ and $S_b(u1)$, and append
the bit giving the larger value, breaking ties
arbitrarily. Since $C_{b,u}$ chooses the next bit
of the challenge uniformly,
\[
S_b(u)=\frac{S_b(u0)+S_b(u1)}{2}\le\max\{S_b(u0),S_b(u1)\}.
\]
We choose the bit attaining this maximum, so $S_b$
does not decrease. After $m$ steps, the resulting
challenge $y$ satisfies
\[
S_b(y)\ge S_b(\varnothing).
\]
Finally, use the $\mathsf{PP}$ oracle to check whether
\[
p_{b\circ y}<\frac12p_b.
\]
Return $y$ if this holds, and $\perp$ otherwise.
Since $p_b>0$ and
$p_{b\circ y}=p_b\,\Pr[V_x(\rho_b,y)=1]$,
the check succeeds exactly when
$\Pr[V_x(\rho_b,y)=1]<1/2$.
Thus every returned challenge is halving.\\
\\
It remains to show that the final check always succeeds
on NO instances. By amplified soundness, there exists
$y^*$ such that
$\Pr[V_x(\rho_b,y^*)=1]\le\varepsilon\le1/8$.
The first Hoeffding bound gives $r_b(y^*)\ge1-\delta$.
Using $S_b(y)=p_b^K r_b(y)$ and $p_b>0$, we obtain
\[
\begin{aligned}
r_b(y)=\frac{S_b(y)}{p_b^K}
\ge\frac{S_b(\varnothing)}{p_b^K}\ge 2^{-m}r_b(y^*)
\ge 2^{-m}(1-\delta)>\delta.
\end{aligned}
\]
The second Hoeffding bound therefore implies
$\Pr[V_x(\rho_b,y)=1]<1/2$.
Hence the final check succeeds, and the procedure
returns a halving challenge on every NO instance.\\
\\
The procedure makes $m+1$ queries to the $\mathsf{PP}$
oracle, and all queried circuits have polynomial size.
\end{proof}
\paragraph{Proof of \Cref{thm:qcinP^PP}} Initialize $b=\varnothing$, and repeat the following for at most $T=q+2$ rounds. First use the $\mathsf{PP}$ oracle to check whether $p_b>0$; if $p_b=0$, reject. Otherwise invoke \Cref{lem:PPoracle_proced}. If it returns $\perp$, accept. If it returns a challenge $y$, append it to the current sequence, and continue. If $T$ challenges are appended, reject.

\paragraph{Completeness:} Suppose $x\in L_{\text{yes}}$, and let $\rho^*$ be the amplified honest witness. Consider any sequence $b=(y_1,\dots,y_t)$ of at most $T$ challenges. Now apply the sequential projection bound of \cite{Gao15} on the proof register together with
all the fresh workspaces, initially in the zero
state. Each projector accepts this initial honest
witness with probability at least $1-\varepsilon$,
so the entire sequence is accepted with probability
at least $1-4t\varepsilon$. Recall now that we start with the maximally mixed state for which we have
\[
\frac{I}{2^q}\succeq\frac{\rho^*}{2^q}.
\]
Thus, for every fixed sequence of challenges, the
probability that all measurements accept satisfies
\[
p_b\ge 2^{-q}(1-4t\varepsilon).
\]
Now since $t\le T$ and the amplification is chosen so $4T\varepsilon\le1/2$, we obtain
\begin{equation}\label{eq:lowerboundYes}
p_b\ge2^{-q-1}.
\end{equation}
In particular, the simulation never rejects because of $p_b=0$.\\
\\
Suppose now, toward a contradiction, that the procedure from \Cref{lem:PPoracle_proced} returns a halving challenge in every one of the $T$ rounds. Since each returned challenge satisfies
\[
p_{b\circ y}<\frac12p_b,
\]
and $p_{\varnothing}=1$, after $T=q+2$ rounds we would have
\[
p_b<2^{-T}=2^{-q-2}
\]
contradicting \Cref{eq:lowerboundYes}. Therefore the procedure must return $\perp$ before $T$ challenges are appended, at which point the simulation accepts.

\paragraph{Soundness:} Now suppose $x\in L_{\text{no}}$. If at any stage $p_b=0$, the algorithm rejects. Otherwise, $p_b>0$, and \Cref{lem:PPoracle_proced} guarantees that the procedure never returns $\perp$: it always returns another challenge $y$ satisfying
\[
\Pr[V_x(\rho_b,y)=1]<\frac12.
\]
Hence the algorithm can never accept. If $p_b$ remains positive, it appends a challenge in every round and therefore rejects after $T$ rounds.\\
\\
Finally, $T=q+2$ is polynomial, and each round makes polynomially many $\mathsf{PP}$ queries by \Cref{lem:PPoracle_proced}. Hence the entire procedure runs in deterministic polynomial time with access to a $\mathsf{PP}$ oracle. Therefore $L\in \mathsf{P}^{\mathsf{PP}}$, proving $\qc\subseteq \mathsf{P}^{\mathsf{PP}}$.

\section{Faking classical shadows is hard}\label{scn:faking}
In this section, we show that even when a complete collection of
two-qubit marginals determines a unique global state, extending these
marginals to a requested three-qubit marginal can remain computationally
hard. More precisely, an efficient classical algorithm for producing a desired 3-qubit marginal from the set of 2-qubit marginals would give a $\mathsf{BPP}$ simulation of
$\mathsf{QCMA}$. The proof proceeds in two stages: we first use such an
extension algorithm to recover a classical witness, and then
use it again to simulate the quantum verification of that witness.

\begin{theorem}\label{thm:recover-verify}
    For an $n$-qubit density operator $\rho$, write $\rho_S=\operatorname{Tr}_{\overline{S}}\rho$ and
$D\coloneqq(\rho_{\{i,j\}})_{1\leq i<j\leq n}$.
Suppose a randomized classical algorithm $A(n,D,B)$ runs in polynomial
 time on every input and has the following guarantee: whenever the pair
list $D$ has a unique global completion $\rho$ \emph{among all density
operators}, and any $B\subseteq\{1,\ldots,n\}$ has size three, its output is an
$8\times8$ matrix $M$ satisfying
\begin{equation}
 \Pr[\|M-\rho_B\|_1\leq1/20]\geq2/3.
 \label{eq:A}
\end{equation}
Then $\mathsf{QCMA} \subseteq \mathsf{BPP}$.
\end{theorem}

\begin{proof}
    For $L=(L_{\mathrm{yes}},L_{\mathrm{no}})$, let $x \in L$ have a uniform polynomial (in $|x|$) sized verifier $V_x$, taking an $m$-bit classical witness and $a$ many $\lvert 0 \rangle$ auxiliary qubits, where $m,a = \operatorname{poly}(|x|)$. Write
the acceptance probability of a witness $w$ as $p_x(w)=\Pr[V_x\lvert w,0^a\rangle\text{ outputs }1]$.
Suppose each YES instance has exactly one $w^\star$ with $p_x(w^\star)=1$,
whereas on each NO instance $p_x(w)\leq1/3$ for every $w$.

\smallskip\noindent\textbf{1. Recover.}
From input $x$ alone, use \Cref{lem:recovery-verification}(i) to compute
$2$-qubit reduced density matrices (2-RDMs, which we also call `pair marginals') denoted by $D_{\mathrm R}(x)$ and 3-qubit indices $B_1,\ldots,B_m$. On YES instances the
pair list has a unique completion (a unique consistent global state) $\rho_{\mathrm R}$, where the subscript $R$ denotes the recovery phase. Let $(\rho_{\mathrm R})_{B_i}$ be the $3$-qubit RDM of $\rho_{\mathrm R}$ on 3 qubits indexed by the triple $B_i$.
For each $B_i$, define
$z_i=\operatorname{Tr}(Z^{\otimes3}(\rho_{\mathrm R})_{B_i})$ to be the expectation value of $Z^{\otimes 3}$ on $(\rho_{\mathrm R})_{B_i}$, where $Z=\operatorname{diag}(1,-1)$ is the usual Pauli $Z$ operator. Then, \Cref{lem:recovery-verification}(i) ensures that the estimate $z_i$ is s.t.
\begin{equation}
|z_i-(-1)^{w_i^\star}|\leq\frac{1}{6},
\label{eq:recover}
\end{equation}
where $w_i^{\star}$ is the $i$-th bit of the m-bit witness $w^{\star}$.
Since each $w_i^{\star} \in \{0,1\}$, it follows from \Cref{eq:recover} that
\begin{align}
w_i^\star=0&\Longrightarrow z_i\geq\frac56,\label{eq:recover-zero}\\
w_i^\star=1&\Longrightarrow z_i\leq-\frac56.\label{eq:recover-one}
\end{align}

Call $A(n,D_{\mathrm{R}}(x),B_i)$, with the construction's qubit count $n$, and set
$\widehat z_i=\operatorname{Re}\operatorname{Tr}(Z^{\otimes3}M)$. A successful call gives
\begin{equation}
 |\widehat z_i-z_i|\leq
 \|Z^{\otimes3}\|_\infty\|M-(\rho_{\mathrm R})_{B_i}\|_1\leq1/20.
\end{equation}
Thus $w_i^\star=0$ implies $\widehat z_i\geq5/6-1/20=47/60>0$,
while $w_i^\star=1$ implies $\widehat z_i\leq-47/60<0$.
Therefore, the $i$-th bit of the recovered candidate witness $\widehat w$ is $0$ for a nonnegative estimate and $1$ otherwise,
and is correct with probability at least $2/3$, by \Cref{eq:A}. For each bit, take the majority of $O(\log(m))$ independent trials.
An application of the Chernoff bound makes the failure probability at most $1/(10m)$ for each bit. 
So a union bound gives that the entire $m$-bit recovered witness $\widehat w$ is indeed $w^{\star}$ w.h.p. More concretely,  $\Pr[\widehat w=w^\star]\geq9/10$.

\smallskip\noindent\textbf{2. Verify.}
Fix $w=\widehat w$. By \Cref{lem:recovery-verification}(ii), compute
$D_{\mathrm V}(x,w)$ and the triple $B_{\mathrm{out}}$. For every $x,w$,
this pair list has a unique completion $\rho_{\mathrm V}(x,w)$,
whether or not the verifier accepts. Define
$z_{\mathrm{out}}=\operatorname{Tr}(Z^{\otimes3}(\rho_{\mathrm V}(x,w))_{B_{\mathrm{out}}})$.
If $\eta_{\mathrm{out}}$ is the verifier's final logical output-qubit
state, then outcome $1$ means acceptance and
\begin{equation}
\operatorname{Tr}(Z\eta_{\mathrm{out}})
=\Pr[0]-\Pr[1]=1-2p_x(w).
\label{eq:output-expectation}
\end{equation}
The estimate supplied by \Cref{lem:recovery-verification}(ii) is
\begin{equation}
|z_{\mathrm{out}}-(1-2p_x(w))|\leq\frac16.
\label{eq:verify}
\end{equation}
Substituting the two acceptance cases into this estimate gives
\begin{align}
p_x(w)=1&\Longrightarrow z_{\mathrm{out}}\leq-\frac56,\label{eq:verify-perfect}\\
p_x(w)\leq\frac13&\Longrightarrow z_{\mathrm{out}}\geq\frac16.\label{eq:verify-sound}
\end{align}

A successful call to $A(n,D_{\mathrm{V}}(x,w),B_{\mathrm{out}})$
therefore yields an empirical estimate $\widehat z_{\mathrm{out}}$ with the following guarantee
\begin{align}
p_x(w)=1&\Longrightarrow\widehat z_{\mathrm{out}}\leq-47/60<0,\notag\\
p_x(w)\leq1/3&\Longrightarrow\widehat z_{\mathrm{out}}\geq7/60>0.
\label{eq:estimated-output-signs}
\end{align}

The rule is to Accept when the estimate is negative. Unlike the recovery phase, since we simulate the logical $Z$ measurement on just a single qubit, it suffices to have a constant number of independent trials with fresh randomness which reduces majority vote error
to at most $1/10$. Note that this is a classical calculation on $A$'s output matrices,
not a quantum execution of $V_x$.

On YES instances the probability of recovering $w^\star$ and accepting
is at least $(9/10)^2=81/100>2/3$. On NO instances the verifier promise
gives $p_x(w)\leq1/3$ for \emph{every} candidate, so Verify rejects with
probability at least $9/10$.

Thus the two phases decide $L$ with bounded error under the stated
promise. We now use this procedure to decide an arbitrary
$\mathsf{QCMA}$ problem.

\smallskip\noindent\textbf{3. Reduction to the unique and perfectly accepting witness.}
By \cite[Theorem 1]{JKNN12}, every $\mathsf{QCMA}$ problem admits a
perfectly complete verifier, with soundness amplified to at most $1/3$.
For a YES instance, let
\begin{equation}
S=\{w\in\{0,1\}^m:p_x(w)=1\}
\end{equation}
be its nonempty set of perfectly accepting witnesses.

Apply the Valiant--Vazirani construction \cite[Theorem 2.4]{VV86} to $S$. Their construction samples
random vectors $r_1,\ldots,r_m\in\{0,1\}^m$ and considers
\begin{equation}
S_i=
\{w\in S:w\cdot r_1=\cdots=w\cdot r_i=0\},
\qquad i=1,\ldots,m.
\end{equation}
With probability at least $1/4$, there exists an $i$ for which
$|S_i|=1$.

For each $i$, modify the verifier so that it first checks
$w\cdot r_1=\cdots=w\cdot r_i=0$ and rejects if any check fails;
otherwise it runs the original verifier unchanged. The perfectly
accepting witnesses of this modified verifier are exactly the elements
of $S_i$. Hence, whenever $|S_i|=1$, the modified verifier has a unique
perfectly accepting witness. On a NO instance, the modification can
only decrease acceptance probabilities, so every witness is still
accepted with probability at most $1/3$.

Repeating the Valiant-Vazirani construction a constant number of times
makes the probability that no modified verifier satisfies the
unique-perfect-witness promise arbitrarily small. Run the preceding
Recover and Verify procedures on every modified verifier and accept if any
run accepts. Since all calls terminate in polynomial time and NO
instances preserve the soundness promise for every restriction, standard
amplification gives a bounded-error polynomial-time algorithm.
Therefore $\mathsf{QCMA}\subseteq\mathsf{BPP}$.

\end{proof}

\begin{lemma}[Witness recovery and verification]\label[lemma]{lem:recovery-verification}
Let $L=(L_{\mathrm{yes}},L_{\mathrm{no}})$ be a promise problem and let
$V_x$ be a uniformly generated polynomial-size quantum verifier for $x \in L$ that takes an $m$-bit classical witness and $a$ many
auxiliary qubits initialized to zero, where $m, a = \mathrm{poly}(|x|)$. Write
the acceptance probability of input $x$ on witness $w$ as $p_x(w)=\Pr[V_x|w,0^a\rangle\text{ outputs }1]$ and $Z$ be a Pauli Z operator,
$Z=\operatorname{diag}(1,-1)$. Further, define a global completion of a pair list to be a density operator whose two-qubit reductions are exactly that list.
There are deterministic classical polynomial-time procedures with the
following properties.

\begin{enumerate}
\item[(i)] \textbf{Witness recovery.}
Suppose every $x\in L_{\mathrm{yes}}$ has exactly one
$w^\star\in\{0,1\}^m$ with $p_x(w^\star)=1$. Given $x$ alone, the
procedure produces candidate pair data $D_{\mathrm R}(x)$ and triples
$B_1,\ldots,B_m$ of physical qubit indices. On YES instances these data
have a unique global completion $\rho_{\mathrm R}$ among all density
operators. With
$z_i=\operatorname{Tr}(Z^{\otimes3}(\rho_{\mathrm R})_{B_i})$,
\Cref{eq:recover} holds for every $i\in\{1,\ldots,m\}$.
On other inputs, no consistency or uniqueness guarantee is required.

\item[(ii)] \textbf{Verification.}
Given any $x$ and $w\in\{0,1\}^m$, the procedure produces pair data
$D_{\mathrm V}(x,w)$ and a triple $B_{\mathrm{out}}$ of physical qubit
indices. These data have a unique global completion
$\rho_{\mathrm V}(x,w)$ among all density operators. With
$z_{\mathrm{out}}=\operatorname{Tr}(Z^{\otimes3}(\rho_{\mathrm V}(x,w))_{B_{\mathrm{out}}})$,
\Cref{eq:verify} holds. No uniqueness or acceptance promise on $w$ is
needed for this part.
\end{enumerate}
\end{lemma}

\begin{proof}[{Proof of (i)}]

First, for the recovery phase, we use two encodings, both based on the Steane code, for different
purposes. Let $E_{\mathrm{2\text{-}priv}}$ denote one level of the
$[[7,1,3]]$ Steane encoding. This encoding will be used only to hide
the logical state from all two-qubit marginals while retaining its
logical $Z$-expectation in a suitable three-qubit marginal. We also use a fixed $k$-fold concatenated Steane encoding
$E_{\mathrm{Sim}}$, which maps one logical qubit to
$q=7^k$
physical qubits. The purpose of $E_{\mathrm{Sim}}$ is different:
it allows small physical marginals to be computed during an encoded
computation independently of the unknown logical input.

More precisely, an $s$-simulatable code has a deterministic classical
algorithm which, given a logical gate, a time within its fixed encoded
implementation, and a set of at most $s$ physical qubits, outputs the
reduced density matrix on those qubits independently of the encoded
logical state. See \cite[Definition 4.1]{BG22}. Sufficiently many
concatenations of the Steane code are $s$-simulatable for every fixed
$s$ \cite[Lemma 4.4]{BG22}. Below we will only require simulation on a
constant number of qubits, so we fix $k$ sufficiently large once and
for all. In particular, $q$ is a constant independent of $|x|$.

The encoder and decoder for $E_{\mathrm{Sim}}$ are explicit circuits of
size polynomial in $q$ \cite[Section 4]{BG22}. We only use their classical
gate descriptions: given the description of $V_x$, we can therefore
construct the complete encoded circuit, including all encoding and
decoding steps, and keep track of every physical qubit and every circuit
time in polynomial time.

We now record the properties of $E_{\mathrm{2\text{-}priv}}$ that will
later allow us to recover witness bits. For an arbitrary one-qubit
density operator $\sigma$, let $\tau=E_{\mathrm{2\text{-}priv}}\sigma E_{\mathrm{2\text{-}priv}}^\dagger$, choose Steane coordinates so that
$Z_1Z_2Z_3Z_4$ is a stabilizer \cite[Sec 5.5.1]{GSY19} and $Z^{\otimes7}$ is logical $Z$, giving the last two equalities in \Cref{eq:code}.

\begin{align}
  \tau_S&=I_S/2^{|S|}\text{ for }|S|\leq2, \quad E_{\mathrm{2\text{-}priv}}^\dagger Z_5Z_6Z_7E_{\mathrm{2\text{-}priv}}=Z, \quad
\operatorname{Tr}(Z^{\otimes3}\tau_{\{5,6,7\}})=\operatorname{Tr}(Z\sigma).
\label{eq:code}
\end{align}
The first equality in \Cref{eq:code} holds because the Steane code is a nondegenerate $[[7,1,3]]$ stabilizer code,
every reduction to fewer than three physical qubits is maximally mixed
\cite[Lemmas~21--22 and the discussion following Lemma~22]{GSY19}.
Now, apply $E_{\mathrm{2\text{-}priv}}$ to the one-qubit logical state $\sigma$,
and then encoding each of its seven output qubits separately with
$E_{\mathrm{Sim}}$, gives the concatenated state
\begin{equation}\label{eq:witness-enc}
\Sigma \coloneqq
  E_{\mathrm{Sim}}^{\otimes 7}
\left(
E_{\mathrm{2\text{-}priv}}
\sigma
E_{\mathrm{2\text{-}priv}}^\dagger
\right)
E_{\mathrm{Sim}}^{\dagger\otimes 7}.
\end{equation}
Decoding the $E_{\mathrm{Sim}}$ blocks encoding qubits $5$, $6$, and $7$
of the $E_{\mathrm{2\text{-}priv}}$ codeword gives $\left(
E_{\mathrm{2\text{-}priv}}
\sigma
E_{\mathrm{2\text{-}priv}}^\dagger
\right)_{\{5,6,7\}}$. So by \Cref{eq:code},
\begin{equation}
\operatorname{Tr}\left(
Z^{\otimes 3}
\left(
E_{\mathrm{2\text{-}priv}}
\sigma
E_{\mathrm{2\text{-}priv}}^\dagger
\right)_{\{5,6,7\}}
\right)
=
\operatorname{Tr}(Z\sigma).
\end{equation}

 We now attach a clock so that pair marginals can also constrain
the computation producing that readout. The following
argument is common to recovery and verification phases.

For a compiled circuit $C=U_T\cdots U_1$ on $N=r+b$ work qubits,
let the first $r$ qubits form the physical input register and initialize
the remaining $b$ qubits to zero. Write $P_t=U_t\cdots U_1$, $P_0=I$,
and $\lvert\widehat t\rangle=\lvert1^t0^{T-t}\rangle$. Define
\begin{align}
J_C = \frac{1}{\sqrt{T+1}}\sum_{t=0}^T
\left[P_t\left(I_{2^r}\otimes|0^b\rangle\right)\right]
\otimes|\widehat t\rangle,\qquad
\rho_C(\omega)\coloneqq J_C\omega J_C^\dagger,
\label{eq:history}
\end{align}
where $ n=N+T$, $I_{2^r}$ is the identity on the input register and $\omega$ is its
density operator. The recovery circuit has $r=7qm$. Note that the verification
circuit has $r=0$ and starts entirely from zeros. Notably, the clock is unencoded
in both constructions.

For the fixed circuit $C$, \cite[Lemma~4.4]{KR25} supplies, for every
inverse-polynomial accuracy $\varepsilon>0$, a two-local Hamiltonian $H$
on the same $N+T$ qubits. Every history state $\rho_C(\omega)$ has zero
expectation, and any density operator $\xi$ with
$\operatorname{Tr}(H\xi)\leq1$ is within trace norm $\varepsilon$ of a
history state of $C$. 
Suppose $\xi$ has the same pair marginals as $\rho_C(\omega)$.
For each positive integer $d$, apply that lemma with
$\varepsilon=n^{-d}$. Since the Hamiltonian is 2-local, we have that 
\begin{equation}
\operatorname{Tr}(H\xi)=\operatorname{Tr}(H\rho_C(\omega))=0.
\end{equation}
Consequently there is an input density operator $\omega_d$ such that
\begin{equation}
\|\xi-\rho_C(\omega_d)\|_1\leq n^{-d}.
\end{equation}
For this fixed circuit, $n\geq2$, so these histories converge to $\xi$.
The input density operators form a compact set, and
$\omega\mapsto\rho_C(\omega)$ is continuous. Its image is therefore
compact and closed, proving that $\xi$ is a history of the same circuit,
including when $\xi$ is mixed.

Thus, a competing completion has the form $\rho_C(\omega')$ for
some input density operator $\omega'$. For verification, where $r=0$,
this already fixes the history. For recovery, however, $\omega'$ is
still an arbitrary physical input. We must first constrain it to encode
an $m$-qubit logical witness and then require that witness to accept
perfectly. We start with the encoding condition, which will be enforced
through decoder outputs at specified circuit times.

To express these conditions using pairs, we need the following
single-time extraction identity.

Expand \Cref{eq:history} using
$\Omega_{u,v}=P_u(\omega\otimes|0^b\rangle\langle0^b|)P_v^\dagger$,
and put $\gamma_j(t)=(\Omega_{t,t})_{\{j\}}$, the reduced work-qubit
state at time $t$. The $t$-th clock qubit has global label $c_t=N+t$.
In the $|0\rangle\langle1|$ clock block of the pair $\{j,c_t\}$, only
$(u,v)=(t-1,t)$ survives: these are the only unary strings with the
specified values at $c_t$ and equal values at every discarded clock
qubit. Thus, at an inserted identity $U_t=I$,
\cite[Lemmas~4.5--4.6]{KR25} gives

\begin{equation}
 (T+1)(I_j\otimes\langle 0\rvert_{c_t})
       (\rho_C(\omega))_{\{j,c_t\}}(I_j\otimes\lvert 1\rangle_{c_t})
   =\operatorname{Tr}_{\mathrm{work}\setminus\{j\}}\Omega_{t-1,t}
   =\gamma_j(t),
 \label{eq:identity}
\end{equation}

Here $P_t=P_{t-1}$, so $\Omega_{t-1,t}=\Omega_{t,t}$.
To require $\gamma_j(t)=|0\rangle\langle0|$, the pair entry must obey
\begin{equation}
(I_j\otimes\langle0|_{c_t})
(D_{\mathrm R}(x))_{\{j,c_t\}}
(I_j\otimes|1\rangle_{c_t})
=\frac{|0\rangle\langle0|}{T+1}.
\label{eq:zero-pair}
\end{equation}
This is a constraint on an off-diagonal matrix block, not a measurement.
Below we calculate the entire pair table.

We now use \Cref{eq:zero-pair} on decoder's auxiliary outputs to verify that the input witness was encoded correctly.
The point is that a pure zero auxiliary state certifies an encoded
input without fixing its logical data.
Let $D_{\mathrm{Sim}}$ be the fixed unitary decoder for
$E_{\mathrm{Sim}}$. For the intended $\Sigma$ in
\Cref{eq:witness-enc}, and with the data outputs listed before the
decoder auxiliary outputs,
\begin{equation}\label{eq:decoder-constr-Sim}
D_{\mathrm{Sim}}^{\otimes 7}
\Sigma
D_{\mathrm{Sim}}^{\dagger\otimes 7}
=
\tau\otimes
|0^{7(q-1)}\rangle\langle0^{7(q-1)}|,
\end{equation}
and hence
\begin{equation}
\Sigma
=
E_{\mathrm{Sim}}^{\otimes 7}
\tau
E_{\mathrm{Sim}}^{\dagger\otimes 7}.
\end{equation}
Thus the condition $\gamma_j(t)=|0\rangle\langle0|$ for
auxiliary qubits of the decoder $D_{\mathrm{Sim}}$ certifies correct encoding under $E_{\mathrm{Sim}}$. Similarly, if $D_{\mathrm{2\text{-}priv}}$ denotes the decoder for
$E_{\mathrm{2\text{-}priv}}$, then
\begin{equation}\label{eq:decoder-constr-priv}
D_{\mathrm{2\text{-}priv}}
\tau
D_{\mathrm{2\text{-}priv}}^\dagger
=
\sigma\otimes|0^6\rangle\langle0^6|
\end{equation}
is equivalent to
\begin{equation}
\tau
=
E_{\mathrm{2\text{-}priv}}
\sigma
E_{\mathrm{2\text{-}priv}}^\dagger.
\end{equation}
Thus choosing the same $\gamma_j(t)=|0\rangle\langle0|$, now for the auxiliary qubits of the decoder $D_{\mathrm{2\text{-}priv}}$ certifies correct encoding under
$E_{\mathrm{2\text{-}priv}}$.

The encoding constraints leave an arbitrary logical witness state. Replace $V_x$ by
$\widetilde V_x$, which
retains a witness register $W$, copies its computational-basis value to a fresh register $W'$, and applies $V_x$ to $W'$ and its auxiliary qubits. Note that this is a fixed product of $m$ CNOT gates and does not require knowing the value of the $m$-bit witness.

Write an arbitrary $m$-qubit density matrix $\eta$ on register $W$ as
$\eta=\sum_{u,v \in \{0,1\}^m}\eta_{uv}\lvert u\rangle\langle v\rvert$  and let

$|\psi_w\rangle=V_x|w,0^a\rangle$ denote the final state on the
working witness register $W'$ and the verifier auxiliaries. The joint
state, with this auxiliary register included in the second factor, is
\begin{equation}
\sum_{u,v\in\{0,1\}^m}\eta_{uv}|u\rangle\langle v|_W
\otimes|\psi_u\rangle\langle\psi_v|.
\end{equation}
The acceptance projector acts on the second factor. Tracing out $W$
uses $\operatorname{Tr}(|u\rangle\langle v|)=\delta_{uv}$ and gives
\begin{equation}
\widetilde{p}_x(\eta)
\coloneqq
\sum_{w\in\{0,1\}^m}
\eta_{ww}
\langle \psi_w|\Pi_{\mathrm{acc}}|\psi_w\rangle
=
\sum_{w\in\{0,1\}^m}
\langle w|\eta|w\rangle p_x(w) =\sum_w\eta_{ww}p_x(w).
\label{eq:acceptance-diagonal}
\end{equation}

On a YES instance, if the pair data force $\widetilde p_x(\eta)=1$, then
\begin{equation}
 1-\widetilde p_x(\eta)= 0 = \sum_w\eta_{ww}(1-p_x(w)).
 \label{eq:diagonal}
\end{equation}

Since all summands are non-negative and due to uniqueness of $w^{\star}$, $(1-p_x(w)) > 0$ for $w \neq w^{\star}$ implies $\eta_{ww} = 0$. Furthermore, normalization of $\eta$ ensures that $\eta_{ w^{\star} w^{\star}} = 1$. Moreover, we show that the off-diagonals of $\eta$ also vanish. This is because positive semidefiniteness (PSD) of $\eta$ implies every $2\times2$ principal submatrix indexed by distinct $u,v\in\{0,1\}^m$ is PSD and hence has nonnegative determinant, which gives $|\eta_{uv}|^2\leq\eta_{uu}\eta_{vv}$. Thus, $\eta_{uv} = 0$ for all $u \neq v \in \{0,1\}^m$. Hence,
\begin{equation}\label{eq:eta-unique}
\eta=\lvert w^\star\rangle \langle w^\star\rvert
\end{equation}
is the only logical input accepted perfectly by $\widetilde V_x$.
It remains to implement this test on the certified physical register
and express perfect acceptance as a pair condition.

By the decoder constraints in \Cref{eq:decoder-constr-Sim,eq:decoder-constr-priv}, the physical witness register is
already forced to be correctly encoded under both layers.

At the circuit input, its logical state is some $m$-qubit density operator $\eta$, encoded as
\begin{equation}
E_{\mathrm{Sim}}^{\otimes 7m}
E_{\mathrm{2\text{-}priv}}^{\otimes m}
\eta
E_{\mathrm{2\text{-}priv}}^{\dagger\otimes m}
E_{\mathrm{Sim}}^{\dagger\otimes 7m}.
\end{equation}

Let $\eta_i$ denote the $i$th one-qubit marginal of $\eta$. After the
encoded computation is complete, decode, one block at a time, the
$E_{\mathrm{Sim}}$ blocks encoding coordinates $5$, $6$, and $7$ of
each retained witness codeword. Let $B_i$ be the three distinguished
data-output qubit indices, and leave those qubits untouched thereafter.
Their $Z^{\otimes3}$ expectation is $\operatorname{Tr}(Z\eta_i)$.
On the perfectly accepting history, where \Cref{eq:eta-unique} fixes
$\eta=|w^\star\rangle\langle w^\star|$, their final reduced state is
\begin{equation}
\left(
E_{\mathrm{2\text{-}priv}}
\eta_i
E_{\mathrm{2\text{-}priv}}^\dagger
\right)_{\{5,6,7\}}.
\end{equation}
Hence, by \Cref{eq:code},
\begin{equation}
\operatorname{Tr}\left(
Z^{\otimes 3}
\left(
E_{\mathrm{2\text{-}priv}}
\eta_i
E_{\mathrm{2\text{-}priv}}^\dagger
\right)_{\{5,6,7\}}
\right)
=
\operatorname{Tr}(Z\eta_i).
\end{equation}

It remains to force $\eta$ to be perfectly accepting. Decode the verifier output $E_{\mathrm{Sim}}$ block last and leave its data qubit $o$ unchanged afterwards. Let $c_T$ be the last clock qubit.
Since $c_T$ is in state $|1\rangle$ only at the final circuit time, the
two-qubit marginal on $o$ and $c_T$ satisfies
\begin{equation}
\operatorname{Tr}\left[
\left(
|0\rangle\langle0|_o
\otimes
|1\rangle\langle1|_{c_T}
\right)
(\rho_C(\omega))_{\{o,c_T\}}
\right]
=
\frac{1-\widetilde p_x(\eta)}{T+1}.
\end{equation}

The accepting-history pair table would therefore have
\begin{equation}
\operatorname{Tr}\left[
\left(
|0\rangle\langle0|_o
\otimes
|1\rangle\langle1|_{c_T}
\right)
(D_{\mathrm R}(x))_{\{o,c_T\}}
\right]
=
0.
\end{equation}
Any history state consistent with $D_{\mathrm R}(x)$ must then satisfy
$\widetilde p_x(\eta)=1$ and hence we get $\eta
=
|w^\star\rangle\langle w^\star|$ in \Cref{eq:eta-unique} to be the unique witness state.

It follows that $\eta_i=|w_i^\star\rangle\langle w_i^\star|$ and
$\operatorname{Tr}(Z\eta_i)=(-1)^{w_i^\star}$. By \Cref{eq:code},
\begin{equation}
\operatorname{Tr}\left(
Z^{\otimes 3}
\left(
E_{\mathrm{2\text{-}priv}}
\eta_i
E_{\mathrm{2\text{-}priv}}^\dagger
\right)_{\{5,6,7\}}
\right)
=
(-1)^{w_i^\star}.
\end{equation}

We have now identified the only possible completion of the recovery
pair data: a history with the correctly encoded input
$|w^\star\rangle\langle w^\star|$. This is a uniqueness argument,
but not a procedure for producing all the pair matrices. To finish
recovery, we must compute the whole table without knowing $w^\star$
and relate the final-time readout to the history state marginal. We do both
next, in a form that can also be used for the Verification phase.

\smallskip\noindent\textbf{Computation of the pair data for Recovery and Verification phases.}
We now compute the complete pair table, rather than only the entries
used to constrain its completion. The calculation below applies to construct the 2-RDM data for both Recovery and Verification phases.

Write $D$ for the pair table being constructed. For a requested two-qubit pair
$S$, let $S_{\mathrm{work}}$ and $S_{\mathrm{clock}}$ be its work and
clock qubits, and let $s=|S_{\mathrm{clock}}|$. Since the history state has separate work and clock registers, a requested 2-qubit marginal can be either on the work register, clock register or between the two registers. The cases when the marginal is fully on the work or fully on the clock
registers are straightforward. The cross-register marginal will be a bit
more involved. We therefore treat all three cases uniformly by splitting
the requested pair into its work and clock components.

Thus
$|S_{\mathrm{work}}|=2-s$ and $s\in\{0,1,2\}$ since a 2-qubit pair between work and clock registers can contain $0, 1,$ or $2$ clock qubits. Taking the partial trace
of the history expansion in \Cref{eq:history} gives
\begin{equation}
D_S=\frac{1}{T+1}\sum_{u,v=0}^{T}
(\Omega_{u,v})_{S_{\mathrm{work}}}\otimes
\operatorname{Tr}_{\mathrm{clock}\setminus S_{\mathrm{clock}}}
\left(|\widehat u\rangle\langle\widehat v|\right).
\label{eq:pair-production}
\end{equation}
The clock factor is known explicitly. Hence, we
calculate the work factor of each surviving term as follows.

First, for $u<v$, the unary strings in the clock register differ exactly at clock positions
$u+1,\ldots,v$. Tracing out any of these positions gives
$\operatorname{Tr}(|0\rangle\langle1|)=0$. Hence a term survives only
if $\{c_{u+1},\ldots,c_v\}\subseteq S_{\mathrm{clock}}$, which implies
$v-u\leq s$ since $|S_{\mathrm{clock}}|=s$. Define
\begin{align}
F&=S_{\mathrm{work}}\cup\bigcup_{k=u+1}^{v}\operatorname{supp}(U_k)
\end{align}
Since each $U_k$ is a 2-qubit gate and recalling that $|S_{\mathrm{work}}|=2-s$, we get that
\begin{equation}\label{eq:work-support}
  |F|\leq(2-s)+2(v-u)\leq2+s\leq4,
\end{equation}
where the penultimate and final inequalities follow due to $v-u \leq s$ and $|s|\leq 2$ respectively.

Here $S_{\mathrm{work}}$ is the set of work qubits retained in the
requested pair, not a density matrix. The larger set $F$ also contains
every work qubit on which one of $U_{u+1},\ldots,U_v$ acts.

For $u<v$, the identity $P_v=(U_v\cdots U_{u+1})P_u$ gives
$\Omega_{u,v}=\Omega_{u,u}(U_v\cdots U_{u+1})^\dagger$. By the definition
of $F$, the extra intervening gates $(U_v\cdots U_{u+1})^\dagger$ act as the identity on all work qubits
outside $F$. We may
therefore trace out the work qubits outside $F$ before multiplying by
this product, and then trace out $F\setminus S_{\mathrm{work}}$:
\begin{align}
(\Omega_{u,v})_{S_{\mathrm{work}}}
&=\operatorname{Tr}_{\mathrm{work}\setminus S_{\mathrm{work}}}
\left[\Omega_{u,u}(U_v\cdots U_{u+1})^\dagger\right]\notag\\
&=\operatorname{Tr}_{F\setminus S_{\mathrm{work}}}
\left[
\operatorname{Tr}_{\mathrm{work}\setminus F}(\Omega_{u,u})
(U_v\cdots U_{u+1})^\dagger
\right]\notag\\
&=\operatorname{Tr}_{F\setminus S_{\mathrm{work}}}
\left[(\Omega_{u,u})_F(U_v\cdots U_{u+1})^\dagger\right].
\label{eq:cross-time-calculation}
\end{align}
In the last two lines, the gate product is represented on $F$, acting
as the identity on any qubits of $F$ that it does not touch. 
Thus, once $(\Omega_{u,u})_F$ has been computed 
as we explain in the following paragraph, we form the known matrix
$(U_v\cdots U_{u+1})^\dagger$ on $F$, multiply on the right, and trace out
$F\setminus S_{\mathrm{work}}$. These operations use matrices of size
at most $16\times16$, by \Cref{eq:work-support}. 

If $u>v$, use
$(\Omega_{u,v})_{S_{\mathrm{work}}}
=((\Omega_{v,u})_{S_{\mathrm{work}}})^\dagger$.
For $u\neq v$ the resulting operator need not be a density matrix.
It is a contribution to the pair matrix, not a separate element of
$D$. 

Thus the only quantity that remains to be computed is
$(\Omega_{u,u})_F$, with $|F|\leq4$. During the encoded computation,
this is exactly the reduction onto the work qubits at time $t$ supplied by the
$E_{\mathrm{Sim}}$ simulator. For the temporary decoder checks described
above, the relevant logical input is already known: it is $I/2$ for an
initial $E_{\mathrm{Sim}}$ block check, by \Cref{eq:code}, and
$|0\rangle\langle0|$ for the auxiliary blocks of the
$E_{\mathrm{2\text{-}priv}}$ decoder. Their reductions are therefore
computed directly by applying the known decoder circuit to the
corresponding encoded state and calculating the reduced state on $F$.

For $u=v$, this directly gives the required work factor in
\Cref{eq:pair-production}. For $u\neq v$,
\Cref{eq:cross-time-calculation} gives the corresponding cross-time
operator by multiplying $(\Omega_{u,u})_F$ by the known extra gates $(U_v\cdots U_{u+1})^\dagger$ (or their adjoint) and calculating the reduced state on $S_{\mathrm{work}}$. This also covers clock-clock
pairs: in that case $S_{\mathrm{work}}=\varnothing$, so the resulting
work factor is a scalar, which is multiplied by the known two-clock
factor in \Cref{eq:pair-production}. Hence all two-qubit marginals are
computable up to the final decodings considered next.

It remains only to account for the final output decodings. In
Recovery, perfect acceptance fixes the logical verifier output to
$|1\rangle\langle1|$, so its $E_{\mathrm{Sim}}$-encoded block has the
known state
$E_{\mathrm{Sim}}|1\rangle\langle1|E_{\mathrm{Sim}}^\dagger$.
In Verification, the output is not known, but it is first encoded under
$E_{\mathrm{2\text{-}priv}}$. Every surviving contribution to a
two-qubit history marginal involves at most two coordinates of this
codeword, whose joint state is $I/2$ or $I_4/4$ by \Cref{eq:code}.
Thus the required $(\Omega_{u,u})_F$ is again obtained from a known
state supported on constant number of qubits and the known decoding circuit.

We have therefore computed $(\Omega_{u,u})_F$ for every part of the
circuit. For $u=v$, this is the work factor appearing directly in
\Cref{eq:pair-production}. For $u\neq v$,
\Cref{eq:cross-time-calculation} combines it with the known intervening
gates to obtain $(\Omega_{u,v})_{S_{\mathrm{work}}}$. Multiplying this
by the explicitly known clock factor on $S_{\mathrm{clock}}$ and summing
over $u,v$ in \Cref{eq:pair-production} gives the complete two-qubit
marginal $D_S$.

There are $\binom n2$ choices of $S$ and at most $(T+1)^2$ terms for
each pair, so the complete pair table is computed in polynomial time
using only constant-sized matrices. In Recovery this defines
$D_{\mathrm R}(x)$ without knowing $w^\star$, which on  YES instances 
are exactly the pair marginals of the accepting history whose uniqueness
was proved above in this Lemma in \Cref{eq:eta-unique}.

\smallskip\noindent\textbf{Padding the history state to recover the final readout.}
Let $T_0$ be the length of the circuit before the final padding. Append
$11T_0$ identity gates and set $T=12T_0$ before assigning the clock
register. The readout identities established above concern the work
state at the final computational time $T_0$. However, after tracing out
the clock register, a work marginal of the history state is the uniform
average of the corresponding work marginals over all circuit times.
Hence the earlier states of the computation can shift the
$Z^{\otimes3}$ expectation away from its desired final value. The
identity padding ensures that the work register remains in its final
state for most of the history, so this time average is close to the
final work state.

Since all gates after time $T_0$ are identities,
$\Omega_{t,t}=\Omega_{T_0,T_0}$ for every $t\geq T_0$. Thus, for either
construction and any designated triple $B$,
\begin{equation}
(\rho_C(\omega))_B
=
\frac{1}{T+1}
\sum_{t=0}^{T}
(\Omega_{t,t})_B.
\end{equation}
Subtracting the final-time marginal and using
$\Omega_{t,t}=\Omega_{T_0,T_0}$ for $t\geq T_0$ gives
\begin{equation}
\left\|
(\rho_C(\omega))_B-(\Omega_{T_0,T_0})_B
\right\|_1
\leq
\frac{1}{T+1}
\sum_{t=0}^{T_0-1}
\left\|
(\Omega_{t,t})_B-(\Omega_{T_0,T_0})_B
\right\|_1.
\end{equation}
Each term in the sum is at most $2$, since both
$(\Omega_{t,t})_B$ and $(\Omega_{T_0,T_0})_B$ are density operators and
therefore have trace norm one. Using $T=12T_0$,
\begin{equation}
\left\|
(\rho_C(\omega))_B-(\Omega_{T_0,T_0})_B
\right\|_1
\leq
\frac{2T_0}{12T_0+1}
<
\frac{1}{6}.
\label{eq:averaging}
\end{equation}

To translate this trace-distance bound into a bound on the
$Z^{\otimes3}$ expectation, we use H\"older's inequality and have that
$|\operatorname{Tr}(AX)|\leq\|A\|_\infty\|X\|_1$. Since
$\|Z^{\otimes3}\|_\infty=1$,
\begin{equation}
\left|
\operatorname{Tr}\left(
Z^{\otimes3}(\rho_C(\omega))_B
\right)
-
\operatorname{Tr}\left(
Z^{\otimes3}(\Omega_{T_0,T_0})_B
\right)
\right|
<
\frac{1}{6}.
\end{equation}
Taking $B=B_i$ and using the previously established 
$\operatorname{Tr}\left(
Z^{\otimes3}(\Omega_{T_0,T_0})_{B_i}
\right)=(-1)^{w_i^\star}$
gives \Cref{eq:recover}, completing part (i).
\end{proof}

\begin{proof}[Proof of (ii)]
Fix $x$ and $w\in\{0,1\}^m$. Construct a separate circuit starting
entirely from physical zero  because now since we have recovered the witness, all the inputs are fixed. So we prepare $w$ using the specified
$X$ gates, encode the verifier inputs under $E_{\mathrm{Sim}}$, prepare
the fixed encoded resources, and run $V_x$ logically. Thus there is no
unspecified input to constrain, unlike the Recovery phase.

Let $\eta_{\mathrm{out}}$ be the final logical output qubit state.
Apply $E_{\mathrm{2\text{-}priv}}$ logically to the output and six
$E_{\mathrm{Sim}}$ encoded zeros. Decode, one at a time, the three
$E_{\mathrm{Sim}}$ blocks encoding coordinates $5$, $6$, and $7$ of
this codeword. Let $B_{\mathrm{out}}$ be their three data-output
qubits. At the end of the computation,
\begin{equation}
(\Omega_{T_0,T_0})_{B_{\mathrm{out}}}
=
\left(
E_{\mathrm{2\text{-}priv}}
\eta_{\mathrm{out}}
E_{\mathrm{2\text{-}priv}}^\dagger
\right)_{\{5,6,7\}}.
\end{equation}
Hence, by \Cref{eq:code,eq:output-expectation},
\begin{equation}
\operatorname{Tr}\left(
Z^{\otimes3}
(\Omega_{T_0,T_0})_{B_{\mathrm{out}}}
\right)
=
\operatorname{Tr}(Z\eta_{\mathrm{out}})
=
1-2p_x(w).
\label{eq:verification-final}
\end{equation}

Compute $D_{\mathrm V}(x,w)$ using the pair-data construction from
part (i). Unlike Recovery, there are no input-code constraints to
enforce because the entire input of this circuit is fixed. During the
final decoding of $B_{\mathrm{out}}$, every two-qubit marginal
calculation involves at most two coordinates of the
$E_{\mathrm{2\text{-}priv}}$ codeword. By \Cref{eq:code}, their
one- and two-coordinate states are $I/2$ and $I_4/4$, respectively,
independently of $\eta_{\mathrm{out}}$. Hence these marginals are
computable without knowing $p_x(w)$. Therefore $D_{\mathrm V}(x,w)$
is the actual pair table of this fixed-input history for every $x,w$.

By the history-space argument following \Cref{eq:history}, any global
completion $\xi$ of $D_{\mathrm V}(x,w)$ has the form
\begin{equation}
\xi=J_C\omega'J_C^\dagger.
\end{equation}
Here $r=0$, since the verification circuit starts entirely from fixed
zero qubits. Thus $\omega'$ acts on a one-dimensional input space, so
due to trace normalization $\omega'=1$. Consequently,
\begin{equation}
\xi=J_CJ_C^\dagger=\rho_{\mathrm V}(x,w).
\end{equation}
Hence $D_{\mathrm V}(x,w)$ has a unique global completion for every
$x,w$, independently of whether $V_x$ accepts $w$.

Finally, apply the same identity padding as in part (i), with $T_0$
now denoting the length of this verification circuit before padding.
By \Cref{eq:averaging,eq:verification-final},
\begin{equation}
\left|
\operatorname{Tr}\left(
Z^{\otimes3}
(\rho_{\mathrm V}(x,w))_{B_{\mathrm{out}}}
\right)
-
\left(1-2p_x(w)\right)
\right|
<
\frac16.
\end{equation}
This is \Cref{eq:verify}, completing part (ii).
\end{proof}

\printbibliography
\appendix
\section{Supporting proofs for real classical-shadow validity}\label{app:real}

\subsection{\texorpdfstring{Proof of \Cref{lem:1d-xz-lh}}{Proof of Lemma}}\label{app:1d-xz-lh}\label{sapp:onedim}
\begin{proof}
Let $L\in\QMA$, and let $x$ be an instance of $L$. By $\QMA=\QMA_{\mathbb R}$ \cite{McKAGUE13}, we may assume that the verification circuit for $x$ consists of real gates. Applying the construction of \cite{HNN13} to this verifier gives a nearest-neighbor Hamiltonian
    \[
    H_x^{\mathsf{HNN}}=J_{\mathsf{in}}H_{\mathsf{in}}+J_{\mathsf{prop}}H_{\mathsf{prop}}+J_{\mathsf{pen}}H_{\mathsf{pen}}+H_{\mathsf{out}}
    \]
    on a line of $N=\poly(|x|)$ eight-dimensional particles, with an inverse-polynomial gap. Let $a$ and $b$ be its YES and NO instance thresholds, respectively, and set $\delta:=b-a\ge 1/\poly(|x|)$.

    After padding every one-site contribution by the identity on a neighboring particle and collecting all terms assigned to the same edge, write
    \[
    H_x^{\mathsf{HNN}}=\sum_{i=1}^{N-1}h_i.
    \]
    Each $h_i$ acts on two eight-dimensional particles. Fix the natural HNN product basis
    \[
    \Sigma=\{M_1,M_2,M_3,M_4,Q^{(0)},Q^{(1)},Q^{\prime(0)},Q^{\prime(1)}\}.
    \]
    Thus $d:=|\Sigma|=8$.
    The initialization, output, and penalty terms are diagonal in this basis, while the propagation terms involve the real verifier gates and their transposes; hence every local term of $H_x^{\mathsf{HNN}}$ is represented by a real symmetric matrix.  We will encode the entire HNN Hilbert space into a larger qubit Hilbert space so that every encoded local term belongs to $\operatorname{span}_{\mathbb R}\{I,X,Z\}^{\otimes O(1)}$.\\
    \\
    For an HNN computational basis state
    \[
    \ket{x_1,x_2,\dots,x_N},\quad x_i\in\Sigma,
    \]
    introduce a boundary symbol $\#\notin\Sigma$, set $\overline\Sigma:=\Sigma\cup\{\#\}$, and define
    \begin{equation}\label{eq:overlap-encoding}
         V\ket{x_1,x_2,\dots ,x_N}=\ket{(\#,x_1)}_{E_0}\ket{(x_1,x_2)}_{E_1}\cdots \ket{(x_{N-1},x_N)}_{E_{N-1}}\ket{(x_N,\#)}_{E_N}.
    \end{equation}
    Here each $E_j$ is a new register whose basis labels are ordered pairs of HNN symbols. If $\lambda$ is a pair label, then $\ket{\lambda}$ denotes the computational-basis state having a single $1$ in the position corresponding to $\lambda$. Define the code space $\mathcal C:=\operatorname{im}V$. For a label $\lambda$ in register $E_j$, set
\[
    n_{j,\lambda}:=\ketbra{1}{1}_{j,\lambda}
    =\frac{I-Z_{j,\lambda}}2,
    \quad
    N_j:=\sum_{\lambda\in\overline\Sigma^2}n_{j,\lambda}.
\]
The allowed label sets are
\[
\Lambda_0=\{(\#,a):a\in\Sigma\},\qquad
\Lambda_j=\Sigma\times\Sigma\ (1\le j\le N-1),\qquad
\Lambda_N=\{(a,\#):a\in\Sigma\}.
\]
Now define
\begin{equation}
H_{\mathrm{code}} := \sum_{j=0}^{N}(N_j-I)^2 + \sum_{j=0}^{N}\sum_{\lambda\notin\Lambda_j}n_{j,\lambda}+ \sum_{j=0}^{N-1}
\sum_{\substack{\lambda=(a,b)\in\Lambda_j,\\ \mu=(c,e)\in\Lambda_{j+1}\\
b\neq c}}n_{j,\lambda}n_{j+1,\mu}.
\label{eq:Hcode}
\end{equation}
The first term enforces encoding weight one, the second enforces the
correct boundary/interior alphabets, and the third enforces agreement
of neighboring overlaps. Every term $\in\operatorname{span}_{\mathbb R}\{I,Z\}^{\otimes O(1)}$.

Since $H_{\mathrm{code}}$ is diagonal and positive semidefinite, a computational-basis state has zero energy exactly when every register contains one allowed pair label and all adjacent pair labels agree on their overlap. Such strings are precisely the images under $V$. Hence
\begin{equation}
\label{eq:code-kernel}
    \ker H_{\mathrm{code}}=\operatorname{im}V,\quad H_{\mathrm{code}}\big|_{(\operatorname{im}V)^\perp}\succeq I.
\end{equation}
The single-excitation representation above allows us to change labels using only $X$ operators. If $\lambda\neq\mu$, then
\[
X_{j,\lambda}X_{j,\mu}\ket{\lambda}=\ket{\mu},
\qquad
X_{j,\lambda}X_{j,\mu}\ket{\mu}=\ket{\lambda}.
\]
Thus, after compression to the single-excitation subspace, $X_{j,\lambda}X_{j,\mu}$ acts exactly as the transition $\ket{\lambda}\bra{\mu}+\ket{\mu}\bra{\lambda}$.
Define
\begin{equation}
\label{eq:S-def}
S_j(\lambda,\mu):=
\begin{cases}
 n_{j,\lambda},&\lambda=\mu,\\[1mm]
 X_{j,\lambda}X_{j,\mu},&\lambda\neq\mu.
\end{cases}
\end{equation}
Fix an original HNN edge $i,i+1$ and let
\[
    \alpha=(a,b),\quad \beta=(a',b')\in\Sigma^2,
\]
be two local basis configurations on that  edge. If the neighboring values are $\ell$ and $r$, then the three relevant encoded labels for local value $\alpha$ are
\[
(\ell,a),\quad (a,b),\quad(b,r),
\]
whereas after changing $\alpha$ to $\beta$, they should become
\[
(\ell,a'),\quad (a',b'),\quad (b',r).
\]
Define
\begin{align}
D_i^{\alpha,\beta}
:=\sum_{\ell,r}
S_{i-1}\bigl((\ell,a),(\ell,a')\bigr)\cdot
S_i\bigl((a,b),(a',b')\bigr)
\nonumber\cdot
S_{i+1}\bigl((b,r),(b',r)\bigr),
\label{eq:D-alpha-beta}
\end{align}
where $\ell=\#$ at the left boundary and $r=\#$ at the right boundary; otherwise, the contexts range over $\Sigma$. Then
\begin{equation}
\label{eq:compression-matrix-unit}
V^\dagger D_i^{\alpha,\beta}V
=
\begin{cases}
 \ketbra{\alpha}{\alpha}_{i,i+1},&\alpha=\beta,\\[1mm]
 \ketbra{\alpha}{\beta}_{i,i+1}
 +\ketbra{\beta}{\alpha}_{i,i+1},&\alpha\neq\beta.
\end{cases}
\end{equation}
Since every $h_i$ is real symmetric, write
\[
h_i=
\sum_{\alpha}(h_i)_{\alpha,\alpha}\ketbra{\alpha}{\alpha}
+
\sum_{\alpha<\beta}(h_i)_{\alpha,\beta}
\bigl(\ketbra{\alpha}{\beta}+\ketbra{\beta}{\alpha}\bigr)
\]
and define
\begin{equation}
\label{eq:hhat-def}
\widehat h_i
:=
\sum_{\alpha}(h_i)_{\alpha,\alpha}D_i^{\alpha,\alpha}
+
\sum_{\alpha<\beta}(h_i)_{\alpha,\beta}D_i^{\alpha,\beta}.
\end{equation}
Then
\begin{equation}
\label{eq:local-compression}
    V^\dagger\widehat h_iV=h_i.
\end{equation}
Every factor in \eqref{eq:hhat-def} is either
$X\otimes X$ or $(I-Z)/2$, so $\widehat h_i\in\operatorname{span}_{\mathbb R}\{I,X,Z\}^{\otimes O(1)}$. Moreover, $\widehat h_i$ acts only on the three consecutive registers
$E_{i-1},E_i,E_{i+1}$. Let
\[
    \widehat H:=\sum_{i=1}^{N-1}\widehat h_i.
\]
By linearity, $V^\dagger\widehat HV=H_x^{\mathsf{HNN}}$, so the compression of $\widehat{H}$ to the valid code space reproduces the original HNN Hamiltonian.\\
\\
Now choose an efficiently computable polynomial bound
$B\ge\norm{\widehat H}_\infty$, and define
\[
    H':=\widehat H+JH_{\mathrm{code}},
    \qquad
    J>2B+\frac{4B^2}{\delta}.
\]
To compare the ground energies of $H'$ and $H_x^{\mathsf{HNN}}$, we use the following lemma.
\begin{lemma}[Projection lemma \cite{KKR05}]\label{lem:projection}
Let $H = H_1 + H_2$ be the sum of two Hamiltonians operating on some Hilbert space $\mathcal{H} = S + S^\perp$. The Hamiltonian \(H_2\) is such that $S$ is a zero eigenspace and the eigenvectors in $S^\perp$ have eigenvalue at least $J > 2\|H_1\|_\infty$. Then,
\[
\lambda\left(H_1\big|_S\right)-
\frac{\|H_1\|_\infty^2}
     {J - 2\|H_1\|_\infty}
\leq
\lambda(H)
\leq
\lambda\left(H_1\big|_S\right).
\]
\end{lemma}
By \eqref{eq:code-kernel}, $JH_{\mathrm{code}}$ has
ground space $\mathcal C=\operatorname{im}V$ and gap
at least $J$. Applying \Cref{lem:projection} with
\[
H_1=\widehat H,\quad H_2=JH_{\mathrm{code}},\quad S=\mathcal C=\operatorname{im}V,
\]
we obtain
\begin{equation}
\label{eq:projection-bound}
\lambda_{\min}(\widehat H|_{\mathcal C})-\frac{B^2}{J-2B}
\le \lambda_{\min}(H')
\le \lambda_{\min}(\widehat H|_{\mathcal C})
=\lambda_{\min}(H_x^{\mathsf{HNN}}).
\end{equation}
Our choice of $J$ makes the error in
\eqref{eq:projection-bound} at most $\delta/4$. Therefore
\[
\lambda_{\min}(H_x^{\mathsf{HNN}})\le a
\implies
\lambda_{\min}(H')\le a,
\]
and
\[
\lambda_{\min}(H_x^{\mathsf{HNN}})\ge b
\implies
\lambda_{\min}(H')\ge b-\frac{\delta}{4}.
\]
Thus the new promise gap is at least $3\delta/4$.

The code terms act on at most two consecutive $E$-registers and the
terms $\widehat h_i$ act on three consecutive $E$-registers.  Block the
registers in pairs,
\[
    B_0=(E_0,E_1),\quad B_1=(E_2,E_3),\quad\ldots.
\]
Every interval of three consecutive $E$-registers intersects at most two
neighboring blocks. Hence, after blocking, every term is supported on
one block or two neighboring blocks.  Each block contains
\[
    p=2(d+1)^2
\]
qubits; for $d=8$, one may take $p=162$. Blocking changes neither the
spectrum nor the $XZ$-only Pauli expansion. Finally, one-block terms
may be padded by the identity on an adjacent block, so all terms may be
regarded as edge-supported.  This proves the lemma.
\end{proof}

\subsection{\texorpdfstring{Proof of \Cref{lem:xz-certified-cldm}}{Proof of Lemma}}\label{sapp:xz-certified-cldm}
\begin{proof}

Start with an amplified arbitrary $\QMA$ verifier $V_x$ so that its completeness and soundness errors are negligible. Fix a constant \cite{BG22}- simulator parameter $s$ large enough for all constant-size regions used below. Apply the BG construction and get a simulatable verifier $V_x^{(s)}$. This, roughly, means that for every sufficiently small set $Y$ of work qubits, we can efficiently compute a density matrix $\widetilde{\rho}_Y$ such that, for a good witness, $\widetilde{\rho}_Y\approx \Tr_{\overline{Y}}(\rho_{\text{true}})$.\\
\\
We realify $V_x^{(s)}$ as in \cite{McKAGUE13}, obtaining
$V_x^{s,\mathbb R}$ with one additional rebit $r$ included
in the witness register. For a matrix $B$, define
\[
\mathcal R(B):=
\begin{pmatrix}
\operatorname{Re}B & -\operatorname{Im}B\\
\operatorname{Im}B & \operatorname{Re}B
\end{pmatrix},
\qquad
\Gamma(\rho):=\frac12\mathcal R(\rho),
\]
where $r$ indexes the two block rows and columns.
Each gate $U$ is replaced by the real orthogonal gate
$\mathcal R(U)$, and an honest BG witness $\omega$ is
replaced by the real density operator $\Gamma(\omega)$.
Indeed,
\[
\Gamma(U\rho U^\dagger)
=
\mathcal R(U)\Gamma(\rho)\mathcal R(U)^{\mathsf T}.
\]
Acceptance probabilities are preserved, and soundness
is unchanged because the realified acceptance operator
has the same largest eigenvalue as the original one.

This density-matrix convention preserves local
simulatability. For any subset $Y$ of the original work
register $D$, writing $\rho_Y:=\Tr_{D\setminus Y}(\rho)$,
we have
\[
\Tr_{D\setminus Y}\Gamma(\rho)=\Gamma(\rho_Y),
\qquad
\Tr_r\Gamma(\rho_Y)=\operatorname{Re}\rho_Y.
\]
Thus a BG simulator output $\widetilde\rho_Y$ gives
$\Gamma(\widetilde\rho_Y)$ for the marginal on
$\{r\}\cup Y$, and $\operatorname{Re}\widetilde\rho_Y$
for the marginal on $Y$ alone. These operations do not
increase trace-norm error, since
\[
\|\Gamma(\rho)-\Gamma(\sigma)\|_1=\|\rho-\sigma\|_1
\]
and partial trace is contractive on Hermitian operators.
For an intermediate step in a constant-size real-gate
decomposition, we include the gate's support in the
simulated region, apply the known partial circuit, and
trace out the unused qubits. The fixed parameter $s$
is chosen large enough for these enlarged regions.
The history-state density operator used below is
the one induced by the possibly mixed honest witness
$\Gamma(\omega)$; the history-state construction is
extended to mixed witnesses by linearity.\\
\\
Next we insert $\text{SWAP}$ and identity gates as needed to make the computation nearest-neighbor, and apply the HNN construction. Because the circuit is real, this produces a nearest-neighbor Hamiltonian
\[
H_x^{\mathsf{HNN}}=\sum_i h_i
\]
whose local terms are real symmetric. Appendix B of \cite{KRMEG25} says that if the computation going into HNN is BG-simulatable, then the relevant HNN history-state marginals are also efficiently simulatable. Although Appendix B states this for one- and two-site HNN regions, inspection of the proof shows that, for any fixed constant $q$, the same construction simulates intervals of at most $q$ consecutive HNN sites after increasing the constant simulator parameter $s$. Indeed, the corresponding marker intervals and work-qubit light cones remain of constant size. We apply this observation with $q=5$.\\
\\
Thus, for every interval $R$ of at most five consecutive HNN sites,
we can compute in polynomial time a density matrix $\tau_R$ such that,
if $x\in L_{\mathrm{yes}}$,
\begin{equation}
\label{eq:hww-fixed-window}
\left\|
\tau_R-
\Tr_{\overline R}\left(\Phi_x^{\mathsf{HNN}}\right)
\right\|_1
\leq \eta_x,
\end{equation}
where $\eta_x=\operatorname{negl}(|x|)$ and
$\Phi_x^{\mathsf{HNN}}$ is the HNN history-state density operator
corresponding to a good witness. Equally importantly, these simulated targets satisfy the exact HNN
local-energy identities of Appendix~B: all initialization, output,
penalty, and boundary contributions vanish, while every grouped
propagation contribution vanishes as well.\\
\\
Now apply the overlap encoding $V$ of
Lemma~\ref{lem:1d-xz-lh}. After blocking, we obtain
\[
H_x^{XZ}
=
\widehat H_x^{\mathsf{HNN}}
+
JH_{\mathrm{code}}
=
\sum_{\ell=1}^{M}g_\ell,
\]
where every $g_\ell$ is $XZ$-only and is supported on one block or two
neighboring blocks. Padding one-block terms by the identity on a
neighboring block, we regard every $g_\ell$ as edge-supported.

Let $e$ be an edge of the blocked chain. It contains at most four
consecutive overlap registers and therefore depends on an interval
$R_e$ of at most five consecutive HNN sites. The explicit overlap
encoding induces a fixed CPTP map
\[
\mathcal E_e:\mathcal D(R_e)\longrightarrow\mathcal D(e).
\]
Explicitly, $\mathcal E_e$ dephases the boundary HNN symbols that also
occur in overlap registers outside $e$, and then writes the
overlapping pair labels contained in $e$. For every global HNN state
$\omega$,
\begin{equation}
\label{eq:hww-local-encoding}
\mathcal E_e
\left(
\Tr_{\overline{R_e}}(\omega)
\right)
=
\Tr_{\overline e}
\left(
V\omega V^\dagger
\right).
\end{equation}
Define the edge target
\[
\sigma_e:=\mathcal E_e(\tau_{R_e}).
\]
Since $R_e$ has constant size and $\mathcal E_e$ is fixed, every
$\sigma_e$ is computable in polynomial time.\\
\\
We next show that the HNN energy expectations transfer to the encoded terms. For every original HNN
term $h_i$, choose a blocked edge $e(i)$ containing the support of its
encoded term $\widehat h_i$. By the definition of $\mathcal E_{e(i)}$
and Equation~\eqref{eq:local-compression},
\begin{equation}
\label{eq:hww-energy-transfer}
\Tr\left(
\widehat h_i\sigma_{e(i)}
\right)
=
\Tr\left(
h_i\,
\Tr_{R_{e(i)}\setminus\{i,i+1\}}
\left(
\tau_{R_{e(i)}}
\right)
\right).
\end{equation}
Moreover, every $\sigma_e$ is supported on locally
single-excitation, allowed, overlap-consistent labels, as  in \Cref{lem:1d-xz-lh}, so every code-penalty
term has also zero expectation. Consequently,
\begin{equation}
\label{eq:hww-zero-certificate}
\sum_{\ell=1}^{M}
\Tr\left(
g_\ell\sigma_{e_\ell}
\right)
=
0
\qquad
\text{for every input }x,
\end{equation}
where $e_\ell$ denotes the edge supporting $g_\ell$.\\
\\
By the proof of Lemma~\ref{lem:1d-xz-lh}, if $x\in L_{\mathrm{no}}$, then
\[
    \lambda_{\min}(H_x^{XZ})
    \ge b-\frac{\delta}{4}
    \ge \frac{3\delta}{4}
    =:\Delta_x,
\]
where $\delta=b-a$ and the HNN thresholds are chosen with $0\le a<b$.
Thus $\Delta_x$ is efficiently computable and satisfies
$\Delta_x\ge 1/\poly(|x|)$.\\
\\
Choose an efficiently computable $W_x$ satisfying
\[
\sum_{\ell=1}^{M}\|g_\ell\|_\infty
\leq W_x
\leq\poly(|x|),
\]
and choose the BG simulation accuracy so that $\eta_x\leq\frac{\Delta_x}{16W_x}$.
Set
\[
\alpha_{\mathrm{CLDM}}:=\eta_x,
\qquad
\beta_{\mathrm{CLDM}}
:=
\frac{\Delta_x}{4W_x}.
\]
\textbf{Completeness:} Suppose $x\in L_{\mathrm{yes}}$, and set
\[
\rho^*
:=
V\Phi_x^{\mathsf{HNN}}V^\dagger.
\]
By \eqref{eq:hww-fixed-window},
\eqref{eq:hww-local-encoding}, and contractivity of trace distance
under CPTP maps,
\begin{align*}
\left\|
\sigma_e-\Tr_{\overline e}(\rho^*)
\right\|_1
&=
\left\|
\mathcal E_e(\tau_{R_e})
-
\mathcal E_e
\left(
\Tr_{\overline{R_e}}
\left(
\Phi_x^{\mathsf{HNN}}
\right)
\right)
\right\|_1\\
&\leq
\left\|
\tau_{R_e}
-
\Tr_{\overline{R_e}}
\left(
\Phi_x^{\mathsf{HNN}}
\right)
\right\|_1\\
&\leq\eta_x
\quad
\forall e.
\end{align*}
Hence the constructed CLDM instance is YES with $\alpha_{\text{CLDM}}=\eta_x$.\\
\\
\textbf{Soundness:}
Suppose $x\in L_{\mathrm{no}}$ and, toward a contradiction, suppose
that there exists a global state $\rho$ satisfying
\[
\left\|
\Tr_{\overline e}(\rho)-\sigma_e
\right\|_1
<
\beta_{\mathrm{CLDM}}
\qquad
\forall e.
\]
Using \eqref{eq:hww-zero-certificate}, we obtain
\begin{align*}
\Tr\left(H_x^{XZ}\rho\right)
&=
\sum_{\ell=1}^{M}
\Tr\left(
g_\ell
\Tr_{\overline{e_\ell}}(\rho)
\right)\\
&=
\sum_{\ell=1}^{M}
\Tr\left(
g_\ell
\left[
\Tr_{\overline{e_\ell}}(\rho)
-\sigma_{e_\ell}
\right]
\right)
+
\underbrace{
\sum_{\ell=1}^{M}
\Tr(g_\ell\sigma_{e_\ell})
}_{=\,0}\\
&\leq
\sum_{\ell=1}^{M}
\|g_\ell\|_\infty
\left\|
\Tr_{\overline{e_\ell}}(\rho)
-\sigma_{e_\ell}
\right\|_1\\
&<
W_x\beta_{\mathrm{CLDM}}
=
\frac{\Delta_x}{4}
<
\Delta_x,
\end{align*}
contradicting
\[
\Tr(H_x^{XZ}\rho)
\geq
\lambda_{\min}(H_x^{XZ})
\geq
\Delta_x.
\]
Thus every global state violates at least one edge target by trace
distance at least $\beta_{\mathrm{CLDM}}$.
Finally,
\[
\beta_{\mathrm{CLDM}}-\alpha_{\mathrm{CLDM}}
\geq
\frac{\Delta_x}{4W_x}
-
\frac{\Delta_x}{16W_x}
=
\frac{3\Delta_x}{16W_x}
\geq
\frac{1}{\poly(|x|)}.
\]
Therefore the construction is a polynomial-time reduction to
$XZ$-certified $1D$-CLDM with an inverse-polynomial promise gap.
\end{proof}

\section{Supporting proofs for CV classical-shadow validity}\label{app:CV}
\subsection{\texorpdfstring{Proof of \Cref{lem:QE-close}}{Proof of Lemma}}\label{sapp:QE-close}
\begin{proof}
Set $H:=I-P$. Since $P+H=I$, we can expand
\[
Q=(P+H)Q(P+H)=PQP+PQH+HQP+HQH.
\]
Since $E=PQP$, this gives
\[
Q-E=PQH+HQP+HQH.
\]
Therefore, by the triangle inequality
\[
\left|\Tr((Q-E)\sigma)\right|
\le
|\Tr(PQH\sigma)|
+
|\Tr(HQP\sigma)|
+
|\Tr(HQH\sigma)|.
\]
We now bound these three terms separately.\\
\\
Start with the first off-diagonal term. By the cyclicity of trace,
\[
|\Tr(PQH\sigma)|=|\Tr(\sigma^{1/2}PQH\sigma^{1/2})|.
\]
Using the Cauchy-Schwarz inequality,
\[
|\Tr(\sigma^{1/2}PQH\sigma^{1/2})|
\le \|\sigma^{1/2}PQ\|_2\,\| H\sigma^{1/2}\|_2.
\]
For the first factor, using the standard Schatten norm inequality $\|AB\|_2\le \|A\|_2\|B\|_\infty$, we get
\[
\|\sigma^{1/2}PQ\|_2
\le
\|\sigma^{1/2}\|_2\|PQ\|_\infty
=
\sqrt{\Tr(\sigma)}\,\|PQ\|_\infty
=
\|PQ\|_\infty.
\]
For the second factor, the Hilbert-Schmidt norm gives
\[
\|H\sigma^{1/2}\|_2^2
=\Tr\left((H\sigma^{1/2})^\dagger(H\sigma^{1/2})\right)
=\Tr(\sigma^{1/2}H\sigma^{1/2})
=\Tr(H\sigma).
\]
Therefore,
\[
|\Tr(PQH\sigma)|
\le
\|PQ\|_\infty\sqrt{\Tr(H\sigma)}.
\]
It remains to bound $\|PQ\|_\infty$. Since 
\[ 
P=\bigotimes_{j\in\Lambda}P_1^{(j)},\quad
Q=\bigotimes_{j\in\Lambda}A_j
\]
we have
\[
PQ=\bigotimes_{j\in\Lambda}P_1^{(j)}A_j.
\]
Hence
\[
\|PQ\|_\infty
=
\prod_{j\in\Lambda}\|P_1^{(j)}A_j\|_\infty.
\]
For the one-mode factors,
\[
\|P_1I\|_\infty=1,\quad
\|P_1a\|_\infty=\sqrt 2,
\quad
\|P_1a^\dagger\|_\infty=1,
\]
and so
\[
\|P_1(a+a^\dagger)\|_\infty
<3,\quad
\|P_1(-i(a-a^\dagger))\|_\infty<3.
\]
Finally,
\[
\|P_1(I-2\widehat N)\|_\infty=1.
\]
Thus $\|P_1A_j\|_\infty\le3
\;\;\forall j\in\Lambda$.
Since $|\Lambda|\le k$, $\|PQ\|_\infty\le 3^k$.
Therefore,
\[
|\Tr(PQH\sigma)|
\le 3^k\sqrt{\Tr(H\sigma)}
\le 3^k\sqrt{\frac{kB}{3^\ell}}.
\]
Where the last inequality comes from \Cref{cor:localmomentbound}.\\
\\
The second off-diagonal term is bounded in the same way. Namely,
\[
|\Tr(HQP\sigma)|=|\Tr(\sigma^{1/2}HQP\sigma^{1/2})|\le\|\sigma^{1/2}H\|_2\,\|QP\sigma^{1/2}\|_2.
\]
We have
\[
\|\sigma^{1/2}H\|_2=\sqrt{\Tr(H\sigma)}
\]
and
\[
\|QP\sigma^{1/2}\|_2\le\|QP\|_\infty\|\sigma^{1/2}\|_2=\|QP\|_\infty.
\]
Since $Q$ is Hermitian,
\[
\|QP\|_\infty
=\|(PQ)^\dagger\|_\infty=\|PQ\|_\infty\le3^k.
\]
Thus
\[
|\Tr(HQP\sigma)|\le3^k\sqrt{\Tr(H\sigma)}\le3^k\sqrt{\frac{kB}{3^\ell}}.
\]
It remains to bound the diagonal term $|\Tr(HQH\sigma)|$. We first prove a weighted one-mode bound.\\
\\
First notice that for every $j\in\Lambda$, $\|W_j^{-1/2}A_j W_j^{-1/2}\|_\infty\le 2$.
Indeed,
\[
 W_j^{-1/2}a_jW_j^{-1/2}\ket{n}
=
\frac{1}{\sqrt{n+1}}\ket{n-1},\quad
W_j^{-1/2}a_j^\dagger W_j^{-1/2}\ket{n}
=\frac{1}{\sqrt{n+2}}\ket{n+1}.
\]
Therefore,
\[
\|W_j^{-1/2}IW_j^{-1/2}\|_\infty= 1,\quad\|W_j^{-1/2}a_jW_j^{-1/2}\|_\infty=\frac {1}{\sqrt{2}},
\quad
\|W_j^{-1/2}a_j^\dagger W_j^{-1/2}\|_\infty=\frac {1}{\sqrt{2}}.
\]
By the triangle inequality,
\[
\|W_j^{-1/2}(a_j+a_j^\dagger)W_j^{-1/2}\|_\infty\le \sqrt2,\quad
\|W_j^{-1/2}(-i(a_j-a_j^\dagger))W_j^{-1/2}\|_\infty\le \sqrt2.
\]
Finally,
\[
W_j^{-1/2}(I-2\widehat N_j)W_j^{-1/2}\ket{n}
=
\frac{1-2n}{1+n}\ket{n},
\]
Hence
\[
\|W_j^{-1/2}(I-2\widehat N_j)W_j^{-1/2}\|_\infty\le 2.
\]
This proves the weighted one-mode bound.\\
\\
Since the different factors act on different modes,
\[
W_\Lambda^{-1/2}QW_\Lambda^{-1/2}
=
\bigotimes_{j\in\Lambda}
\left(W_j^{-1/2}A_jW_j^{-1/2}\right).
\]
Therefore, $\|W_\Lambda^{-1/2}QW_\Lambda^{-1/2}\|_\infty\le 2^k$.\\
\\
Now since, $W_\Lambda^{-\frac12}QW_\Lambda^{-\frac12}$ is hermitian we get
\begin{align*}
    \begin{aligned}
      \|W_\Lambda^{-\frac12}QW_\Lambda^{-\frac12}\|_\infty\le 2^k\iff-2^kI\preceq W_\Lambda^{-\frac12}QW_\Lambda^{-\frac12}\preceq 2^kI  
    \end{aligned}
\end{align*}
Conjugating by $\sigma^{1/2}HW_\Lambda^{\frac12}$ gives 
\[
-2^k\sigma^{\frac12}HW_{\Lambda}H\sigma^{\frac12}\preceq \sigma^{\frac12}HQH\sigma^{\frac12}\preceq 2^k\sigma^{\frac12}HW_{\Lambda}H\sigma^{\frac12}
\]
Finally by taking traces we get
\[
|\Tr(HQH\sigma)|\le 2^k\Tr(HW_{\Lambda}H\sigma)
\]
Now we bound $\Tr(HW_\Lambda H\sigma)$. Take a Fock basis vector $\ket{n}=\ket{n_j}_{j\in\Lambda}$ in the support of $H$. This means that at least one mode satisfies $n_j\ge2$, and so for at least one $j$, $x_j:=1+n_j\ge3$. Now since $|\Lambda|\le k$ and for $\ell\ge k$ we get
\[
\prod_{j\in\Lambda}x_j\le (\max_{j\in\Lambda}x_j)^k\le 3^{k-\ell}(\max_{j\in\Lambda}x_j)^\ell\le 3^{k-\ell}\sum_{j\in\Lambda}x_j^\ell.
\]
Now because $W_\Lambda$ and $W_j^\ell$ are diagonal in the same basis and their eigenvalues obey the equation above we have:
\[
HW_\Lambda H\preceq 3^{k-\ell}\sum_{j\in\Lambda}HW_j^{\ell}H\preceq 3^{k-\ell}\sum_{j\in\Lambda}W_j^{\ell}.
\]
The second inequality follows because $0\preceq H\preceq I$ and $H$ commutes with each $W_j^\ell$.\\
\\
Taking the trace against $\sigma\succeq 0$ gives
\[
\Tr(HW_\Lambda H\sigma)
\le
3^{k-\ell}\sum_{j\in\Lambda}\Tr(W_j^\ell\sigma).
\]
By the moment bound $\Tr(W_j^\ell\sigma)\le B$, we obtain
\[
\Tr(HW_\Lambda H\sigma)
\le
kB\,3^{k-\ell}.
\]
Therefore, using the weighted estimate
\[
|\Tr(HQH\sigma)|
\le
2^k\Tr(HW_\Lambda H\sigma),
\]
we get
\[
|\Tr(HQH\sigma)|
\le
2^k kB\,3^{k-\ell}.
\]
Putting the three estimates together,
\begin{align*}
\left|\Tr((Q-E)\sigma)\right|
&\le |\Tr(PQH\sigma)|+|\Tr(HQP\sigma)|+|\Tr(HQH\sigma)|\\
&\le 3^k\sqrt{\frac{kB}{3^\ell}}+3^k\sqrt{\frac{kB}{3^\ell}}+2^k kB\,3^{k-\ell}\\
&=2\cdot 3^k\sqrt{\frac{kB}{3^\ell}}+2^k kB\,3^{k-\ell}.
\end{align*}  
\end{proof}

\end{document}

%% file: commands.tex
\newcommand*{\Tr}{\operatorname{Tr}}

\DeclarePairedDelimiterX\ketbra[2]{\lvert}{\rvert}{#1 \delimsize\rangle\delimsize\langle #2}

\newcommand{\polylog}{\operatorname{polylog}}

\newcommand{\poly}{\textup{poly}}
\renewcommand{\exp}{\textup{exp}}

\mathchardef\mhyphen="2D % Define a "math hyphen"

\DeclarePairedDelimiter\norm{\lVert}{\rVert}

\newcommand{\trace}{\Tr}